%% file: main.tex
\documentclass[journal]{IEEEtran}

\usepackage{amsmath, amsthm, amssymb, amsfonts}
\usepackage{bm}
\usepackage{graphicx}
\usepackage{xcolor}

\usepackage{tabularx}
\usepackage{booktabs}
\usepackage{multirow}
\usepackage{array}

\usepackage{tikz}
\usetikzlibrary{
	decorations.pathreplacing,
	shapes.arrows,
	shadows,
	calc,
	shapes,
	arrows.meta,
	positioning,
	backgrounds,
	patterns
}
\usepackage{pgfplots}
\pgfplotsset{compat=1.17}
\usepgfplotslibrary{groupplots}

\usepackage{placeins} 
\usepackage{flafter} 
\usepackage[font=small,skip=0pt]{caption} 

\usepackage{cite}
\usepackage{hyperref}

\definecolor{ieee_blue}{RGB}{0, 76, 153}
\definecolor{ieee_orange}{RGB}{255, 127, 0}
\definecolor{ieee_red}{RGB}{180, 20, 20}
\definecolor{tech_red}{RGB}{220, 50, 47}
\definecolor{tech_blue}{RGB}{38, 139, 210}
\definecolor{tech_gray}{RGB}{101, 123, 131}
\definecolor{myblue}{RGB}{0, 114, 189}
\definecolor{myblack}{RGB}{30, 30, 30}

\hypersetup{hidelinks}

\theoremstyle{definition}
\newtheorem{theorem}{Theorem}
\newtheorem{corollary}{Corollary}
\newtheorem{definition}{Definition}

\newtheoremstyle{remarkbf}
{3pt}{3pt}{\normalfont}{}{\bfseries}{.}{0.5em}{}
\theoremstyle{remarkbf}
\newtheorem{remark}{Remark}
\newcounter{protocolalg}

\begin{document}
	
	\title{BioZKFHE: Scalable Encrypted Biometric Identification via Verifiable Homomorphic Similarity Evaluation}
	
	 \author{Rundong~Xin,~Taotao~Wang,~\IEEEmembership{Member,~IEEE},~Xiaoxiao~Wu,~\IEEEmembership{Senior~Member,~IEEE},~Weizhi~Meng,~\IEEEmembership{Senior~Member,~IEEE},~Shengli~Zhang,~\IEEEmembership{Senior~Member,~IEEE},~and~Shui~Yu,~\IEEEmembership{Fellow,~IEEE}
	 \thanks{R.~Xin,~T.~Wang,~X.~Wu,~and~S.~Zhang are with the College of Electronics and Information Engineering, Shenzhen University, Shenzhen 518060, China (e-mail: rundongxin1@gmail.com, ttwang@szu.edu.cn, xxwu.eesissi@szu.edu.cn, zsl@szu.edu.cn). T.~Wang is the corresponding author.}
	
	 \thanks{W.~Meng is with the School of Computing and Communications, Lancaster University, Lancaster LA1 4WA, U.K. (e-mail: w.meng3@lancaster.ac.uk).}
	
	 \thanks{Shui~Yu is with the School of Computer Science, University of Technology Sydney, Sydney, NSW 2007, Australia (e-mail: shui.yu@uts.edu.au).}

	 \thanks{Accepted for publication in \emph{IEEE Transactions on Dependable and Secure Computing}. Digital Object Identifier: \href{https://doi.org/10.1109/TDSC.2026.3716308}{10.1109/TDSC.2026.3716308}.}

	 \thanks{arXiv v2: this version corrects, relative to the journal version, the placement of PVSC verification (the algebraic check of opened proofs runs inside the committee's threshold boundary and is attested on-chain, Sections~\ref{sec:system_overview} and~\ref{sec:PVSC_ver}), the noise-correctness statement (Section~\ref{sec:noise_analysis}), and the login decision rule (Section~\ref{sec:SCMV_design}).}

	 \thanks{\copyright~2026 IEEE. Personal use of this material is permitted. Permission from IEEE must be obtained for all other uses, in any current or future media, including reprinting/republishing this material for advertising or promotional purposes, creating new collective works, for resale or redistribution to servers or lists, or reuse of any copyrighted component of this work in other works.}
	
	 }
	
	\markboth{IEEE Transactions on Dependable and Secure Computing}{}%

	
	\maketitle

	\begin{abstract}
		Large-scale biometric identification in outsourced settings requires two properties simultaneously: biometric templates and queries must remain protected during computation, and the encrypted similarity outputs produced by an untrusted compute node must be verifiably correct before any application result is released. Existing FHE-based biometric systems primarily address confidentiality, while practical verifiability introduces two bottlenecks in the underlying encrypted $1{:}N$ matching layer: rotation- and bandwidth-heavy similarity evaluation and the high cost of proving repeated homomorphic similarity traces. We present BioZKFHE, a framework for scalable encrypted biometric identification via verifiable homomorphic similarity evaluation that combines BGV homomorphic computation with committee-based threshold opening, verification, and decryption, and smart-contract enforcement of the committee's verification attestations. To reduce encrypted storage and avoid rotation-heavy encrypted $1{:}N$ matching, we propose \emph{Single-Coefficient Multi-Value (SCMV)} packing, which binds multiple quantized embedding values into each plaintext entry through base-$T$ expansion. To make proof generation practical, we propose \emph{Parallelizable and Verifiable Similarity Computation (PVSC)}, which exploits the Double-CRT execution structure of BGV to decompose each blockwise similarity trace into parallel proof instances that are opened and checked before result release. Under standard lattice assumptions and explicit committee/verifier assumptions, we analyze recoverability, noise growth, confidentiality, encrypted-output integrity, and finalized-result integrity. Experiments on FaceNet and MobileFaceNet show near-lossless biometric utility, up to 67\% encrypted-storage reduction, and about 22--44~s end-to-end proof-verified runtime for 10k--40k templates.
	\end{abstract}
	

	\section{Introduction}
	\label{sec:introduction}
	
	Biometric identification is a natural primitive for binding digital identities to unique humans. This need is especially prominent in Sybil-resistant decentralized systems and other blockchain-assisted identity workflows~\cite{cao2025web,wang2023account}, where biometric $1{:}N$ matching is too expensive to execute directly on-chain. Practical deployments outsource similarity computation off-chain and use a smart contract for commitments, verification attestations, and finalization; this requires protected biometric data and committee-attested verification before release.

	These requirements must hold \emph{simultaneously}. Biometric data are highly sensitive, and leakage is irreversible; even protected representations may suffer inversion and linkability risks~\cite{otroshi2024inversion,otroshi2023linkability}. At the same time, an untrusted compute node should not be able to make the system accept an incorrect encrypted similarity output, nor should downstream opening or decryption turn a proof-verified ciphertext result into an incorrect identification outcome~\cite{zhang2023ppba}. The goal is therefore not merely encrypted biometric matching, but verifiable encrypted biometric identification with committee-governed proof opening and result release.

	Fully homomorphic encryption (FHE) naturally addresses confidentiality by enabling similarity computation on encrypted biometric representations with post-quantum security foundations~\cite{CKKS}. Yet most FHE-based biometric systems still focus on privacy against honest-but-curious servers rather than the malicious outsourced setting in which encrypted outputs must be checked before result release. General verifiable encrypted-computation systems have made rapid progress~\cite{chatel2024veritas,aranha2024heliopolis,zhang2025phalanx}, but they are not tailored to large-scale biometric identification, where one encrypted query must be matched against a large encrypted gallery and the resulting similarity scores must support a final identification decision.

	This paper focuses on two practical bottlenecks that arise when one tries to make encrypted biometric identification verifiable at scale:
	\begin{enumerate}
		\item \textit{Rotation- and bandwidth-heavy encrypted $1{:}N$ matching}: Existing packing techniques~\cite{packing1,packing2,xinpacking3} batch templates into SIMD slots, but similarity computation still often requires expensive ciphertext rotations and slot aggregations, which dominate latency and communication as the repository grows~\cite{halevi2019improved}.
		
		\item \textit{Proof-generation cost for repeated homomorphic similarity traces}: Verifiability can enforce integrity, but proving the full encrypted computation monolithically is impractical. Existing systems~\cite{ringfhe,peev} still incur large circuits and limited parallelism for high-dimensional similarity evaluation, leading to high proof-generation latency.
	\end{enumerate}

	To address these bottlenecks, we present BioZKFHE, a framework for scalable encrypted biometric identification via verifiable homomorphic similarity evaluation. BioZKFHE combines BGV homomorphic computation, committee-based threshold opening and decryption, and smart-contract verification in a blockchain-assisted outsourced workflow: the Compute Node performs encrypted $1{:}N$ similarity evaluation off-chain, the committee threshold-opens proof objects and later threshold-decrypts proof-verified ciphertext outputs, and the smart contract runs public structural checks and records at least $t$ consistent verification attestations before any identification result is finalized.

	BioZKFHE is built around two co-designed techniques. \emph{Single-Coefficient Multi-Value (SCMV)} packing binds multiple quantized embedding values into each packed plaintext entry through a base-$T$ expansion, increasing ciphertext utilization, reducing encrypted storage, and removing rotation-heavy operations from the online similarity path. \emph{Parallelizable and Verifiable Similarity Computation (PVSC)} exploits the Double-CRT execution structure of BGV to decompose each blockwise similarity trace into structured local proof instances, enabling parallel proof generation and committee-side verification after threshold opening. In BioZKFHE, verification is committee attested. The compute node first generates encrypted PVSC proof objects for the prescribed homomorphic matching trace; a threshold committee opens and verifies them inside its threshold boundary, because the soundness of the designated-verifier backend requires the verification state, and the opened proof values from which it could be recovered, to stay secret from the prover; the smart contract then performs public structural checks and records the committee's attestations against the session-bound outputs. Thus, BioZKFHE provides a public audit record for outsourced encrypted matching whose integrity rests on the committee threshold assumption, while release of the final identification result remains threshold-governed.

	In summary, this paper makes the following contributions:
	\begin{itemize}
		\item We formulate BioZKFHE as a blockchain-assisted framework for committee-attested auditability in encrypted biometric identification, with explicit roles for off-chain homomorphic evaluation, threshold-governed proof opening, verification, and decryption, and smart-contract enforcement of structural checks and committee verification attestations.
		
		\item We propose \emph{Single-Coefficient Multi-Value (SCMV)} packing to address the storage and rotation bottleneck in encrypted $1{:}N$ matching, increasing effective packing density, reducing encrypted storage, and eliminating rotation-heavy operations from the online similarity path.
		
		\item We propose \emph{Parallelizable and Verifiable Similarity Computation (PVSC)} to address the proof-scalability bottleneck, decomposing each blockwise BGV similarity trace into a structured batch of parallel proof instances suitable for threshold opening and subsequent committee verification.
		
		\item We formalize recoverability, noise growth, confidentiality, encrypted-output integrity, and finalized-result integrity under explicit committee and verifier assumptions, and implement a Microsoft SEAL prototype that demonstrates near-lossless biometric utility, up to 67\% encrypted-storage reduction, and about 22--44~s end-to-end proof-verified runtime at 10k--40k scale.
	\end{itemize}
	

	\section{Related Work}
	\label{sec:related work}
	
	BioZKFHE lies at the intersection of three lines of work: privacy-preserving biometric computation with FHE, efficiency techniques for encrypted $1{:}N$ matching, and verifiable homomorphic computation under malicious servers. Table~\ref{tab:related_work_summary} summarizes representative works and the gap targeted by BioZKFHE.

	\begin{table}[!t]
		\centering
		{
		\caption{Representative lines of related work, their main ideas, and the remaining gap addressed by BioZKFHE.}
		\label{tab:related_work_summary}
		\scriptsize
		\setlength{\tabcolsep}{2.2pt}
		\setlength{\extrarowheight}{0.5pt}
		\renewcommand{\arraystretch}{1.12}
		\begin{tabularx}{\linewidth}{@{}>{\raggedright\arraybackslash}p{0.24\linewidth}>{\raggedright\arraybackslash}X>{\raggedright\arraybackslash}X@{}}
			\toprule
			\textbf{Line of Work} & \textbf{Representative Focus} & \textbf{Remaining Gap} \\
			\midrule
			FHE biometric systems & HERS~\cite{HERS}; HEFT~\cite{sperling2022heft}; recent HE biometric systems~\cite{choi2024blindmatch,ao2025cryptoface}. Confidential encrypted search, fusion, and identification. & No public verification of outsourced similarity outputs. \\
			\addlinespace[2pt]
			Packing and preselection for encrypted $1{:}N$ matching & Secure Face Matching~\cite{packing1}; coefficient packing~\cite{packing2,CryptoMask}; preselection/group testing~\cite{preselection1,ibarrondo2023grote,xinpacking3}. Higher packing density or fewer comparisons. & Matching-efficiency only; no outsourced-trace verification. \\
			\addlinespace[2pt]
			General VC on encrypted data & Generic VC-on-encrypted-data frameworks~\cite{vfhe1,vfhe2,vfhe3,vfhe4}; practical verifiable FHE systems~\cite{ringfhe,knabenhans2024vfhe,chatel2024veritas,aranha2024heliopolis,zhang2025phalanx}. Generic proofs for encrypted/ring computation. & Not tailored to large biometric galleries or repeated similarity traces. \\
			\addlinespace[2pt]
			Biometric-oriented verifiable FHE & PEEV~\cite{peev}. Prototype verifiable framework for encrypted biometric computation. & Does not jointly address rotation-heavy encrypted $1{:}N$ matching and proof scalability. \\
			\bottomrule
		\end{tabularx}
		}
	\end{table}

	\subsection{FHE-Based Privacy-Preserving Biometric Computation}
	FHE enables biometric computation directly on encrypted templates and queries, providing confidentiality under lattice-based assumptions. Prior work has explored both approximate-number schemes and quantized integer HE such as BGV/BFV~\cite{CKKS,marcolla2022fhe,bgv2014}. Representative systems demonstrate the feasibility of encrypted representation search, template fusion, and encrypted biometric identification~\cite{HERS,sperling2022heft,choi2024blindmatch,ao2025cryptoface}, and recent surveys summarize the broader landscape of privacy-preserving biometric protocols~\cite{zeng2025ppbi}. These works establish that biometric utility can be preserved under encryption, but they primarily optimize confidential execution rather than the malicious outsourced setting considered here, where encrypted outputs must be checked before any identification result is released.
	
	\subsection{Packing and Comparison Reduction for Encrypted $1{:}N$ Matching}
	The scalability of encrypted $1{:}N$ matching depends heavily on how many templates can be packed per ciphertext and how much rotation or comparison overhead is incurred during similarity evaluation. Early work packed one template per ciphertext~\cite{packing1}, while later schemes amortized ciphertext overhead by packing multiple templates~\cite{packing2,xinpacking3}. Orthogonal approaches further reduce encrypted comparisons through preselection or group-testing mechanisms~\cite{preselection1,ibarrondo2023grote}. These techniques improve matching efficiency, but they do not address the correctness of the outsourced computation under malicious behavior. SCMV is complementary to this line of work: rather than pruning candidates, it changes the packed-entry granularity to increase effective packing density and remove rotation-heavy operations from the online encrypted $1{:}N$ matching path.
	
	\subsection{Verifiable Homomorphic Computation and Zero-Knowledge Proofs}
	FHE alone does not guarantee that an untrusted server performed the intended computation correctly. Earlier verifiable-computation frameworks either incurred high overhead or supported restricted classes of encrypted computation~\cite{vfhe1,vfhe2,vfhe3,vfhe4}. More recent work improves practicality through ring-oriented SNARKs, verifiable FHE systems, and FHE-friendly proof optimizations~\cite{ringfhe,lattice-fhe,knabenhans2024vfhe,chatel2024veritas,aranha2024heliopolis,zhang2025phalanx}. Closer to biometric applications, PEEV presents a prototype framework for verifying encrypted-domain biometric computation~\cite{peev}. To the best of our knowledge, however, no prior work simultaneously targets efficient encrypted $1{:}N$ matching and verification of the resulting similarity trace in a blockchain-assisted outsourced biometric-identification workflow. In short, prior work improves confidentiality, matching efficiency, or verifiability largely in isolation, whereas BioZKFHE integrates all three in one workflow. BioZKFHE addresses this gap by combining SCMV for matching efficiency with PVSC for per-block proof parallelism after committee-side opening.

	\begin{figure}[t]
		\centering
		\includegraphics[width=0.85\linewidth]{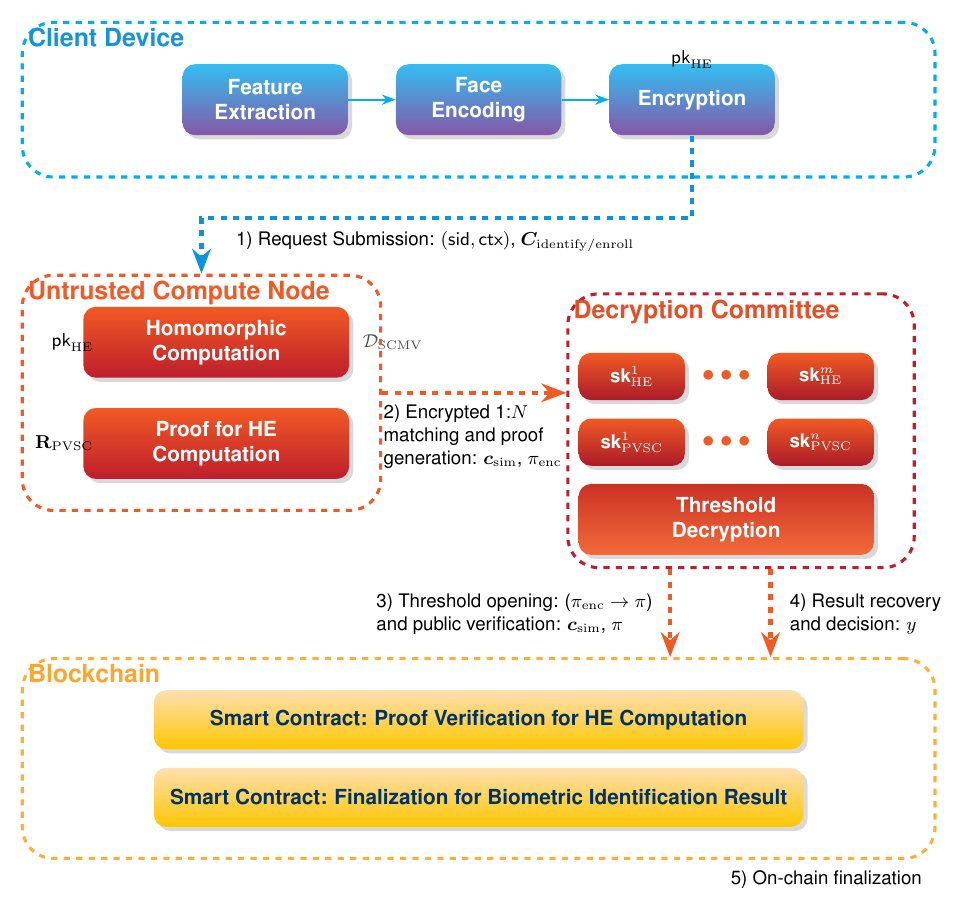}
		\caption{BioZKFHE architecture and protocol flow among the Client Device, Compute Node, Decryption Committee, and smart contract. In the corrected design, the algebraic check of opened proofs is executed by the Decryption Committee under threshold-held verification state and attested on-chain (Section~\ref{sec:PVSC_ver}); the contract performs the structural checks and records the attestations.}
		\label{fig:archi}
	\end{figure}	
	
	\section{System Architecture and Security Model}
	\label{sec:system_overview}
	
	This section presents the system architecture, party roles, workflow, and threat model of BioZKFHE. The system combines BGV FHE~\cite{bgv2014}, the designated-verifier lattice proving paradigm~\cite{ishai2021dvzksnark}, and threshold opening/decryption mechanisms~\cite{bendlin2010threshold}; cryptographic background is deferred to the Supplementary Material, Secs.~I--III.

	\subsection{System Model}
	\label{sec:system_model}
	
	BioZKFHE supports outsourced biometric identification over encrypted templates in a blockchain-assisted setting. The heavy encrypted $1{:}N$ matching runs off-chain at the Compute Node, while the smart contract records commitments, runs public structural checks, records committee verification attestations, and gates session finalization. Fig.~\ref{fig:archi} illustrates the resulting architecture and workflow. For readability, Table~\ref{tab:notation_summary} summarizes the notation used most frequently in the protocol, SCMV construction, and PVSC verification layer.

	\begin{table}[!t]
		\centering
		{
		\caption{Summary of frequently used notation.}
		\label{tab:notation_summary}
		\scriptsize
		\setlength{\tabcolsep}{2.5pt}
		\renewcommand{\arraystretch}{1.05}
		\begin{tabularx}{\linewidth}{@{}>{\raggedright\arraybackslash}p{0.34\linewidth}X@{}}
			\toprule
			\textbf{Notation} & \textbf{Description} \\
			\midrule
			$\lambda$, $d$, $p$, $q$ & Security parameter, BGV ring dimension, plaintext modulus, and ciphertext modulus. \\
			$\mathsf{pk}_{\mathrm{HE}}$, $\mathsf{sk}_{\mathrm{HE}}$ & BGV public key and threshold-shared secret key. \\
			$\mathsf{sk}_{\mathrm{open}}$, $\mathsf{sk}_{\mathrm{open}}^{\ell}$ & Master proof-opening key and the $\ell$-th committee member's opening share. \\
			$(t,n_{\mathrm{com}})$ & Committee threshold and committee size. \\
			$N$, $L$, $\Gamma$ & Number of enrolled templates, embedding dimension, and quantization scale. \\
			$T$, $n$, $D$ & SCMV base, binding factor, and digit-range bound for recoverable similarity scores. \\
			$\mathcal{D}_{\mathrm{SCMV}}$, $\bm{C}_m$, $M$ & SCMV encrypted gallery, its $m$-th encrypted block, and the number of blocks. \\
			$\bm{V}_{\mathrm{identify}}$, $\bm{V}_{\mathrm{enroll}}$ & Plaintext matrices prepared for encrypted identification and conditional enrollment. \\
			$\bm{C}_{\mathrm{identify}}$, $\bm{C}_{\mathrm{enroll}}$ & Query and enrollment ciphertext matrices sent by the client. \\
			$\bm{c}^{\mathrm{sim}}_m$, $\bm{c}_{\mathrm{sim}}$, $y$ & Encrypted similarity output for block $m$, the blockwise output collection, and the final application output. \\
			$\mathsf{ctx}$, $\mathsf{sid}$, $\mathsf{op}$ & Session context, fresh session identifier, and operation type. \\
			$h_{\mathrm{id}}$, $u_m$, $h_m$, $\rho_m$ & Query commitment, block commitment preimage, session-bound block commitment, and descriptor digest. \\
			$\Pi_{\mathrm{enc}}^{(m)}$, $\Pi^{(m)}$ & Encrypted and threshold-opened PVSC proof batch for block $m$. \\
			$\bm{x}^{(m)}$, $\bm{w}^{(m)}$ & PVSC public statement and private witness for block $m$. \\
			$\mathbf{R}_{\mathrm{PVSC}}$, $\mathsf{pp}_{\mathrm{PVSC}}$, $\mathsf{vk}_{\mathrm{PVSC}}$ & PVSC relation family, public parameters, and threshold-shared secret verification state. \\
			$q^{(j)}$, $n_q$, $d_{\mathrm{ZK}}$, $n_{\mathrm{chunk}}$ & RNS modulus, number of RNS limbs, proof-template chunk size, and chunks per limb. \\
			\bottomrule
		\end{tabularx}
		}
		\vspace{-2pt}
	\end{table}

	\subsubsection{Entity Roles and Responsibilities}
	\label{sec:roles}
	
	BioZKFHE involves four entities with distinct roles:
	
	\begin{itemize}
		\item \textbf{Client Device}: The biometric data owner. It holds the FHE public key $\mathsf{pk}_{\mathrm{HE}}$ and is the only entity with access to plaintext biometric data.
		
		\item \textbf{Compute Node}: The outsourced computation service. It uses $\mathsf{pk}_{\mathrm{HE}}$ and the public PVSC parameters $\mathsf{pp}_{\mathrm{PVSC}}$ to store the encrypted database $\mathcal{D}_{\mathrm{SCMV}}$, perform homomorphic matching, and generate PVSC proofs.
		
		\item \textbf{Decryption Committee}: The threshold trust anchor. Its $n_{\mathrm{com}}$ members collectively hold shares $\{\mathsf{sk}_{\mathrm{HE}}^{i}, \mathsf{sk}_{\mathrm{open}}^{i}\}_{i=1}^{n_{\mathrm{com}}}$ and jointly open and verify proofs, attest verification outcomes on-chain, and decrypt similarity outputs. The opening shares also cover the secret verification state $\mathsf{vk}_{\mathrm{PVSC}}$.
		
		\item \textbf{Blockchain Smart Contract}: The public record and state gatekeeper. It stores the public descriptor material $\mathbf{R}_{\mathrm{PVSC}}$ and the committee's attestation keys, runs deterministic structural checks on each submitted batch, records verification attestations and accepted block indices, and finalizes outputs after threshold confirmations.
	\end{itemize}

	\subsubsection{System Initialization}
	\label{sec:init}
	
	BioZKFHE is bootstrapped once, e.g., via a secure Multiparty Computation protocol. This setup generates:
	\begin{itemize}
		\item the key pair $(\mathsf{pk}_{\mathrm{HE}}, \mathsf{sk}_{\mathrm{HE}})$ for the BGV scheme~\cite{bgv2014};
		\item the public PVSC parameters $\mathsf{pp}_{\mathrm{PVSC}}$ (defined concretely in Section~\ref{sec:PVSC_setup}) together with a master proof-opening key $\mathsf{sk}_{\mathrm{open}}$ for the lattice-based designated-verifier ZK proving mechanism~\cite{ishai2021dvzksnark}.
	\end{itemize}
	
	The secret keys $\mathsf{sk}_{\mathrm{HE}}$ and $\mathsf{sk}_{\mathrm{open}}$ are then shared among the Decryption Committee under a $(t,n_{\mathrm{com}})$ threshold mechanism~\cite{bendlin2010threshold}, so no single entity can unilaterally open proofs or decrypt results. The threshold opening/decryption abstraction used here is summarized in the Supplementary Material, Sec.~III. The Client Device receives $\mathsf{pk}_{\mathrm{HE}}$, the Compute Node receives $\mathsf{pk}_{\mathrm{HE}}$ and $\mathsf{pp}_{\mathrm{PVSC}}$, the secret verification state $\mathsf{vk}_{\mathrm{PVSC}}$ is threshold-shared among the committee together with $\mathsf{sk}_{\mathrm{open}}$, and the smart contract records only $\mathbf{R}_{\mathrm{PVSC}}$ and the committee's attestation keys.

		\begin{figure*}[!t]
		\refstepcounter{protocolalg}
		{
			\scriptsize
			\noindent\rule{\textwidth}{0.4pt}
			\par\vspace{1pt}
			\noindent\textbf{Algorithm~\theprotocolalg: BioZKFHE end-to-end protocol pipeline.}
			\label{alg:biozkfhe_pipeline}
			\par\vspace{2pt}
			\noindent\textbf{Input:} operation type $\mathsf{op}\in\{\mathsf{enroll},\mathsf{login}\}$, client biometric sample, and the current ordered SCMV gallery snapshot $\mathcal{D}_{\mathrm{SCMV}}=\{\bm{C}_m\}_{m=1}^{M}$.
			\par
			\noindent\textbf{Public parameters:} $\mathsf{pk}_{\mathrm{HE}}$, $\mathsf{pp}_{\mathrm{PVSC}}$, $\mathbf{R}_{\mathrm{PVSC}}$, the committee attestation keys, and committee threshold policy $(t,n_{\mathrm{com}})$.
			\par
			\noindent\textbf{Output:} a finalized enrollment uniqueness decision or login identity decision $y$.
			\vspace{-2pt}
			\begin{enumerate}
				\setlength{\itemsep}{0pt}
				\setlength{\parsep}{0pt}
				\setlength{\topsep}{2pt}
				\item \textbf{Setup.} Generate $(\mathsf{pk}_{\mathrm{HE}},\mathsf{sk}_{\mathrm{HE}})$ and $\mathsf{pp}_{\mathrm{PVSC}}$; share $\mathsf{sk}_{\mathrm{HE}}$, $\mathsf{sk}_{\mathrm{open}}$, and $\mathsf{vk}_{\mathrm{PVSC}}$ among the committee; let the contract store $\mathbf{R}_{\mathrm{PVSC}}$ and the committee roster.
				\item \textbf{Session creation.} The contract creates a fresh session context $\mathsf{ctx}$ binding $\mathsf{sid}$, $\mathsf{op}$, the ordered gallery snapshot, and ledger metadata.
				\item \textbf{Client preparation.} The client extracts, normalizes, and quantizes the embedding, then runs
				$(\bm{V}_{\mathrm{identify}},\bm{V}_{\mathrm{enroll}})\leftarrow\mathsf{SCMV.Prepare}(\bm{v})$ and
				$(\bm{C}_{\mathrm{identify}},\bm{C}_{\mathrm{enroll}})\leftarrow\mathsf{SCMV.Encrypt}(\bm{V}_{\mathrm{identify}},\bm{V}_{\mathrm{enroll}},\mathsf{pk}_{\mathrm{HE}})$.
				\item \textbf{Bind request.} The contract records $h_{\mathrm{id}}=\mathcal{H}(\mathsf{ctx}\Vert\bm{C}_{\mathrm{identify}})$ and resolves the snapshot commitment for each block index $m$.
				\item \textbf{For each gallery block $m=1,\dots,M$, execute the block-verification loop:}
				\begin{enumerate}
					\setlength{\itemsep}{0pt}
					\setlength{\parsep}{0pt}
					\setlength{\topsep}{1pt}
					\item The Compute Node evaluates $\bm{c}^{\mathrm{sim}}_m\leftarrow\mathsf{SCMV.Identify}(\bm{C}_{\mathrm{identify}},\{\bm{C}_m\})$ and materializes the prescribed BGV multiply--relinearize--accumulate trace.
					\item The Compute Node generates $\Pi_{\mathrm{enc}}^{(m)}\leftarrow\mathsf{PVSC.ProofGen}(m,\mathsf{ctx},\bm{C}_{\mathrm{identify}},\bm{C}_m,\bm{c}^{\mathrm{sim}}_m,\mathsf{pp}_{\mathrm{PVSC}})$ and the descriptor digest $\rho_m$.
					\item At least $t$ committee members open $\Pi_{\mathrm{enc}}^{(m)}$ and run $\mathsf{PVSC.Verify}$ as one threshold computation that outputs only the accept bit, and sign attestations; an untrusted relay submits $(m,\bm{c}^{\mathrm{sim}}_m,h_{\mathrm{id}},h_m,\rho_m)$, the descriptor sequence, and the attestations.
					\item The contract checks $\mathsf{ctx}$, index validity, $h_m$, descriptor order, and $\rho_m$, then checks that at least $t$ consistent attestations accept the batch.
					\item If these checks pass and $m$ has not been accepted in this session, the contract records $m$ as covered; otherwise it rejects the block submission.
				\end{enumerate}
				\item \textbf{Coverage gate.} If any snapshotted block index is missing or duplicated, abort finalization; otherwise enable threshold decryption of the verified encrypted-output collection.
				\item \textbf{Threshold recovery.} The committee threshold-decrypts $\{\bm{c}^{\mathrm{sim}}_m\}_{m=1}^{M}$, recovers the similarity vector internally, and derives the minimal application output $y$.
				\item \textbf{Finalize and update.} Committee members submit confirmations binding $y$ to $\mathsf{ctx}$ and the verified ciphertext-output collection. The contract finalizes after at least $t$ consistent confirmations; if $\mathsf{op}=\mathsf{enroll}$ and $y$ accepts uniqueness, update the encrypted database with $\mathsf{SCMV.Update}(\bm{C}_{\mathrm{enroll}},\mathcal{D}_{\mathrm{SCMV}})$.
			\end{enumerate}
			\vspace{-4pt}
			\noindent\rule{\textwidth}{0.4pt}
		}
		\vspace{-0.8cm}
	\end{figure*}

	\subsubsection{Protocol Overview}
	
	BioZKFHE executes the following high-level workflow:
	\begin{enumerate}
		\item \textbf{Request submission}: For either enrollment or login, the client encrypts a quantized biometric vector under $\mathsf{pk}_{\mathrm{HE}}$. We denote the query ciphertext by $\bm{C}_{\mathrm{identify}}$, the enrollment ciphertext by $\bm{C}_{\mathrm{enroll}}$, and the session context by $\mathsf{ctx}$, which contains a fresh session identifier $\mathsf{sid}$, the fixed SCMV gallery snapshot, and the relevant ledger metadata.
		
		\item \textbf{Encrypted $1{:}N$ matching and proof generation}: The Compute Node evaluates homomorphic similarity over $\bm{C}_{\mathrm{identify}}$ and the SCMV-formatted encrypted database $\mathcal{D}_{\mathrm{SCMV}}=\{\bm{C}_m\}_m$, producing blockwise encrypted similarity outputs $\bm{c}_{\mathrm{sim}} := \{\bm{c}^{\mathrm{sim}}_m\}_m$. It then runs PVSC to generate the encrypted proof collection $\Pi_{\mathrm{enc}} := \{\Pi_{\mathrm{enc}}^{(m)}\}_m$, where each proof batch attests to the correctness of the corresponding block output under the committed session context and inputs.
		
		\item \textbf{Threshold opening, verification, and attestation}: A quorum of at least $t$ committee members opens $\Pi_{\mathrm{enc}}$ and runs $\mathsf{PVSC.Verify}$ on each batch as one threshold computation: the opened proof values stay secret-shared among the participating members, are never reconstructed in the clear, and the computation outputs only the accept bit. An aggregation party submits, for each block, the encrypted output, the descriptor sequence and its digest $\rho_m$, and the committee's signed attestations on-chain; the smart contract re-runs the public structural checks against $\mathsf{ctx}$, records an index as accepted only after at least $t$ consistent attestations, and rejects duplicates within the same session.
		
		\item \textbf{Result recovery and decision}: Once the smart contract has accepted the complete proof-verified block collection exactly once, the committee threshold-decrypts $\bm{c}_{\mathrm{sim}}$ and derives the minimal application output $y$. For enrollment, $y$ is a uniqueness decision indicating whether all similarities are below the threshold. For login, $y$ is either rejection or the matched identity index when the threshold is exceeded.
		
		\item \textbf{On-chain finalization}: Committee members submit consistent confirmations binding the decrypted result $y$ to the verified session, and the smart contract finalizes only after receiving at least $t$ such confirmations.
	\end{enumerate}

	Algorithm~\ref{alg:biozkfhe_pipeline} summarizes the end-to-end BioZKFHE pipeline in pseudo-algorithmic form, making the interaction among SCMV, PVSC, threshold opening/decryption, and contract-side verification explicit. Details of SCMV and PVSC are given in Sections~\ref{sec:SCMV_design} and~\ref{sec:PVSC_design}, respectively.

	
	\subsection{Threat Model}
	\label{sec:threat_model}
	
	BioZKFHE assumes that fewer than $t$ committee members collude, that neither the PVSC verification state nor opened proof values are reconstructed outside the committee's threshold boundary, that the blockchain executes the deployed smart-contract code correctly and deterministically, and that the client follows the specified biometric extraction, quantization, and encryption procedures.
	
	\subsubsection{Security Goals}
	
	BioZKFHE targets the following goals:
	\begin{itemize}
		\item \textbf{G1 (Biometric Confidentiality Against Untrusted Compute and Sub-Threshold Collusion)}: An adversary controlling the Compute Node and fewer than $t$ committee members should learn nothing about the plaintext query embedding, enrolled templates, or intermediate homomorphic states beyond the intentionally released application-level output and the information inherently available to an authorized decryption quorum.
		
		\item \textbf{G2 (Committee-Attested Encrypted-Output Integrity)}: A malicious Compute Node, possibly colluding with fewer than $t$ committee members, should not be able to cause the blockchain to record as accepted an incorrect encrypted similarity output for the committed ciphertext inputs and the indexed committed gallery block resolved from the session snapshot, except with negligible probability.
		
		\item \textbf{G3 (Finalized Result Integrity)}: A malicious aggregation party or a partially corrupted committee should not be able to cause the blockchain to finalize an application output that is inconsistent with the complete proof-verified encrypted similarity result collection, except with negligible probability.
		
		\item \textbf{G4 (Post-Quantum Security Basis)}: The system security should reduce to standard lattice-based assumptions, such as Ring-LWE, thereby targeting resistance against quantum polynomial-time adversaries; the signatures used for committee attestations and confirmations require separate assessment.
	\end{itemize}
	
	\subsubsection{Adversary Model}
	
	We analyze security under the following trust assumptions and adversarial capabilities:
	\begin{enumerate}
		\item \textbf{Malicious Compute Node}: The Compute Node is untrusted. It may deviate arbitrarily from the protocol, e.g., by substituting inputs, replaying stale data, or fabricating results. It also attempts to learn information from the encrypted database and query ciphertexts.
		
		\item \textbf{Partially Corrupted Committee}: We assume a threshold adversary who can corrupt up to $t-1$ committee members. These colluding members may share their internal states, but cannot reconstruct the full secret keys or by themselves produce a valid threshold opening or threshold decryption.
		
		\item \textbf{Malicious Aggregation Party}: Any party that collects and relays proof-opening shares or decryption shares may be malicious. It may delay, withhold, reorder, or misaggregate shares. Our protocol therefore treats the aggregation party as untrusted: it may affect liveness, but it should not be able to make the blockchain accept an invalid proof or finalize an incorrect plaintext result.
		
		\item \textbf{Public Observer}: The blockchain is public. Any external observer can view all on-chain data, including commitments, descriptor digests, verification attestations, and finalized outputs.
	\end{enumerate}
	
	\paragraph*{Scope and Boundaries}
	Our analysis does not cover committee corruption at or above the threshold $t$, liveness or censorship resistance under network or blockchain failures, malicious enrollment governance, inference leakage induced by repeated observation of intentionally released application outputs, or low-level implementation side channels such as timing and memory-access leakage. Within this scope, PVSC certifies per-block correctness for the committed ciphertext inputs, the indexed committed gallery block resolved from the session snapshot, and the prescribed homomorphic similarity circuit, while the smart contract enforces that every gallery-block index is accepted exactly once before result recovery starts. The downstream release of the final application output remains threshold-governed: after the Decryption Committee jointly opens and verifies the encrypted proof objects within its threshold boundary and the smart contract records its attestations, the committee decrypts the similarity result needed by the application, and the blockchain finalizes the session only after receiving at least $t$ mutually consistent committee confirmations. Accordingly, BioZKFHE protects biometric templates, query embeddings, and intermediate homomorphic states against the untrusted Compute Node, public observers, and any coalition of fewer than $t$ committee members, but it does not hide the recovered similarity vector from an authorized quorum participating in threshold decryption.

	\section{Design Details of SCMV}
	\label{sec:SCMV_design}
	
	Given the system architecture in Section~\ref{sec:system_model}, this section focuses on SCMV, i.e., how quantized biometric vectors are packed, encrypted, matched, recovered, and conditionally inserted into the encrypted database. Its goal is to improve the efficiency and scalability of FHE-based biometric identification by redesigning how quantized biometric embeddings are mapped to packed plaintext entries. A background on the adopted BGV computation model and its Double-CRT implementation view is given in the Supplementary Material, Sec.~I.
	
	\subsection{Motivation and Structural Advantage}
	\label{sec:SCMV_motivation}
	\label{sec:SCMV_comparison}
	
	FHE provides privacy by computing directly on ciphertexts, but homomorphic polynomial arithmetic over large moduli is costly. SIMD batching mitigates this cost by packing multiple values into plaintext slots, yet standard biometric layouts remain limited by the slot count $d$: horizontal packing processes templates largely one at a time, while vertical packing stores one quantized coordinate per slot across a row-wise block.
	
	SCMV uses the unused numeric range inside each plaintext entry to bind $n$ quantized values through base-$T$ expansion, so one row polynomial carries $nd$ coordinate values and one processed block covers about $nd$ templates. This changes packing granularity without changing the inner-product matching rule, reduces the block count to $\lceil N/(nd)\rceil$, and keeps the online similarity path rotation-free. Base-$T$ digit packing itself is not claimed as new; SCMV's contribution is the row-wise slot-level binding and its PVSC-aligned block layout. Fig.~\ref{fig:three-scheme} visualizes this structural contrast against horizontal and vertical packing. Supplementary Material, Sec.~IV gives explicit coefficient/slot-domain comparisons.
	

	\subsection{SCMV-Formatted Database Structure}
	\label{sec:SCMVDataStruc}
	
	\begin{figure}[t]
		\centering
		\includegraphics[width=0.9\linewidth]{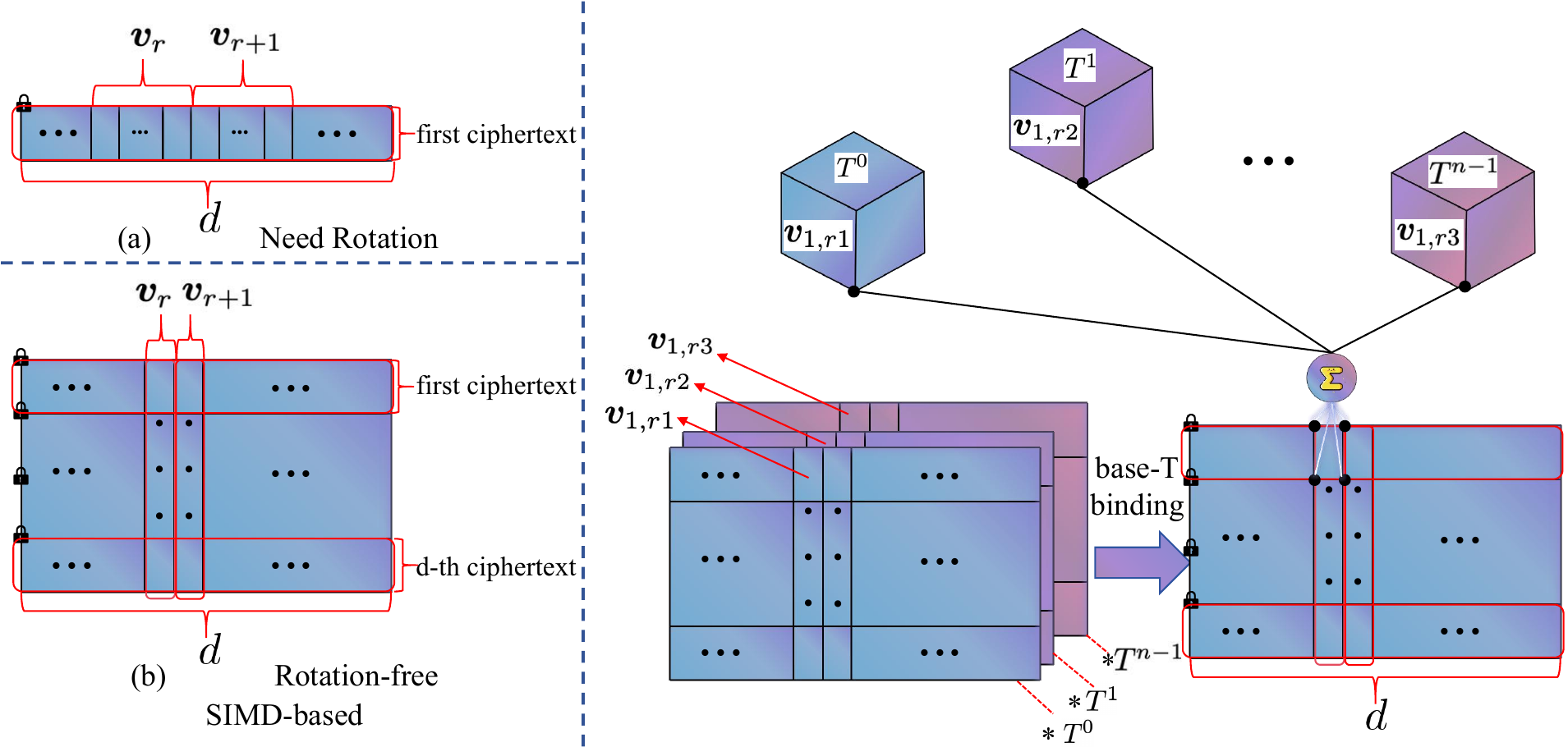}
		\caption{Comparison of horizontal, vertical, and SCMV encoding schemes.}
		\label{fig:three-scheme}
	\end{figure}
	
	The SCMV-formatted database organizes enrolled encrypted templates in a layout that simultaneously improves packing density, supports efficient blockwise homomorphic matching, and preserves a simple update mechanism for new enrollments. We denote this layout by $\mathcal{D}_{\mathrm{SCMV}}$ and describe it in three steps: template organization, base-$T$ multi-value binding, and ciphertext construction.
	
	\subsubsection{Quantized embedding organization}
	Let the enrolled template of user $r$ be represented by an $L$-dimensional quantized biometric embedding $\bm{v}_{r}=(v_{1,r},\dots,v_{L,r})^{\top}$, where $r\in\{1,\dots,N\}$. Partition the $N$ vectors into size-$d$ groups, and define the $\ell$-th group matrix $\bm{V}_{\ell}\in\mathbb{Z}^{L\times d}$ as $\bm{V}_{\ell}:=[\bm{v}_{(\ell-1)d+1},\dots,\bm{v}_{\ell d}]$ for $\ell\in\{1,\dots,\lceil N/d\rceil\}$. Padding with zero vectors is applied if $N$ is not a multiple of $d$.

	\subsubsection{Base-$T$ $n$-value binding}
	SCMV next compresses every $n$ consecutive group matrices into one consolidated matrix by base-$T$ multi-value binding within each packed plaintext entry. Specifically, for $m\in\{1,\dots,M\}$ with $M := \lceil N/(nd)\rceil$, SCMV combines $\{\bm{V}_{(m-1)n+1},\dots,\bm{V}_{mn}\}$ into a bound matrix $\bm{W}_{m}\in\mathbb{Z}^{L\times d}$ via $\bm{W}_{m}=\sum_{t=1}^{n}\bm{V}_{(m-1)n+t}T^{t-1}$. If the number of group matrices is not a multiple of $n$, the missing matrices in the last bundle are treated as all-zero matrices. Then, the entry $w^m_{i,j}$ at the $i$-th row and $j$-th column of $\bm{W}_m$ is given by 
	\begin{equation}
		\label{eq:t_expansion}
		w^m_{i,j} = \sum_{t=1}^{n} v_{i,(m-1)nd+(t-1)d+j}\,T^{t-1},
	\end{equation}
	where $i\in\{1,\dots,L\}$ and $j\in\{1,\dots,d\}$. Eq.~\eqref{eq:t_expansion} shows that the values at coordinate $i$ from $n$ different users are bound into a single integer $w^m_{i,j}$.
	
	\subsubsection{Ciphertext construction}
	Finally, each consolidated matrix $\bm{W}_m$ is converted into an encrypted block. Let $\bm{w}^{m}_{i}$ denote the $i$-th row of $\bm{W}_m$. Each row is batch-encoded and encrypted into a BGV ciphertext $\bm{c}^{m}_{i}\leftarrow\mathsf{BGV.Encrypt}(\mathsf{BGV.Encode}(\bm{w}^{m}_{i}),\mathsf{pk}_{\mathrm{HE}})$. The resulting block is $\bm{C}_{m}:=[\bm{c}^{m}_{1},\dots,\bm{c}^{m}_{L}]$, and the SCMV-formatted encrypted template database is $\mathcal{D}_{\mathrm{SCMV}}:=\{\bm{C}_1,\dots,\bm{C}_{M}\}$.
	
	This description specifies the target encrypted layout of $\mathcal{D}_{\mathrm{SCMV}}$, rather than the full client-side preprocessing workflow. In deployment, client devices locally prepare and encrypt biometric data into SCMV-compatible ciphertexts, while the Compute Node only stores and updates $\mathcal{D}_{\mathrm{SCMV}}$ through the ciphertext-level functions defined below, without accessing plaintext biometric data.
	
	\begin{remark}
		The SCMV layout substantially increases effective packing density. Under standard slot-level SIMD, one ciphertext can encode about $d/L$ users under vertical packing, whereas SCMV encodes about $nd/L$ users by binding $n$ values into each plaintext entry. This $n$-fold increase directly reduces ciphertext storage and the number of blocks processed during encrypted 1:$N$ matching. The admissible choices of $T$ and $n$, which must balance packing density against correctness and noise constraints, are analyzed in Section~\ref{sec:analysis}.
	\end{remark}
	
	With the SCMV data layout in place, we next define how a biometric embedding is mapped into this layout and how encrypted matching, similarity recovery, and database maintenance are carried out over $\mathcal{D}_{\mathrm{SCMV}}$.
	
	\subsection{SCMV Operating Functions}
	\label{sec:operation}
	
	Building on the above database layout, we now formalize the functions of SCMV, namely, biometric preparation, encryption, encrypted 1:$N$ matching, similarity recovery, and database update. In BioZKFHE, enrollment-time uniqueness checking and login-time identification share the same encrypted matching primitive $\mathsf{SCMV.Identify}$ and differ only in the decision rule applied to the recovered similarity vector.

	\subsubsection{Biometric Preparation Function}
	
	Let $\bm{{v}}'\in\mathbb{R}^L$ denote the real-valued biometric embedding extracted from a biometric sample. The preparation function first normalizes $\bm{{v}}'$ to a fixed $L_2$ norm $\Gamma$ and rounds component-wise to obtain the quantized vector $\bm{v}\in\mathbb{Z}^L$, and then maps $\bm{v}$ to two formatted plaintext matrices suitable for encrypted matching and conditional enrollment, formally defined as $(\bm{V}_{\mathrm{identify}},\bm{V}_{\mathrm{enroll}})\leftarrow\mathsf{SCMV.Prepare}(\bm{v})$. The output matrices are constructed as follows:
	\begin{itemize}
		\item $\bm{V}_{\mathrm{identify}}\in\mathbb{Z}^{L\times d}$ is constructed by replicating $\bm{v}$ across all $d$ columns, namely, $\bm{V}_{\mathrm{identify}}=[\bm{v},\bm{v},\dots,\bm{v}]$. This layout enables one encrypted query to be matched in parallel against all users packed into a database block.
		
		\item $\bm{V}_{\mathrm{enroll}}\in\mathbb{Z}^{L\times d}$ prepares the same quantized vector for possible insertion into the SCMV database. Let $N$ be the number of already enrolled users and define the local offset $a:=N\bmod(nd)$. The target column index and binding-layer index are $h:=(a\bmod d)+1$ and $k=\lfloor a/d\rfloor+1$, respectively. The matrix places $\bm{v}T^{k-1}$ in column $h$ and zero vectors elsewhere.
	\end{itemize}
	
	For login, only $\bm{V}_{\mathrm{identify}}$ is needed; for enrollment, $\bm{V}_{\mathrm{identify}}$ supports uniqueness checking, whereas $\bm{V}_{\mathrm{enroll}}$ is retained for conditional insertion.
	
	\subsubsection{Biometric Encryption Function}
	
	The encryption function converts the prepared plaintext matrices into SCMV-compatible ciphertext blocks using the FHE public key $\mathsf{pk}_{\mathrm{HE}}$, defined as $(\bm{C}_{\mathrm{identify}},\bm{C}_{\mathrm{enroll}})\leftarrow\mathsf{SCMV.Encrypt}(\bm{V}_{\mathrm{identify}},\bm{V}_{\mathrm{enroll}},\mathsf{pk}_{\mathrm{HE}})$. As in the database layout, encryption is performed row-wise. For each dimension $i\in\{1,\dots,L\}$, the $i$-th rows of $\bm{V}_{\mathrm{identify}}$ and $\bm{V}_{\mathrm{enroll}}$ are encoded and encrypted to produce ciphertexts $\bm{c}^{\mathrm{id}}_i$ and $\bm{c}^{\mathrm{en}}_i$, respectively. The outputs are the ciphertext matrices $\bm{C}_{\mathrm{identify}}:=[\bm{c}^{\mathrm{id}}_1,\dots,\bm{c}^{\mathrm{id}}_L]$ and $\bm{C}_{\mathrm{enroll}}:=[\bm{c}^{\mathrm{en}}_1,\dots,\bm{c}^{\mathrm{en}}_L]$.
	
	\subsubsection{Encrypted Matching Function}
	
	We use inner-product-based similarity in SCMV. Because all embeddings are normalized to the same $L_2$ norm $\Gamma$, the inner product is proportional to cosine similarity up to a fixed scale factor. The encrypted matching function, executed by the Compute Node, computes encrypted similarity scores between the query embedding and all enrolled templates as $\bm{c}_{\mathrm{sim}}\leftarrow\mathsf{SCMV.Identify}(\bm{C}_{\mathrm{identify}},\mathcal{D}_{\mathrm{SCMV}})$, where $\bm{c}_{\mathrm{sim}}=\{\bm{c}^{\mathrm{sim}}_m\}_{m=1}^{M}$ is the blockwise encrypted similarity output. For each database block $\bm{C}_m\in\mathcal{D}_{\mathrm{SCMV}}$, the node computes
	\begin{equation}
		\label{eq:bioidentify}
		\bm{c}^{\mathrm{sim}}_m = \bigoplus_{i=1}^{L} \big(\bm{c}_i^{\mathrm{id}} \otimes \bm{c}_{i}^m\big),
	\end{equation}
	where $\otimes$ and $\oplus$ denote BGV homomorphic multiplication and addition, respectively. This primitive is shared by enrollment-time uniqueness checking and login-time identification; only the post-recovery decision rule differs.
	
	\textit{Computation Correctness:} Let $\bm{u}_{\mathrm{sim}}^m\in\mathbb{Z}^{d}$ denote the plaintext corresponding to $\bm{c}^{\mathrm{sim}}_m$. Since Eq.~\eqref{eq:bioidentify} is a homomorphic multiply-and-accumulate over the $L$ embedding dimensions, the $j$-th slot of $\bm{u}_{\mathrm{sim}}^m$ is
	\begin{equation}
		\label{eq:u_cos}
		\begin{aligned}
			u^{m}_{j}
			&= \sum_{i=1}^{L} v_i\,w^{m}_{i,j}
			= \sum_{t=1}^{n}\left(\sum_{i=1}^{L} v_i\,v_{i,(m-1)nd+(t-1)d+j}\right)T^{t-1} \\
			&= \sum_{t=1}^{n}\left(\bm{v}^{\top}\bm{v}_{(m-1)nd+(t-1)d+j}\right)T^{t-1}.
		\end{aligned}
	\end{equation}
	Hence each slot of $\bm{c}^{\mathrm{sim}}_m$ compactly stores $n$ similarity scores in base-$T$ representation.
	
	\subsubsection{Similarity Decryption Function}
	
	For exposition, we write the recovery step as $\bm{s}\leftarrow\mathsf{SCMV.Decrypt}(\bm{c}_{\mathrm{sim}},\mathsf{sk}_{\mathrm{HE}})$, where $\bm{s}\in\mathbb{Z}^{N}$ is the plaintext similarity vector. In the BioZKFHE system of Section~\ref{sec:system_model}, this recovery is realized collectively by the Decryption Committee using threshold-held secret shares. For each ciphertext $\bm{c}^{m}_{\mathrm{sim}}$, the decryptor obtains the plaintext vector $\bm{u}^{m}_{\mathrm{sim}}$. Since each element $u^{m}_{j}$ stores $n$ similarity scores in base-$T$ form, the individual scores are recovered by recursive centered decomposition. Specifically, initialize $z_j^{(0)}=u_j^m$, and for $t=1,\dots,n$ compute
	\begin{equation}
		\label{eq:Tdecomp}
		\begin{aligned}
			s_{(m-1)nd+(t-1)d+j} &= \operatorname{ctr}_T\!\left(z_j^{(t-1)} \bmod T\right), \\
			z_j^{(t)} &= \frac{z_j^{(t-1)} - s_{(m-1)nd+(t-1)d+j}}{T},
		\end{aligned}
	\end{equation}
	where $\operatorname{ctr}_T(\cdot)$ maps an integer residue to its centered representative in the interval $(-T/2,T/2]$. If padding was introduced in the last database block, the extra scores corresponding to padded zero templates are discarded, and only the first $N$ recovered scores are retained in $\bm{s}$.
	
	This decomposition is unique and error-free provided that $|\bm{v}^{\top}\bm{v}_r|<T/2$ for all enrolled users, equivalently $T\ge 2D+1$ where $D$ bounds the maximum absolute inner product. Under this condition, each packed similarity score lies strictly within the centered digit range of base $T$ and can therefore be recovered without ambiguity.
	
	In BioZKFHE, $\mathsf{SCMV.Decrypt}$ is an internal recovery primitive executed by the Decryption Committee. The recovered similarity vector $\bm{s}$ need not be published. Instead, the committee derives a minimal application output from $\bm{s}$. For enrollment-time uniqueness checking, insertion is accepted only if $\max_r s_r < \tau_{\mathrm{enroll}}$, where $\tau_{\mathrm{enroll}}$ is the prescribed uniqueness threshold. For login-time identification, access is accepted if $\max_r s_r \ge \tau_{\mathrm{login}}$, in which case the matched identity is $r^{*}:=\arg\max_r s_r$. A gallery-wide $\arg\max$ can bind a query to an account it did not claim whenever another enrolled user's score exceeds $\tau_{\mathrm{login}}$. When login gates access to a specific account, the client should therefore supply the claimed index $\hat r$, and access should be accepted only if $s_{\hat r}\ge\tau_{\mathrm{login}}$; the encrypted matching primitive and the PVSC statement are unchanged by this choice.
	
	\subsubsection{Database Update Function}
	
	Finally, SCMV supports incremental enrollment without decrypting previously stored templates. This function is invoked only after the enrollment-time uniqueness check accepts. The update function integrates a new user's encrypted template into the current database, namely, $\mathcal{D}'_{\mathrm{SCMV}}\leftarrow\mathsf{SCMV.Update}(\bm{C}_{\mathrm{enroll}},\mathcal{D}_{\mathrm{SCMV}})$. If the current last block $\bm{C}_M$ is not fully packed, $\bm{C}_{\mathrm{enroll}}$ is merged into $\bm{C}_M$ via homomorphic addition, i.e., $\bm{c}_i^{M}{}'=\bm{c}^{M}_{i}\oplus\bm{c}^{\mathrm{en}}_{i}$ for all $i\in\{1,\dots,L\}$. Otherwise, $\bm{C}_{\mathrm{enroll}}$ is appended as a new block $\bm{C}_{M+1}$.
	
	\textit{Update integrity:} The append case trivially preserves the SCMV layout. For the merge case, let $w_{i,h}^{M}$ denote the plaintext value at the target slot $h$ in the last block before update, where $h$ and $k$ are the indices defined in $\mathsf{SCMV.Prepare}$. Since the new user has global index $N+1=(M-1)nd+(k-1)d+h$, the updated plaintext becomes
	\begin{equation}
		\label{eq:update}
		w_{i,h}^{M}{}^{\prime}
		= w^{M}_{i,h} + v_{i,N+1} T^{k-1}
		= \sum_{t=1}^{k} v_{i,(M-1)nd+(t-1)d+h} T^{t-1}.
	\end{equation}
	Eq.~\eqref{eq:update} shows that the updated packed entry is exactly the base-$T$ expansion of the $k$ values now assigned to that slot. Therefore, database update preserves the SCMV invariant without decrypting or re-encrypting existing data.

	\section{Design Details of PVSC}
	\label{sec:PVSC_design}
	
		The proposed PVSC scheme certifies a \emph{fixed compiled execution trace} for the SCMV similarity circuit rather than an informal claim that the Compute Node ``performed the intended computation.'' At a high level, PVSC follows the lattice-based designated-verifier proving paradigm of~\cite{ishai2021dvzksnark}; however, BioZKFHE does not use that paradigm as a standalone private-verification mechanism. Instead, for every SCMV block, the Compute Node generates \emph{encrypted} proof instances for a deterministic family of local relations induced by the BGV evaluation trace, an authorized committee threshold-opens those instances, and the smart contract verifies the resulting opened proof batch together with explicit instance descriptors. This stricter view is important: the contract does not accept an opaque monolithic proof, but a parsed, ordered, and schedule-consistent batch that is tied to one unique blockwise execution.
		
		Supplementary Material, Sec.~II provides the pedagogical background on the R1CS/QAP-based compiler viewpoint and the designated-verifier lattice proof mechanism underlying PVSC. The present section focuses on the main-paper abstraction: the blockwise trace being certified, the public statement, the core relation family, and the protocol interfaces exercised by the SCMV similarity circuit of Eq.~\eqref{eq:bioidentify}.
	
	\subsection{Motivation and Decomposition Strategy}
	\label{sec:PVSC_motivation}
	
	The core challenge in verifiable FHE is the algebraic complexity of ring-based ciphertext arithmetic. In BioZKFHE, SCMV turns encrypted biometric matching into a regular blockwise multiply-and-accumulate trace: for each database block, Eq.~\eqref{eq:bioidentify} performs ciphertext multiplications followed by additions over degree-$d$ polynomials with large moduli. Proving this entire trace as one monolithic statement is impractical because the resulting circuit is large and offers limited parallelism.
	
	PVSC addresses this bottleneck by exploiting the Double-CRT execution structure of BGV. In practical libraries such as SEAL \cite{sealcrypto}, ciphertext operations over $\mathcal{R}_q=\mathbb{Z}_q[x]/(x^d+1)$ are implemented through two orthogonal decompositions: an RNS decomposition of the ciphertext modulus and an NTT representation of polynomial multiplication. PVSC leverages these two axes to break one large homomorphic similarity trace into many smaller proof instances of uniform shape. This decomposition preserves the semantics of the original FHE computation while exposing a natural granularity for parallel proof generation. Importantly, under the adopted designated-verifier lattice proof mechanism, these proof instances are generated first in encrypted form and can be checked only after committee-side threshold opening and, in this instantiation, only by holders of the threshold-shared verification state.
	
	\subsection{PVSC-Specific Notation}
	\label{sec:PVSC_notation}
	
	To instantiate this decomposition concretely, we fix the parameters that describe the Double-CRT view used throughout PVSC.
	
	\begin{itemize}
		\item \textbf{RNS decomposition.} The ciphertext modulus is factored as $q=\prod_{j=1}^{n_q} q^{(j)}$, where the moduli $\{q^{(j)}\}_{j=1}^{n_q}$ are pairwise co-prime. Ring operations modulo $q$ are thus realized as $n_q$ parallel channels modulo the individual RNS moduli.
		
		\item \textbf{NTT decomposition and chunking.} Under each modulus $q^{(j)}$, polynomial multiplication is carried out in the NTT domain as coordinate-wise multiplication. To match the capacity of the proof templates, we partition the $d$ NTT coordinates into chunks of size $d_{\mathrm{ZK}}$, yielding $n_{\mathrm{chunk}}:=\lceil d/d_{\mathrm{ZK}}\rceil$ chunks per RNS channel.
	\end{itemize}
	
	Accordingly, one logical ciphertext operation in the SCMV similarity trace is represented inside PVSC by many smaller instances indexed by an RNS modulus and an NTT chunk. This indexing is the basis for the proof parallelism illustrated in Fig.~\ref{fig:PVEC}.
	
	\begin{figure}[t]
		\centering
		\includegraphics[width=0.9\linewidth]{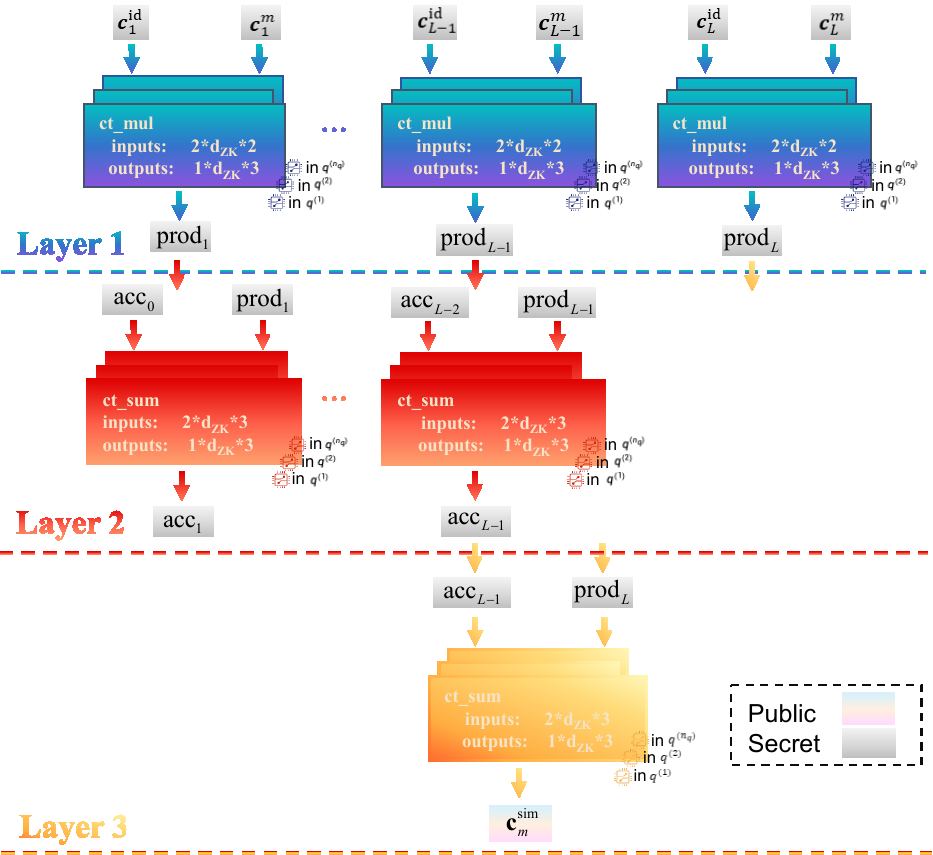}
		\caption{Double-CRT parallelization for PVSC. Operations are decomposed via RNS $\{q^{(j)}\}$ and NTT chunks ($d_{\mathrm{ZK}}$) for parallel ciphertext-domain proof generation.}
		\label{fig:PVEC}
	\end{figure}
	
	\subsection{PVSC Statement, Proof Families, and Verification Scope}
	\label{sec:PVSC_statement}
	
		Since SCMV matching is executed blockwise over $\mathcal{D}_{\mathrm{SCMV}}$, PVSC also proves correctness on a per-block basis. For each SCMV database block $\bm{C}_m\in\mathcal{D}_{\mathrm{SCMV}}$, PVSC certifies a deterministic BGV trace for the fixed input/output tuple $\big(\bm{C}_{\mathrm{identify}}, \bm{C}_m, \bm{c}^{\mathrm{sim}}_m\big)$. Concretely, letting $\bm{p}^{m}_{i}$ denote the relinearized product ciphertext obtained from the $i$-th query row and the $i$-th database row, and $\bm{a}^{m}_{i}$ the running accumulator, the admitted trace is
		\begin{equation}
			\label{eq:pvsc_trace}
			\bm{p}^{m}_{i}:=\mathsf{Relin}\!\left(\bm{c}_{i}^{\mathrm{id}}\otimes \bm{c}_{i}^{m}\right),\quad
			\bm{a}^{m}_{0}:=\bm{0},\quad
			\bm{a}^{m}_{i}:=\bm{a}^{m}_{i-1}\oplus \bm{p}^{m}_{i},
		\end{equation}
	for $i\in\{1,\dots,L\}$, with terminal condition $\bm{a}^{m}_{L}=\bm{c}^{\mathrm{sim}}_{m}$. Thus PVSC certifies one fixed multiply--relinearize--accumulate schedule rather than a generic claim that the server ``computed something consistent.''
		
		To avoid interpreting a blockwise proof with respect to an ad hoc gallery chosen by the Compute Node, BioZKFHE binds PVSC verification to a session-fixed encrypted gallery snapshot. Concretely, the contract binds the current ordered SCMV block-commitment state into $\mathsf{ctx}$ when a session is created, so a proof submitted for block index $m$ is interpreted against the unique committed block at that index. The contract also records accepted indices and enables threshold decryption only after every gallery-block index has been accepted exactly once.
	
	\subsubsection{Public statement}
	
		For block index $m$, the public statement contains the published output ciphertext $\bm{c}^{\mathrm{sim}}_m$, the block index $m$, the session context $\mathsf{ctx}$, the query-binding commitment $h_{\mathrm{id}}:=\mathcal{H}(\mathsf{ctx}\Vert \bm{C}_{\mathrm{identify}})$, the session-bound block commitment $h_m:=\mathcal{H}(\mathsf{ctx}\Vert m\Vert u_m)$ with $u_m:=\mathcal{H}(\bm{C}_m)$, a batch-root digest $\rho_m$ of the canonical ordered proof-instance descriptor sequence, and the public system parameters $\mathsf{pp}_{\mathrm{sys}}$. We write
		\begin{equation}
			\label{eq:pvsc_public_statement}
			\bm{x}^{(m)} :=
			\big(
			\bm{c}^{\mathrm{sim}}_m,\,
			m,\,
			\mathsf{ctx},\,
			h_{\mathrm{id}},\,
			h_m,\,
			\rho_m,\,
			\mathsf{pp}_{\mathrm{sys}}
			\big).
		\end{equation}
		The role of $\rho_m$ is to bind the entire opened proof batch to one canonical ordered blockwise trace and thereby prevent proof splicing or reordering attacks.

	\subsubsection{Private witness}
	
		The prover's witness for block $m$ contains the hidden values needed to instantiate the fixed trace of Eq.~\eqref{eq:pvsc_trace}. Concretely, it includes the ciphertext inputs $\bm{C}_{\mathrm{identify}}$ and $\bm{C}_m$, the product ciphertexts $\{\bm{p}^{m}_{i}\}_{i=1}^{L}$, the accumulator chain $\{\bm{a}^{m}_{i}\}_{i=0}^{L}$, and the Double-CRT decomposition witnesses used by the local proof templates. We denote the full witness by $\bm{w}^{(m)}$.
		
		The exact local witness structure and descriptor syntax are deferred to the Supplementary Material, Sec.~V.
	
	\subsubsection{Proved relation}
	
		PVSC proves that there exists a witness $\bm{w}^{(m)}$ such that the witnessed ciphertext inputs match the public commitments under $\mathsf{ctx}$, the opened proof batch is bound to $\rho_m$, and the published output ciphertext $\bm{c}^{\mathrm{sim}}_m$ equals the result of correctly evaluating the prescribed BGV similarity trace of Eq.~\eqref{eq:pvsc_trace}. Formally, PVSC proves the NP relation
		\begin{equation}
			\label{eq:pvsc_np_statement}
			\exists\,\bm{w}^{(m)} \ \text{s.t.}\ \mathsf{Rel}_{\mathrm{PVSC}}\big(\bm{x}^{(m)},\bm{w}^{(m)}\big)=1.
		\end{equation}
		
		Operationally, $\mathsf{Rel}_{\mathrm{PVSC}}$ combines four kinds of checks, corresponding to the multiplication, accumulation, output-consistency, and block-binding families instantiated in the Supplementary Material, Sec.~V: local multiplication correctness, local accumulation correctness, output consistency with $\bm{c}^{\mathrm{sim}}_m$, and one block-level binding check tying the whole opened batch to $(\mathsf{ctx}, h_{\mathrm{id}}, h_m, \rho_m)$. The smart contract separately checks that the submitted pair $(m,h_m)$ matches the unique committed block at index $m$ in the gallery snapshot bound to $\mathsf{ctx}$. The descriptor-level formulation of these checks is deferred to the Supplementary Material, Sec.~V.
	
	\subsubsection{Instantiation via proof-template families}
	
		PVSC realizes the above relation through a fixed family of proof templates aligned with the Double-CRT execution of the SCMV similarity trace. We denote this template family by $\mathbf{R}_{\mathrm{PVSC}}:=\{\mathcal{R}^{\mathrm{mul}}_{j},\mathcal{R}^{\mathrm{acc}}_{j},\mathcal{R}^{\mathrm{out}}_{j}\}_{j=1}^{n_q}\cup\{\mathcal{R}^{\mathrm{bind}}\}$. The four families play distinct roles: $\mathcal{R}^{\mathrm{mul}}_{j}$ certifies one local multiplication path under modulus $q^{(j)}$, $\mathcal{R}^{\mathrm{acc}}_{j}$ certifies one local accumulator transition, $\mathcal{R}^{\mathrm{out}}_{j}$ links the terminal accumulator to the public ciphertext output, and $\mathcal{R}^{\mathrm{bind}}$ binds the whole batch to $(\mathsf{ctx}, h_{\mathrm{id}}, h_m, \rho_m)$. In our implementation, the canonical descriptor sequence has three counted layers for each block: Layer~1 contributes $n_q n_{\mathrm{chunk}} L$ local instances, Layer~2 contributes $n_q n_{\mathrm{chunk}}(L-2)$ local instances, and Layer~3 contributes $n_q n_{\mathrm{chunk}}$ local instances. Hence the counted arithmetic/output layers contribute $K_m = n_q n_{\mathrm{chunk}}(2L-1)$ instances per block; the full batch additionally includes one block-binding instance.
		The exact descriptor syntax, local relations, and canonical family-major batch layout are given in the Supplementary Material, Sec.~V. At the main-paper level, it suffices to note that proof size and verifier work grow linearly in the counted term $K_m$, with one additional block-binding instance per block.
	
	\subsubsection{Verification scope}
	
		The exact scope of PVSC is intentionally narrower than the full application pipeline. Acceptance of $\Pi^{(m)}$ implies that there exists one \emph{globally consistent} hidden trace satisfying Eq.~\eqref{eq:pvsc_trace}, that the trace is bound to the committed ciphertext inputs and the indexed committed gallery block resolved from the session snapshot bound to $\mathsf{ctx}$, and that its terminal ciphertext equals the published $\bm{c}^{\mathrm{sim}}_m$. Complete $1{:}N$ correctness is obtained only after the contract has accepted every gallery-block index exactly once in the same session.
		
		PVSC does \emph{not} by itself certify client-side feature extraction, client-side encryption, threshold decryption, application-level threshold comparison, or the final identity decision. These later steps remain threshold-governed protocol actions. Accordingly, the direct guarantee provided by PVSC in BioZKFHE is best described as \emph{snapshot-bound, committee-attested encrypted-output integrity after threshold opening}. In the current prototype, this scope includes ciphertext addition, ciphertext multiplication, the relinearization-dependent multiplication path, and the descriptor-level trace-consistency checks described above, but excludes client-side encode/encrypt, threshold decryption, rotations, modulus switching, and bootstrapping.
	
	\subsection{PVSC Core Functions}
	\label{sec:PVSC_functions}
	
	We now describe the protocol-level functions that realize PVSC in BioZKFHE.
	
	\subsubsection{PVSC Setup Function}
	\label{sec:PVSC_setup}
	
		The setup function generates the public PVSC parameters and the threshold-opening material required by PVSC. Formally, $\big(\mathsf{pp}_{\mathrm{PVSC}}, \{\mathsf{sk}_{\mathrm{open}}^{\ell}\}_{\ell=1}^{n_{\mathrm{com}}}\big) \leftarrow \mathsf{PVSC.Setup}(\lambda, \mathbf{R}_{\mathrm{PVSC}}, d_{\mathrm{ZK}})$, where $\mathsf{pp}_{\mathrm{PVSC}} := \big(\mathbf{R}_{\mathrm{PVSC}}, d_{\mathrm{ZK}}, \{\bm{C}_{Q}^{i}\}_{i=1}^{3n_q+1}\big)$. Here, the collection $\{\bm{C}_{Q}^{i}\}_{i=1}^{3n_q+1}$ is a flattened notation for the family-level encrypted proving queries associated with $\big(\mathcal{R}^{\mathrm{mul}}_{j},\mathcal{R}^{\mathrm{acc}}_{j},\mathcal{R}^{\mathrm{out}}_{j}\big)_{j=1}^{n_q}$ together with $\mathcal{R}^{\mathrm{bind}}$, and $\{\mathsf{sk}_{\mathrm{open}}^{\ell}\}_{\ell=1}^{n_{\mathrm{com}}}$ are the threshold shares, collectively held by the Decryption Committee, of the master opening key and of the associated verification state $\mathsf{vk}_{\mathrm{PVSC}}$. Because the soundness of the designated-verifier backend relies on $\mathsf{vk}_{\mathrm{PVSC}}$ remaining secret from the prover, it is not part of $\mathsf{pp}_{\mathrm{PVSC}}$ and is not published. Chunk and concrete step indices appear only when these family-level queries are instantiated into one blockwise proof trace. No single party can unilaterally open an encrypted PVSC proof object.
	
	\subsubsection{PVSC Encrypted Proof Generation Function}
	\label{sec:PVSC_gen}
	
		The proof-generation function is executed by the Compute Node. For each database block $\bm{C}_m$, the Compute Node first computes the output ciphertext $\bm{c}^{\mathrm{sim}}_m$ by evaluating the prescribed SCMV similarity circuit, and then generates the corresponding encrypted proof batch $\Pi_{\mathrm{enc}}^{(m)}:=\langle(\delta^{(m,k)},\pi_{\mathrm{enc}}^{(m,k)})\rangle_{k=1}^{K_m+1}$ (the $K_m$ counted arithmetic and output instances together with the single block-binding instance).
		Formally, $\Pi_{\mathrm{enc}}^{(m)}  \leftarrow
		\mathsf{PVSC.ProofGen}\big(
		m,\mathsf{ctx},\bm{C}_{\mathrm{identify}},\bm{C}_m,
		\bm{c}^{\mathrm{sim}}_m,\mathsf{pp}_{\mathrm{PVSC}}
		\big)$, where the public commitments $h_{\mathrm{id}}:=\mathcal{H}(\mathsf{ctx}\Vert \bm{C}_{\mathrm{identify}})$ and $h_m:=\mathcal{H}(\mathsf{ctx}\Vert m\Vert \mathcal{H}(\bm{C}_m))$ are deterministically derived from the same call context and therefore need not be passed as separate inputs.
		
		Operationally, the prover materializes the fixed trace of Eq.~\eqref{eq:pvsc_trace}, decomposes the relevant ciphertext states into Double-CRT witnesses, forms the canonical ordered descriptor sequence and its digest $\rho_m$, and invokes the local proof templates in $\mathbf{R}_{\mathrm{PVSC}}$. The exact descriptor-level proof-generation rules are given in the Supplementary Material, Sec.~V. Throughout this process, ciphertext inputs and intermediate values remain private witnesses; the only released objects are the published ciphertext output $\bm{c}^{\mathrm{sim}}_m$ and the encrypted proof batch $\Pi_{\mathrm{enc}}^{(m)}$, whose opened form is verified against the public statement $\bm{x}^{(m)}$ of Eq.~\eqref{eq:pvsc_public_statement}.
	
	\subsubsection{PVSC Threshold-Opening Function}
	\label{sec:PVSC_open}
	
	The encrypted proof batch is next opened by the Decryption Committee inside its threshold boundary. Concretely, $[\Pi^{(m)}] \leftarrow \mathsf{PVSC.ThresholdOpen}\big(\Pi_{\mathrm{enc}}^{(m)}, \{\mathsf{sk}_{\mathrm{open}}^{\ell}\}_{\ell\in S}, t\big)$ for any subset $S \subseteq \{1,\dots,n_{\mathrm{com}}\}$ with $|S|\ge t$, where $[\Pi^{(m)}]$ denotes a $(t,n_{\mathrm{com}})$ secret sharing of the opened proof batch. The threshold-opening step preserves the logical structure of the proof families described above: it does not change the proved statement, and it does not output a plaintext proof record. In designated-verifier backends of this family, the opened values are evaluations, at the secret query points, of polynomials known to the prover, so a party that combined them with the prover's view could solve for those points and forge later proofs. The shared batch is therefore consumed directly by the threshold verification of Section~\ref{sec:PVSC_ver}. Relay parties may forward encrypted opening shares between members but never combine them; malformed, reordered, or incomplete shares can at most cause verification failure, but cannot by themselves make a false statement verify.
	
	\subsubsection{PVSC Verification and Attestation Function}
	\label{sec:PVSC_ver}
	
		Verification of an opened batch is split into public structural checks, run by the smart contract, and an algebraic check, run by the committee. Let $\bm{x}^{(m)}$ be the tuple in Eq.~\eqref{eq:pvsc_public_statement}. The contract first parses the session-bound gallery snapshot from $\mathsf{ctx}$, checks that the submitted block index $m$ is valid for that snapshot, resolves the corresponding committed value $u_m$, and computes the expected block commitment $h_m^{\star}:=\mathcal{H}(\mathsf{ctx}\Vert m\Vert u_m)$. It rejects if the submitted $h_m$ differs from $h_m^{\star}$ or if the same block index has already been accepted in the current session.
		
		The contract then deterministically parses the submitted descriptor sequence of the batch, checks that the ordered descriptor sequence has the expected family counts and canonical family-major lexicographic order, and recomputes the sequence digest $\rho_m$. The algebraic check $\mathsf{PVSC.Verify}\big(\bm{x}^{(m)}, \Pi^{(m)}\big)$ requires the verification state $\mathsf{vk}_{\mathrm{PVSC}}$. In designated-verifier constructions of this family, soundness relies on that state remaining secret from the prover~\cite{ishai2021dvzksnark}; publishing it, publishing the opened proof values, or letting a prover adaptively learn it from accept/reject outcomes can allow later proofs to be forged. $\mathsf{PVSC.Verify}$ is therefore executed by the Decryption Committee as a threshold or MPC computation that takes the secret-shared opened batch and the secret-shared verification state as inputs and outputs only the accept bit. Each participating member signs an attestation binding $(\mathsf{ctx},m,h_m,\rho_m,\mathcal{H}(\bm{c}^{\mathrm{sim}}_m),\mathsf{accept})$, and the contract accepts the batch only after receiving at least $t$ consistent attestations for a batch that also passed the structural checks. How the backend behaves under repeated verification, including what a prover can learn from rejected batches, must be analyzed for the exact opening and verification transcript; we leave this analysis open. Opening a lattice ciphertext involves a rounding step, which dominates the cost of evaluating it under MPC; the verification timings of Section~\ref{sec:experiments and analysis} were measured in the clear. The concrete parser conditions are given in the Supplementary Material, Sec.~V. At the protocol level, this step prevents transcript replay, stale-data replay, ciphertext substitution, and proof splicing across sessions or gallery blocks.\footnote{A replayed proof from an old session has a different $\mathsf{sid}$ through $\mathsf{ctx}$; a different query changes $h_{\mathrm{id}}$; a stale gallery block changes the recomputed $h_m^{\star}$; and a same-session duplicate is rejected because $m$ is already recorded. Message withholding is a liveness issue and cannot make a stale proof accepted.}
		
		If all checks pass, the contract records index $m$ as covered. A session becomes eligible for threshold decryption only after every block index in the snapshotted encrypted gallery has been accepted exactly once. Acceptance therefore certifies, under the committee assumption, only that the encrypted similarity output for block $m$ is correct for the committed ciphertext inputs, the indexed committed gallery block resolved from the session snapshot, the canonical descriptor sequence, and the session context. The later release and finalization of the application output remain threshold-governed protocol steps outside $\mathsf{PVSC.Verify}$ itself. The foregoing discussion presents the main-paper abstraction of PVSC; the complete descriptor-level instantiation, including local witness syntax and parser conditions, is given in the Supplementary Material, Sec.~V.

	\textit{Scope and extension.} The PVSC instantiation evaluated here certifies only the fixed SCMV similarity circuit of Eq.~\eqref{eq:bioidentify}, namely a depth-1 multiply--relinearize--accumulate trace; rotations, modulus switching, bootstrapping, encrypted top-$k$/threshold decisions, nonlinear layers, and unnormalized cosine-similarity circuits are outside the current proof domain. Supporting multi-layer BGV circuits would require layered trace descriptors that bind operation tags, layer identifiers, input/output state digests, ciphertext levels, modulus-chain states, and evaluation-key identifiers, together with new local relations for rotations/key switching, modulus switching, and bootstrapping. The parser would also need to enforce topological order and level consistency; Supplementary Material, Sec.~V gives the descriptor-level extension requirements.

	\textit{Toward public verifiability.} Removing the committee from the verification path requires a proof backend whose verification is public, for example hash-based interactive oracle proofs adapted to homomorphically encrypted data~\cite{aranha2024heliopolis}, or a designated-verifier proof augmented with a zero-knowledge proof of correct opening. Either choice would change the proof-generation and verification costs reported in Section~\ref{sec:experiments and analysis}.

	\section{Theoretical Analysis}
	\label{sec:analysis}  
	
	This section analyzes BioZKFHE from three perspectives: parameter selection, decryption correctness, and system security. We first derive admissible choices of the SCMV parameters $T$ and $n$ that maximize packing density while preserving recoverability of the packed similarity scores. We then analyze the growth of BGV noise under the SCMV identification circuit to screen the decryption margin. Finally, we formalize the confidentiality, encrypted-output integrity, and finalized-result integrity goals of BioZKFHE and relate them to concrete attack scenarios in the blockchain-assisted outsourced setting.
	
	\subsection{Parameter Optimization for SCMV}
	\label{sec:param_analysis}
	
	The efficiency of SCMV is governed by two coupled parameters: the base $T$ used to bind multiple similarity scores into one packed plaintext entry and the binding factor $n$ that determines how many scores are packed together. The parameter $T$ must be large enough to support unique recovery of the packed scores under centered base-$T$ decomposition, while $n$ should be as large as possible without causing plaintext wrap-around.
	
	\subsubsection{Optimal Base Selection ($T$)}
	
	We first determine a minimal admissible choice of $T$. Let $\bm{v}'\in\mathbb{R}^L$ denote a real-valued biometric embedding normalized to $\|\bm{v}'\|_2=\Gamma$, and let $\bm{v}\in\mathbb{Z}^L$ be its quantized version. The quantization error is $\Delta \bm{v}=\bm{v}-\bm{v}'$, where each component satisfies $\Delta v_i\in(-0.5,0.5]$.
	
	Since SCMV stores similarity scores as centered base-$T$ digits, unique recovery requires the digit range to contain every possible signed inner product. It is therefore sufficient to choose $T\ge 2D+1$, where $D:=\max_{\bm{v}_1,\bm{v}_2} |\bm{v}_1^{\top}\bm{v}_2|$ denotes the maximum absolute inner product between any two valid quantized embeddings.
	
	To bound $D$, we expand the inner product of two quantized vectors as
	\begin{equation}
		\begin{aligned}
			\bm{v}_1^{\top} \bm{v}_2
			&= (\bm{v}'_1 + \Delta \bm{v}_1)^{\top} (\bm{v}'_2 + \Delta \bm{v}_2) \\
			&= \bm{v}'^{\top}_1 \bm{v}'_2 + \bm{v}'^{\top}_1 \Delta \bm{v}_2
			+ \Delta \bm{v}^{\top}_1 \bm{v}'_2 + \Delta \bm{v}^{\top}_1 \Delta \bm{v}_2.
		\end{aligned}
	\end{equation}
	Using the Cauchy--Schwarz inequality together with $\|\bm{v}'\|_2=\Gamma$ and $\|\Delta \bm{v}\|_\infty\le 0.5$, the four terms are bounded as follows:
	\begin{itemize}
		\item $\bm{v}'^{\top}_1 \bm{v}'_2 \le \|\bm{v}'_1\|_2 \|\bm{v}'_2\|_2 = \Gamma^2$;
		\item $\bm{v}'^{\top}_1 \Delta \bm{v}_2 \le \|\bm{v}'_1\|_1 \|\Delta \bm{v}_2\|_\infty \le \sqrt{L}\|\bm{v}'_1\|_2 \cdot 0.5 = 0.5\Gamma\sqrt{L}$;
		\item similarly, $\Delta \bm{v}^{\top}_1 \bm{v}'_2 \le 0.5\Gamma\sqrt{L}$;
		\item $\Delta \bm{v}^{\top}_1 \Delta \bm{v}_2 \le \sum_{i=1}^{L} |\Delta v_{i,1}\Delta v_{i,2}| \le 0.25L$.
	\end{itemize}
	Since $D$ is integer, the maximum absolute inner product is bounded by $D \le \left\lfloor \Gamma^2 + \Gamma\sqrt{L} + 0.25L \right\rfloor$. Under this bound, the smallest safe choice of the SCMV base is
	\begin{equation}
		\label{eq:T_opt}
		T = 2\left\lfloor \Gamma^2 + \Gamma\sqrt{L} + 0.25L \right\rfloor + 1.
	\end{equation}
	Choosing the smallest admissible $T$ is optimal for packing because it maximizes the number of base-$T$ digits that fit within the plaintext modulus.
	
	\subsubsection{Optimal Binding Factor ($n$)}
	
	We next determine the largest admissible binding factor $n$ once $T$ is fixed. In SCMV, one packed plaintext entry stores $n$ signed similarity scores in base $T$. To avoid modular wrap-around in $\mathbb{Z}_p$, the absolute value of every packed entry must remain below $p/2$.
	
	In the worst case, each packed similarity score has magnitude at most $D$, so the absolute value of a packed entry is bounded by $D\sum_{t=1}^{n} T^{t-1} = D(T^n-1)/(T-1)$. Therefore, a sufficient correctness condition is
	\begin{equation}
		\label{eq:n_condition}
		D\,(T^n-1)/(T-1) < p/2.
	\end{equation}
	Rearranging Eq.~\eqref{eq:n_condition} yields
	\begin{equation}
		\label{eq:n_opt}
		n < \log_T\!\bigl(p(T-1)/(2D) + 1\bigr).
	\end{equation}
	Accordingly, the largest admissible integer binding factor is
	\begin{equation}
		\label{eq:n_max}
		n_{\max} = \left\lceil \log_T\!\bigl(p(T-1)/(2D) + 1\bigr) \right\rceil - 1.
	\end{equation}
	This expression makes explicit how the achievable SCMV packing density grows with the plaintext modulus $p$ and decreases with the similarity bound $D$.
	
	\subsection{Noise Analysis and Decryption Correctness}
	\label{sec:noise_analysis}
	
	The previous subsection ensures that the packed plaintext values fit within the message space. We now screen whether the ciphertext modulus can absorb the noise growth of the homomorphic identification circuit.
	
	\subsubsection{Noise Model}
	
	We measure noise using the canonical embedding norm $\|\cdot\|_{\mathrm{can}}$, which gives a convenient upper bound for cyclotomic rings~\cite{mathCanToolkit}. For a ciphertext $\bm{c}$ decrypting to a message under secret key $\bm{s}$, let $\bm{v}^{*}$ denote the effective decryption error term. Following the standard coarse correctness analysis for BGV~\cite{gentry2012homomorphic}, it suffices to ensure $\|\bm{v}^{*}\|_{\infty}<q/2$, and thus also $\|\bm{v}^{*}\|_{\mathrm{can}}<q/2$.
	
	We use the following standard BGV noise estimates~\cite{gentry2012homomorphic}:
	\begin{itemize}
		\item a fresh ciphertext has noise bound $B_{\text{fresh}} \approx 6p\sqrt{d}\,\sigma$, where $\sigma$ is the standard deviation of the error distribution;
		\item homomorphic addition satisfies $\|\bm{v}^{*}(\bm{c}_1\oplus\bm{c}_2)\|_{\mathrm{can}} \le \|\bm{v}^{*}(\bm{c}_1)\|_{\mathrm{can}} + \|\bm{v}^{*}(\bm{c}_2)\|_{\mathrm{can}}$;
		\item homomorphic multiplication satisfies $\|\bm{v}^{*}(\bm{c}_1\otimes \bm{c}_2)\|_{\mathrm{can}} \le \|\bm{v}^{*}(\bm{c}_1)\|_{\mathrm{can}} \cdot \|\bm{v}^{*}(\bm{c}_2)\|_{\mathrm{can}} \cdot \delta_{\text{exp}}$, where $\delta_{\text{exp}}$ is the ring-expansion factor. In our coarse bound, the additional noise from relinearization is absorbed into the same multiplicative constant.
	\end{itemize}
	
	
	\subsubsection{Noise Bound for SCMV Identification}
	
	We now apply these bounds to the SCMV identification circuit of Eq.~\eqref{eq:bioidentify}. The query ciphertexts in $\bm{C}_{\mathrm{identify}}$ are fresh, so their noise is bounded by $B_{\text{probe}}:=B_{\text{fresh}}$. For the database ciphertexts, two cases arise. If a database block is created directly from a fully packed plaintext row, then each ciphertext in that block is also fresh and has noise $B_{\text{fresh}}$. Under the worst-case online construction path of SCMV, however, a block may be populated incrementally through up to $nd$ ciphertext additions during enrollment updates. We therefore use the conservative bound $B_{\text{db}} \le nd\cdot B_{\text{fresh}}$ for each database ciphertext.
	
	Under this worst-case assumption, the noise of one ciphertext product $\bm{c}_i^{\mathrm{id}} \otimes \bm{c}_{i}^m$ is bounded by $B_{\text{mult}} \le B_{\text{probe}} \cdot B_{\text{db}} \cdot \delta_{\text{exp}} \le nd \cdot \delta_{\text{exp}} \cdot B_{\text{fresh}}^2$. Summing over the $L$ embedding dimensions in Eq.~\eqref{eq:bioidentify} gives the total noise bound
	\begin{equation}
		\label{eq:noise_total}
		B_{\text{total}} \le L \cdot nd \cdot \delta_{\text{exp}} \cdot B_{\text{fresh}}^2 \approx L \cdot nd \cdot \delta_{\text{exp}} \cdot (36p^2 d \sigma^2).
	\end{equation}
	Under this coarse model, the screening condition is $q_{\mathrm{data}} > 2B_{\text{total}}$, where $q_{\mathrm{data}}$ is the modulus carried by the evaluated ciphertexts; in SEAL, the last prime of the coefficient-modulus chain is reserved for key switching and is absent from the data level~\cite{sealcrypto}. Successful BGV decryption is exact, so the digit-level perturbation is $\eta=0$; the SCMV margin is therefore $D<T/2$, or $D+\eta<T/2$ under a residual-error model. Table~\ref{tab:noise_screen} applies this screen to the parameter sets of Section~\ref{efficiency} with $\sigma=3.2$ and $\delta_{\text{exp}}=1$. Eq.~\eqref{eq:noise_total} is a conservative and incomplete estimate: it does not track every message--error cross term or relinearization contribution of the actual SEAL trace, and summing $nd$ independent fresh ciphertexts heuristically grows the noise standard deviation as $\sqrt{nd}$ rather than $nd$~\cite{mathCanToolkit}. Passing the screen therefore does not prove correctness, and failing it does not imply decryption failure; the acceptance criterion for each set is a measured noise budget over fresh galleries and the maximum supported update history, with decoded scores checked against a plaintext reference. Adding one 40-bit prime to the chains of the FaceNet $n=3$ and MobileFaceNet $n=2$ sets satisfies the worst-case screen (data moduli of 165 and 155 bits) while keeping their key moduli (210 and 200 bits) within the 218-bit bound of the Homomorphic Encryption Standard for 128-bit classical security at $d=8192$. The timing and storage results of Section~\ref{sec:experiments and analysis} refer to the sets as listed.

	\begin{table}[!t]
		\centering
		{
		\caption{Noise screening of the evaluated parameter sets ($d=8192$). ``Data'' is the key modulus minus the last (key-switching) prime. Required bits are $\log_2(2B_{\text{total}})$ from Eq.~\eqref{eq:noise_total} (worst case) and with $nd$ replaced by $\sqrt{nd}$ (heuristic).}
		\label{tab:noise_screen}
		\scriptsize
		\setlength{\tabcolsep}{3pt}
		\begin{tabular}{lccccc}
			\toprule
			\textbf{Set} & \textbf{Key/data} & \textbf{$p$ (bits)} & \textbf{Worst} & \textbf{Heur.} & \textbf{Screen} \\
			\midrule
			FaceNet, $n=1$ & 110 / 80 & 20 & 82.5 & 76.0 & heur.\ only \\
			FaceNet, $n=2$ & 150 / 105 & 31 & 105.5 & 98.5 & heur.\ only \\
			FaceNet, $n=3$ & 170 / 125 & 46 & 136.1 & 128.8 & fails \\
			MobileFaceNet, $n=1$ & 120 / 90 & 21 & 86.5 & 80.0 & passes \\
			MobileFaceNet, $n=2$ & 160 / 115 & 40 & 125.5 & 118.5 & fails \\
			\bottomrule
		\end{tabular}
		}
	\end{table}
	
	\subsection{Security Analysis}
	\label{sec:security_analysis}
	
	We now analyze BioZKFHE under the threat model of Section~\ref{sec:threat_model}. The Compute Node may be fully malicious, the aggregation party that relays opening or decryption shares is untrusted, and the Decryption Committee is modeled through a $(t,n_{\mathrm{com}})$-threshold assumption. Our arguments rely on correct and deterministic execution of the committee-side verifier and the contract-side structural checks, on the integrity of on-chain public commitments and committee confirmations, and on the narrower verification scope of PVSC established in Section~\ref{sec:PVSC_statement}: PVSC directly certifies only the correctness of the encrypted similarity output, whereas threshold decryption and result finalization remain protocol-level steps.
	
	\subsubsection{Game-Based Security Definitions}
	
	We formalize G1--G3 through three games. Throughout, $\lambda$ denotes the security parameter and $\mathcal{A}$ is a PPT adversary that may corrupt the Compute Node, the aggregation party, and at most $t-1$ committee members.
	
	\begin{definition}[Biometric Confidentiality Against Untrusted Compute and Sub-Threshold Collusion]
		\label{def:bio_conf}
		In the experiment $\mathsf{Exp}^{\mathrm{conf}}_{\mathcal{A}}(\lambda)$, the challenger runs setup and gives $\mathcal{A}$ all public parameters together with the shares of corrupted committee members. The adversary submits two enrollment sets of equal size and two candidate query embeddings such that the two challenge worlds induce the same public session structure and the same intentionally released application-level output transcript. The challenger samples $b\leftarrow\{0,1\}$, constructs the encrypted database $\mathcal{D}^{(b)}_{\mathrm{SCMV}}$, the encrypted query $\bm{C}^{(b)}_{\mathrm{identify}}$, the encrypted similarity outputs $\bm{c}^{(b)}_{\mathrm{sim}}=\{\bm{c}^{(b),m}_{\mathrm{sim}}\}_{m}$, and the encrypted PVSC proof collection $\Pi_{\mathrm{enc}}^{(b)}=\{\Pi_{\mathrm{enc}}^{(b,m)}\}_{m}$, and runs the threshold opening and verification on it, producing the accept bits and attestations $\beta^{(b)}$. The adversary receives the released view $\mathsf{View}^{(b)}:=(\mathcal{D}^{(b)}_{\mathrm{SCMV}},\bm{C}^{(b)}_{\mathrm{identify}},\bm{c}^{(b)}_{\mathrm{sim}},\Pi_{\mathrm{enc}}^{(b)},\beta^{(b)})$ and the corrupted members' views of that threshold computation, together with the corresponding on-chain commitments and public session metadata, and outputs a bit $b'$. The experiment outputs $1$ iff $b'=b$. Define $\mathsf{Adv}^{\mathrm{conf}}_{\mathcal{A}}(\lambda):=|\Pr[b'=b]-1/2|$. BioZKFHE achieves this confidentiality notion if $\mathsf{Adv}^{\mathrm{conf}}_{\mathcal{A}}(\lambda)$ is negligible for all PPT adversaries $\mathcal{A}$.
	\end{definition}
	
		\begin{definition}[Encrypted Similarity-Output Integrity]
			\label{def:comp_int}
			In the experiment $\mathsf{Exp}^{\mathrm{enc\mbox{-}out\mbox{-}int}}_{\mathcal{A}}(\lambda)$, the adversary outputs a session context $\mathsf{ctx}$, a block index $m$, a query ciphertext $\bm{C}_{\mathrm{identify}}$, a candidate encrypted similarity output $\bm{c}^{*}_{\mathrm{sim}}$, an opened proof $\Pi^{*}$, and commitments $h_{\mathrm{id}}, h_m$. The challenger resolves from the gallery snapshot bound to $\mathsf{ctx}$ the unique committed block $\bm{C}_m$ at index $m$, checks that $h_{\mathrm{id}}=\mathcal{H}(\mathsf{ctx}\|\bm{C}_{\mathrm{identify}})$ and $h_m=\mathcal{H}(\mathsf{ctx}\|m\|\mathcal{H}(\bm{C}_m))$, parses the descriptor sequence contained in $\Pi^{*}$, and computes the induced batch-root digest $\rho^{*}_{m}$. Let $\bm{x}^{*}:=(\bm{c}^{*}_{\mathrm{sim}},m,\mathsf{ctx},h_{\mathrm{id}},h_m,\rho^{*}_{m},\mathsf{pp}_{\mathrm{sys}})$ be the corresponding public statement. If either the deterministic parser rejects $\Pi^{*}$ or $\mathsf{PVSC.Verify}(\bm{x}^{*},\Pi^{*})$ rejects, the experiment outputs $0$. Otherwise, the challenger recomputes the honest encrypted similarity output $\bm{c}^{\mathrm{hon}}_{\mathrm{sim}}:=\mathsf{SCMV.Identify}(\bm{C}_{\mathrm{identify}},\{\bm{C}_m\})$ and outputs $1$ iff $\bm{c}^{*}_{\mathrm{sim}}\neq \bm{c}^{\mathrm{hon}}_{\mathrm{sim}}$. BioZKFHE achieves encrypted similarity-output integrity if $\Pr[\mathsf{Exp}^{\mathrm{enc\mbox{-}out\mbox{-}int}}_{\mathcal{A}}(\lambda)=1]$ is negligible for all PPT adversaries $\mathcal{A}$. Throughout, $\mathsf{PVSC.Verify}$ is executed by the challenger with the verification state $\mathsf{vk}_{\mathrm{PVSC}}$, of which $\mathcal{A}$ holds only the shares of corrupted committee members.
		\end{definition}
	
		\begin{definition}[Finalized Result Integrity]
			\label{def:res_int}
			Let $\mathsf{Dec}_{\mathrm{app}}(\cdot)$ denote the deterministic application decision rule fixed by the session context, e.g., enrollment-time uniqueness checking or login-time identification. In the experiment $\mathsf{Exp}^{\mathrm{res\mbox{-}int}}_{\mathcal{A}}(\lambda)$, the adversary outputs a session context $\mathsf{ctx}$, a query ciphertext $\bm{C}_{\mathrm{identify}}$, a claimed blockwise encrypted similarity-output collection $\bm{c}^{*}_{\mathrm{sim}}:=\{\bm{c}^{*,m}_{\mathrm{sim}}\}_{m}$, a corresponding opened proof collection $\Pi^{*}:=\{\Pi^{*,m}\}_{m}$, a candidate final output $y^{*}$, and a set of purported committee confirmations $\Sigma^{*}$. The challenger computes $h_{\mathrm{id}}:=\mathcal{H}(\mathsf{ctx}\|\bm{C}_{\mathrm{identify}})$, checks that the submitted block indices cover every index in the snapshotted encrypted gallery exactly once, and for each block index $m$ resolves the unique committed block $\bm{C}_m$ from the gallery snapshot bound to $\mathsf{ctx}$. It then computes $h_m:=\mathcal{H}(\mathsf{ctx}\|m\|\mathcal{H}(\bm{C}_m))$, parses the descriptor sequence contained in $\Pi^{*,m}$ to obtain $\rho^{*,m}$, forms the corresponding public statement $\bm{x}^{*,m}:=(\bm{c}^{*,m}_{\mathrm{sim}},m,\mathsf{ctx},h_{\mathrm{id}},h_m,\rho^{*,m},\mathsf{pp}_{\mathrm{sys}})$, and outputs $0$ if any deterministic parsing or verification $\mathsf{PVSC.Verify}(\bm{x}^{*,m},\Pi^{*,m})$ rejects. The challenger next checks whether $\Sigma^{*}$ contains at least $t$ mutually consistent confirmations from distinct committee members recognized by the contract, all binding the same session and the same proof-verified encrypted similarity-output collection together with $\mathcal{H}(y^{*})$; otherwise it outputs $0$. Finally, it threshold-decrypts the proof-verified encrypted similarity-output collection ordered by block index, applies $\mathsf{Dec}_{\mathrm{app}}(\cdot)$ to the recovered similarity result, and outputs $1$ iff the resulting reference output differs from $y^{*}$. BioZKFHE achieves finalized-result integrity if $\Pr[\mathsf{Exp}^{\mathrm{res\mbox{-}int}}_{\mathcal{A}}(\lambda)=1]$ is negligible for all PPT adversaries $\mathcal{A}$.
		\end{definition}
	
	\subsubsection{Main Security Theorems}
	
	\begin{theorem}[Confidentiality]
		\label{thm:conf}
		Assume that (i) BGV under $(\mathsf{pk}_{\mathrm{HE}},\mathsf{sk}_{\mathrm{HE}})$ is IND-CPA secure under the standard Ring-LWE assumption, (ii) the threshold sharing used for $\mathsf{sk}_{\mathrm{HE}}$ and $\mathsf{sk}_{\mathrm{open}}$ reveals no information given fewer than $t$ shares, (iii) the encrypted proof objects output by PVSC are computationally hiding prior to threshold opening, and (iv) the threshold opening-and-verification computation reveals to any coalition of fewer than $t$ members nothing beyond its accept output. Then for any PPT adversary $\mathcal{A}$, the advantage $\mathsf{Adv}^{\mathrm{conf}}_{\mathcal{A}}(\lambda)$ is negligible.
	\end{theorem}
	
	\begin{proof}[Proof sketch]
		We use a hybrid argument. First replace the corrupted members' views of the threshold opening and verification with views simulated from the accept bits alone; by assumption~(iv), this changes the adversary's view only negligibly. Next replace the encrypted proof objects with simulated ciphertext-domain proof objects; by their hiding property and the secrecy of the threshold opening material below threshold, this again changes the view only negligibly. Finally, replace the encrypted database, encrypted query, and encrypted similarity outputs for $b=0$ with those for $b=1$; by IND-CPA security of BGV, the resulting views are computationally indistinguishable. Hence the adversary's distinguishing advantage is negligible.
		
		This theorem captures confidentiality against the untrusted Compute Node, public observers, and any coalition of fewer than $t$ committee members. It does not claim secrecy against an authorized quorum that participates in threshold decryption of the similarity result.
	\end{proof}
	
		\begin{theorem}[Encrypted Similarity-Output Integrity]
			\label{thm:int}
			Assume that (i) every instantiated local relation in $\mathbf{R}_{\mathrm{PVSC}}$ is sound under the adopted designated-verifier proving backend against a prover that does not know $\mathsf{vk}_{\mathrm{PVSC}}$, (ii) neither $\mathsf{vk}_{\mathrm{PVSC}}$ nor any opened proof value is reconstructed outside the threshold boundary, and the backend remains sound when the prover observes accept/reject outcomes across repeated verifications, (iii) fewer than $t$ committee members are corrupted and honest members attest acceptance only after $\mathsf{PVSC.Verify}$ accepts, (iv) the deterministic parser of Section~\ref{sec:PVSC_ver} accepts only complete, canonically ordered, and hash-consistent descriptor sequences, (v) the commitment hash $\mathcal{H}$ is collision resistant, and (vi) the smart contract executes the structural checks and the $t$-attestation rule correctly and deterministically on immutable on-chain commitments and resolves the unique committed block for each submitted index from the gallery snapshot bound to $\mathsf{ctx}$. Then for any PPT adversary $\mathcal{A}$, the probability $\Pr[\mathsf{Exp}^{\mathrm{enc\mbox{-}out\mbox{-}int}}_{\mathcal{A}}(\lambda)=1]$ is negligible.
		\end{theorem}
		
		\begin{proof}[Proof sketch]
			Suppose that $\mathcal{A}$ wins with non-negligible probability. Then the contract records a batch for the public statement $\bm{x}^{*}$ as accepted, yet the published ciphertext $\bm{c}^{*}_{\mathrm{sim}}$ is not the honest encrypted similarity output for the committed query input and the unique committed gallery block at the submitted index $m$ in the snapshot bound to $\mathsf{ctx}$. Acceptance requires at least $t$ consistent attestations, of which at least one comes from an honest member that ran $\mathsf{PVSC.Verify}$ and accepted. Since acceptance also requires successful deterministic parsing, the opened batch yields a complete and hash-consistent descriptor sequence whose digest equals $\rho_m^{*}$. Hence the only ways for $\mathcal{A}$ to win are: (a) some false local arithmetic or block-binding statement is accepted by an honest verifier, contradicting local soundness under assumptions (i)--(ii); or (b) the adversary changes the underlying ciphertext inputs, the descriptor sequence, or the indexed gallery block while preserving $(h_{\mathrm{id}},h_m,\rho_m^{*})$, which implies a collision in $\mathcal{H}$. Therefore the winning probability must be negligible. This argument does not establish public verifiability: it relies on the committee holding $\mathsf{vk}_{\mathrm{PVSC}}$ and on assumption~(ii), whose formal treatment under adaptive accept/reject observation is left open.
			
			This argument bounds one block; session-level accumulation is stated in Corollary~\ref{cor:session_soundness}. We rely on Ring-LWE for BGV, the Module-LWE/Module-SIS-type proof-backend assumptions of~\cite{ishai2021dvzksnark}, and collision resistance of $\mathcal{H}$; see Supplementary Material, Sec.~V.
		\end{proof}

		\begin{corollary}[Session-level soundness]
			\label{cor:session_soundness}
			Let $\epsilon_{\mathrm{loc}}$ be the worst-case soundness error of one opened local proof under the Module-LWE/Module-SIS-type proof backend of~\cite{ishai2021dvzksnark}, and let $\epsilon_{\mathcal{H}}$ be the collision advantage of $\mathcal{H}$. Under the assumptions of Theorem~\ref{thm:int}, a query session over $M=\lceil N/(nd)\rceil$ gallery blocks, each carrying $K_m=n_qn_{\mathrm{chunk}}(2L-1)$ counted arithmetic/output instances plus one block-binding instance ($K_m+1$ verified instances per block), satisfies
			\begin{equation}
				\label{eq:session_soundness_main}
				\epsilon_{\mathrm{sess}}\;\le\;M(K_m+1)\,\epsilon_{\mathrm{loc}}+\epsilon_{\mathcal{H}},
			\end{equation}
			since the deterministic parser contributes no probabilistic error. Hence retaining a target session-soundness level requires the local proof parameter to absorb the additive union-bound loss $\lceil\log_2(M(K_m+1))\rceil$ bits; the representative accounting is tabulated in Supplementary Material, Sec.~V.
		\end{corollary}
		\begin{proof}[Proof sketch]
			Any accepted false session forces either a false accepted local instance or a hash-binding collision, so the claim follows from a union bound over the $M(K_m+1)$ accepted local instances (the $MK_m$ counted arithmetic/output instances together with the $M$ per-block $\mathcal{R}^{\mathrm{bind}}$ instances) and the collision term.
		\end{proof}
	
		\begin{theorem}[Finalized Result Integrity]
			\label{thm:resint}
			Assume that (i) Theorem~\ref{thm:int} holds, (ii) committee confirmations recorded on-chain are member-authenticated and cannot be forged for honest committee members, (iii) fewer than $t$ committee members are corrupted, and (iv) each honest committee member issues a confirmation only after locally deriving the application output from the complete proof-verified encrypted similarity-output collection ordered by block index. Then for any PPT adversary $\mathcal{A}$, the probability $\Pr[\mathsf{Exp}^{\mathrm{res\mbox{-}int}}_{\mathcal{A}}(\lambda)=1]$ is negligible.
		\end{theorem}
		
		\begin{proof}[Proof sketch]
			Assume that $\mathcal{A}$ wins with non-negligible probability. Since Definition~\ref{def:res_int} requires successful deterministic parsing and committee-attested verification of every opened proof against the corresponding encrypted similarity output and also requires complete coverage of the snapshotted gallery block indices, Theorem~\ref{thm:int} implies that the accepted encrypted similarity-output collection is correct except with negligible probability. Since fewer than $t$ committee members are corrupted, any accepted confirmation set of size at least $t$ contains at least one honest member. By assumption, an honest member confirms only the output obtained by locally threshold-decrypting the complete proof-verified encrypted similarity-output collection ordered by block index and applying $\mathsf{Dec}_{\mathrm{app}}(\cdot)$ to the recovered similarity result. Hence the finalized output cannot differ from the reference output except with negligible probability.
		\end{proof}

	\subsubsection{Resistance to Specific Attacks}
	
		Replay and substitution attacks are prevented by the session-bound public statement together with the descriptor-sequence binding. Each PVSC instance is tied to a unique context $\mathsf{ctx}$ containing a fresh session identifier $\mathsf{sid}$, and the smart contract rejects reused or expired sessions. The on-chain commitments $h_{\mathrm{id}}$ and $h_m$ bind the proof to the exact query ciphertext and to the indexed committed gallery block resolved from the session snapshot, while the digest $\rho_m$ binds the entire ordered local-proof descriptor sequence. Hence replaying an old proof, silently substituting $\bm{C}_{\mathrm{identify}}$ or $\bm{C}_m$, splicing local proofs from another batch, or moving a valid proof across sessions makes the public statement inconsistent unless the adversary finds a collision in $\mathcal{H}$ or breaks PVSC soundness. Omission, duplication, and reordering attacks are additionally prevented by the deterministic parser and the contract-side coverage rule: the contract records which indices have already been accepted in the current session and enables threshold decryption only after every gallery-block index has been accepted exactly once.
	
	A malicious aggregation party may still misorder, withhold, or misaggregate proof-opening shares or result-decryption shares, and may attempt to mix a valid proof-verified encrypted similarity-output collection with an incorrect plaintext application output. On the proof side, such behavior can at most cause opening failure or later verification failure, because acceptance requires both the public structural checks and at least $t$ consistent verification attestations. On the result side, the contract accepts an output only after receiving at least $t$ mutually consistent committee confirmations bound to the same proof-verified encrypted similarity-output collection. Since fewer than $t$ committee members may be corrupted, any accepted confirmation set contains at least one honest member, and honest members confirm only outputs obtained by locally deriving the application result from the proof-verified encrypted similarity-output collection. Therefore malicious aggregation may affect liveness, but it cannot by itself cause an incorrect plaintext result to be finalized on-chain, as formalized by Definition~\ref{def:res_int} and Theorem~\ref{thm:resint}.

	\section{Experimental Evaluation}
	\label{sec:experiments and analysis}
	
	We evaluate four claims: integer quantization preserves biometric utility, SCMV improves encrypted $1{:}N$ matching scalability, PVSC makes encrypted-output integrity practical, and the integrated proof-verified pipeline remains feasible. Unless explicitly stated otherwise, verification results refer only to the opened-proof layer, and latency excludes face detection and neural feature extraction and is averaged over 100 randomly selected query images.
	
	\subsection{Experimental Methodology and Setup}
	\label{setup}
	
	\subsubsection{Hardware and implementation}
	
	BioZKFHE was implemented on Microsoft SEAL~\cite{sealcrypto}. All encrypted matching, database-operation, and PVSC parallelization benchmarks were executed on the same Linux server with dual Intel Xeon Platinum 8352V CPUs (2.10~GHz, 36 cores/72 threads per socket), 60~GB RAM, and Ubuntu~22.04.\footnote{Artifact repository containing implementation code, benchmark scripts, plotting scripts, and configuration files: \url{https://github.com/plan-lab-szu/BioZKFHE}.} Unless stated otherwise, ``matching latency'' refers to $\mathsf{SCMV.Identify}$ only and excludes client-side feature extraction, whereas ``verification-side cost'' refers to threshold opening plus verification of the opened proofs.

	For the smart-contract layer, we implemented a Solidity~0.8.24/Hardhat~3.7.0 prototype (optimizer 200, EDR simulated L1) for session binding, opened-proof headers, coverage tracking, committee confirmations, finalization, and the opened-QAP verifier kernels used by $\mathsf{PVSC.Verify}$. These measurements cover contract-side enforcement and verifier-kernel execution for opened local proof instances. They do not treat all uncompressed PVSC proof payloads as a gas-optimized full-batch L1 submission; instead, full-batch size and scaling are accounted for in the proof-footprint and session-scaling results. Costs are 31{,}649--114{,}799 gas for state-management functions and 50{,}980--63{,}830 gas for QAP verifier kernels, with details in Supplementary Sec.~V. Under the corrected verification placement of Section~\ref{sec:PVSC_ver}, $\mathsf{PVSC.Verify}$ runs inside the committee rather than on-chain, so the QAP-kernel figures are reported for reference only; the contract instead records committee attestations, which have the same form as the measured committee confirmations and were not measured separately.

	\subsubsection{Datasets, feature extraction, and evaluation usage}
	
	We constructed the evaluation pool from two public face datasets, LFW and CASIA-WebFace. Each image was first processed by MTCNN for face detection and then converted to a fixed-dimensional embedding using either FaceNet~\cite{facenet} or MobileFaceNet~\cite{mobilefacenets}, producing $L=128$ and $L=512$ embeddings, respectively. The feature extractors are trained on VGGFace2 and kept fixed during evaluation; no additional fine-tuning is performed on LFW or CASIA-WebFace. Each embedding was then $L_2$-normalized, scaled by a quantization factor $\Gamma$, and rounded to integers before encryption, following the SCMV preparation procedure in Section~\ref{sec:operation}.
	
	For biometric evaluation, gallery and query templates are drawn from disjoint images of each retained identity: if an identity has at least two usable images, one image is used as the gallery template and the remaining images are used as query templates. In the encrypted experiments, BioZKFHE stores one enrolled template per user, so $N$ denotes both the gallery template count and the enrolled user count. The same identity set and gallery/query partition are used for the floating-point and quantized evaluations.
	
	\subsubsection{Evaluation protocol and reporting conventions}
	
	We report biometric utility and encrypted-system cost. Utility metrics are top-1 identification accuracy, EER, and FAR/FRR at the operating threshold of the corresponding model and quantization setting. System metrics include encrypted matching latency, encrypted storage, PVSC proof-generation time, verification time on the opened proof collection, and integrated pipeline runtime. Here, ``verification'' denotes verifier execution on the opened PVSC proof batch, measured in the clear; in the corrected design it runs inside the committee's threshold boundary (Section~\ref{sec:PVSC_ver}); contract-side gas, calldata, and storage measurements are reported separately for the session-binding, coverage, confirmation, and finalization logic. Since prior systems are unavailable under matched implementations and parameters, the primary same-platform reference point is the $n=1$ configuration of the same codebase, which preserves the same row-wise encrypted matching pipeline while disabling SCMV multi-value binding; Supplementary Sec.~VI reports an indicative cross-paper comparison for context only.
	
	\subsection{Utility Under Feature Quantization}
	\label{accuracy}
	
	We first test whether the integer embedding used by SCMV preserves the biometric utility needed for downstream identification. To isolate this effect, we compare floating-point and quantized pipelines on the same identity set and gallery/query partition.
	
	Fig.~\ref{fig:accuracy_tradeoff} reports the identification-accuracy loss with respect to the floating-point baseline together with the corresponding EER under different quantization factors $\Gamma$, and Table~\ref{tab:compressed_performance} lists representative absolute values. For both FaceNet ($L=128$) and MobileFaceNet ($L=512$), most degradation is concentrated at very small $\Gamma$ and quickly becomes negligible as $\Gamma$ increases. We therefore use $\Gamma=128$ for FaceNet and $\Gamma=512$ for MobileFaceNet in the subsequent encrypted experiments. At these operating points, BioZKFHE achieves 99.47\% and 98.30\% identification accuracy, respectively, while the corresponding EERs remain essentially unchanged from the floating-point baselines.
	
	In short, both models remain within the 1\% accuracy-loss tolerance once $\Gamma$ is sufficiently large, supporting the claim that integer-domain matching is effectively utility preserving in BioZKFHE.
	
	\begin{figure}[!t]
		\centering
		\includegraphics[width=0.9\linewidth]{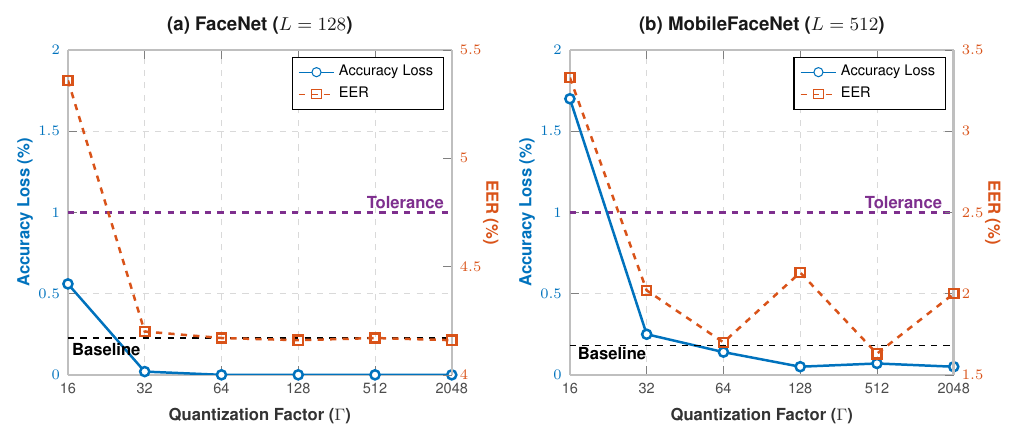}
		\caption{Utility under integer quantization: accuracy loss and EER versus $\Gamma$.}
		\label{fig:accuracy_tradeoff}
	\end{figure}
	
	\begin{table}[!t]
		\centering
		\caption{Identification performance under representative quantization factors $\Gamma$.}
		\label{tab:compressed_performance}
		\scriptsize
		\setlength{\tabcolsep}{2.5pt}
		\begin{tabular}{c c cccc}
			\toprule
			\textbf{Model} & \textbf{$\Gamma$} & \textbf{Acc. (\%)} & \textbf{EER (\%)} & \textbf{FAR (\%)} & \textbf{FRR (\%)} \\
			\midrule
			\multirow{3}{*}{\shortstack{\textbf{FaceNet} \\ ($L=128$)}}
			& Cleartext & 99.47 & 4.17 & 0.20 & 5.76 \\
			& 16        & 98.91 & 5.36 & 0.36 & 12.87 \\
			& 128       & 99.47 & 4.16 & 0.20 & 5.84 \\
			\midrule
			\multirow{3}{*}{\shortstack{\textbf{MobileFaceNet} \\ ($L=512$)}}
			& Cleartext & 98.37 & 1.68 & 1.80 & 1.47 \\
			& 16        & 96.67 & 3.33 & 2.97 & 3.70 \\
			& 512       & 98.30 & 1.63 & 1.77 & 1.83 \\
			\bottomrule
		\end{tabular}
	\end{table}
	
	\subsection{Encrypted Matching Scalability and Storage Cost}
	\label{efficiency}
	
	We next test whether SCMV delivers the predicted structural gain as the gallery size $N$ grows. Since the same encrypted matching primitive underlies enrollment-time uniqueness checking and login-time identification, these results characterize the common encrypted matching layer. Using the matched-codebase $n=1$ setting as the same-platform reference point, we examine the key SCMV claim: one block covers $nd$ templates, so both encrypted storage and processed-block count should scale approximately with $\lceil N/(nd)\rceil$.
	
	We use $\Gamma=128$ for FaceNet and $\Gamma=512$ for MobileFaceNet. In the reported implementation, the corresponding SCMV bases are set to $T=35729$ and $T=547715$, respectively. We use polynomial degree $d=8192$ throughout the scalability experiments. For a given SCMV binding factor $n$, the plaintext modulus $p$ is selected to satisfy the recoverability condition; the noise margin of each set is assessed in Section~\ref{sec:noise_analysis} and Table~\ref{tab:noise_screen}. In our implementation, $p$ is set to a 20-bit, 31-bit, and 46-bit plaintext modulus for FaceNet with $n=1,2,3$, respectively. The corresponding ciphertext moduli are represented in RNS form with coefficient-modulus decompositions 110-bit $\{30,25,25,30\}$, 150-bit $\{45,30,30,45\}$, and 170-bit $\{45,40,40,45\}$. For MobileFaceNet, $p$ is set to a 21-bit and 40-bit plaintext modulus for $n=1,2$, with 120-bit $\{30,30,30,30\}$ and 160-bit $\{45,35,35,45\}$ coefficient-modulus decompositions, respectively.
	
	In $\mathsf{SCMV.Identify}$, each similarity score is computed using one ciphertext--ciphertext multiplication per embedding dimension, and the resulting products are then accumulated through ciphertext additions. Since SCMV avoids query-time rotations and the current evaluation path has multiplicative depth $1$, all experiments are carried out at a fixed modulus level and do not invoke modulus switching.
	
	Fig.~\ref{fig:scalability} summarizes the resulting scalability trends. As expected, both matching latency and encrypted storage grow approximately linearly with the number of SCMV database blocks. Increasing the binding factor $n$ reduces the number of encrypted blocks that must be stored and scanned, which directly lowers both latency and storage.
	
	The per-$n$ security/performance trade-off is reported in Supplementary Material, Sec.~IV. At fixed $d=8192$, the smallest margins are about 135 bits for FaceNet with $n=3$ and about 147 bits for MobileFaceNet with $n=2$ (rough Core-SVP lower-bound estimates~\cite{albrecht2015concrete,dualattack2}, reported to more digits in the supplementary table only for reproducibility); meanwhile, at $N=10^6$, the corresponding block counts drop from 123 to 41 and from 123 to 62, respectively.
	
These results support the main SCMV claim: relative to the same-platform reference point $n=1$, increasing the binding factor reduces both ciphertext expansion and the number of processed encrypted blocks without changing the logical matching rule. Supplementary Sec.~IV separates horizontal packing, the row-wise vertical baseline, and SCMV. At $N=100$k, the FaceNet horizontal baseline uses 1{,}563 blocks and 10{,}941 rotations, whereas SCMV with $n=3$ uses 5 blocks, no online rotations, 543.4~ms latency, and 188.51~MiB storage; for MobileFaceNet, SCMV with $n=2$ reduces the scan from 6{,}250 horizontal blocks to 7 SCMV blocks. The same section reports the per-$n$ security/performance table and lattice-estimator settings. An indicative cross-paper comparison is provided in Sec.~VI of the Supplementary Material, but our primary quantitative reference remains the matched-codebase $n=1$ baseline because implementation environments and evaluation protocols differ substantially across systems.
	
	\begin{figure}[!t]
		\centering
		\includegraphics[width=0.9\linewidth]{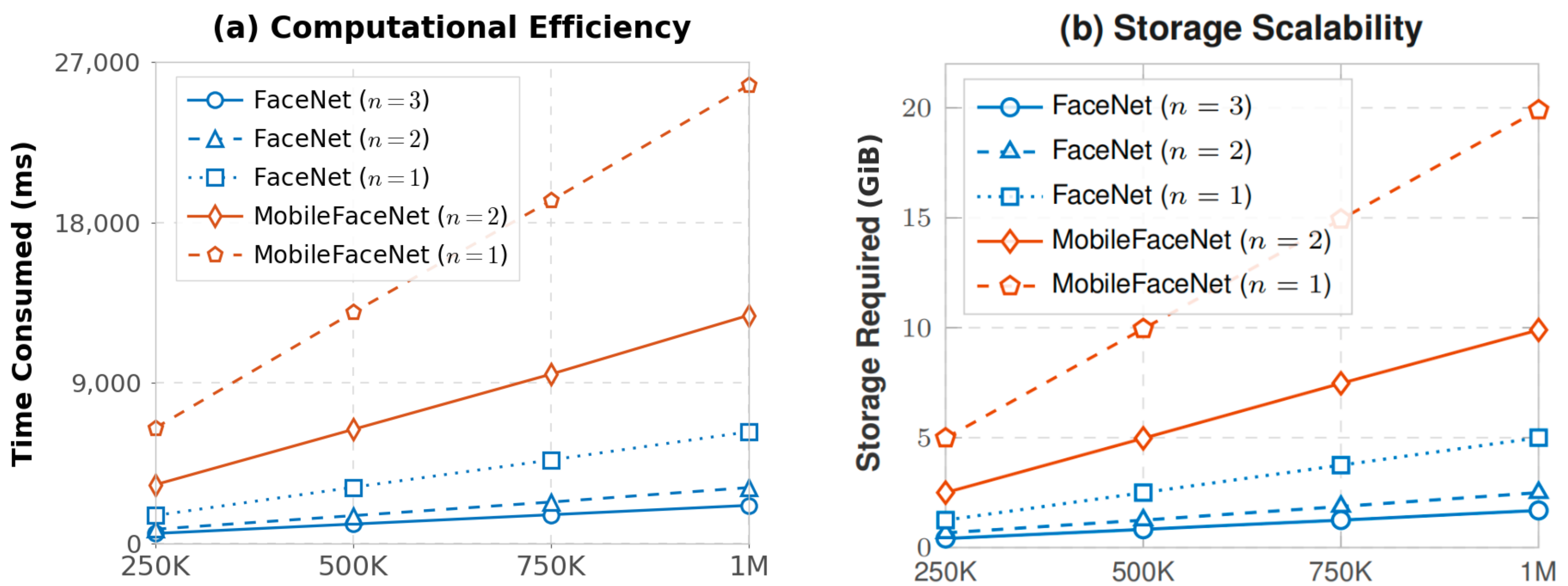}
		\caption{SCMV scalability: latency and encrypted storage versus database size $N$.}
		\label{fig:scalability}
	\end{figure}
	
	\subsection{PVSC Overhead for Encrypted-Output Integrity}
	\label{sec:pvsc_eval}
	
	We now test whether PVSC makes encrypted-output integrity practical at the ciphertext-output layer. Here, ``verification'' refers to execution of the verifier logic on the \emph{opened} proof batch after committee-side threshold opening, measured in the clear, rather than to blockchain-specific gas, calldata, or state-storage cost. Accordingly, the results below characterize protocol-level verification cost and do not include threshold or MPC execution of the verifier. 
	
	
	\subsubsection{Microbenchmark of decomposed proof templates}
	
	Fig.~\ref{fig:pvsc_benchmarks} reports the proof-generation and verification time for decomposed BGV addition and multiplication templates under three chunk sizes, $d_{\mathrm{ZK}} \in \{128,256,512\}$. For both operations, the proof-generation cost grows approximately linearly with $d_{\mathrm{ZK}}$: for example, the multiplication-template cost increases from 46.9~ms at $d_{\mathrm{ZK}}=128$ to 170.8~ms at $d_{\mathrm{ZK}}=512$. In contrast, the verification time remains small, ranging from 0.85 to 1.47~ms across all tested settings. This behavior matches the intended decomposition trade-off: smaller chunk sizes reduce the cost of an individual proof instance but increase the number of instances, whereas larger chunk sizes reduce fragmentation at the cost of more expensive instances.
	
	\subsubsection{Parallel proof generation}
	
	Fig.~\ref{fig:pvec_efficiency} compares single-process and multi-process proof generation for three representative workloads, testing whether Double-CRT decomposition exposes enough parallelism to offset PVSC cost. The benefit is substantial and consistent across all tested cases: the total proof-generation time decreases from 2074~s to 17~s for the 128-dimensional workload, from 5988~s to 49~s for the 256-dimensional workload, and from 10655~s to 85~s for the 512-dimensional workload, corresponding to about $122\times$, $122\times$, and $125\times$ speedups, respectively. The multi-process configuration uses $128$ concurrent proof-generation processes mapped onto the server's $144$ hardware threads, so these speedups stay below the deployed parallelism. Verification remains below 1.1~s and is omitted from Fig.~\ref{fig:pvec_efficiency} for readability.
	
	The key takeaway is not the exact speedup factor of one platform, but that the Double-CRT decomposition exposes many weakly dependent proof tasks, making proof generation highly parallelizable. It also determines footprint: encrypted/opened proof size and verifier-side public input are linear in $K_m = n_q n_{\mathrm{chunk}}(2L-1)$ per block, while committee-side opening communication is additionally linear in $t$. Since each block submission adds only one output ciphertext and a constant-size session-binding tuple, the proof object dominates the non-ciphertext footprint.
	
	\begin{figure}[t]
		\centering
		\includegraphics[width=0.9\linewidth]{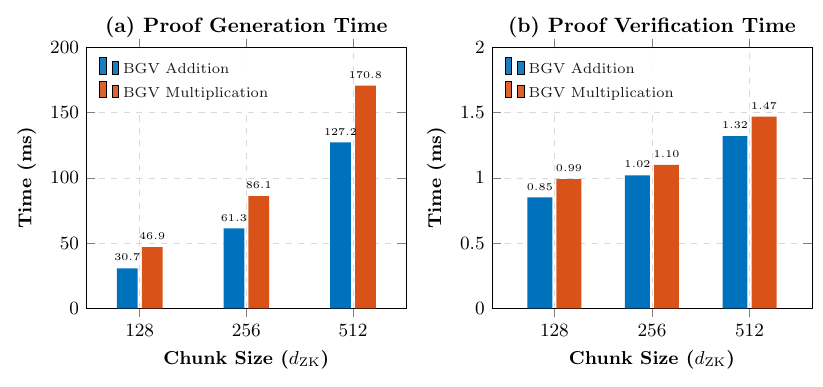}
		\caption{PVSC microbenchmarks for decomposed proof templates under different chunk sizes $d_{\mathrm{ZK}}$.}
		\label{fig:pvsc_benchmarks}
	\end{figure}
	
	\begin{figure}[!t]
		\centering
		\includegraphics[width=0.75\linewidth]{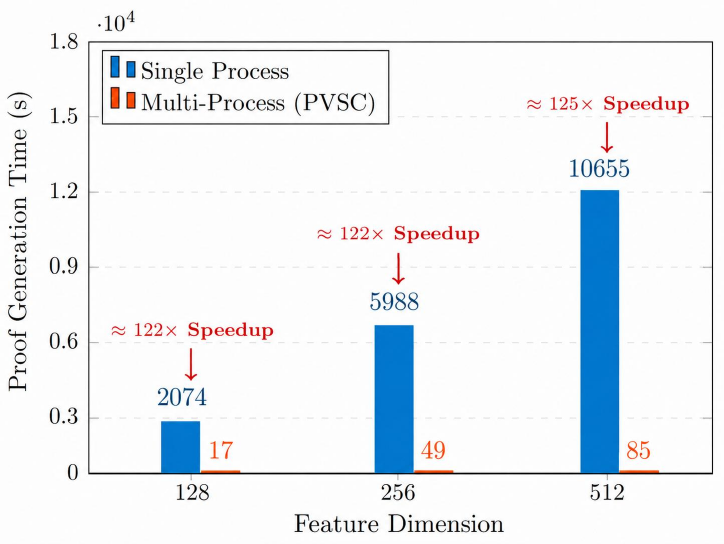}
		\caption{PVSC proof-generation time under single-process and multi-process execution for representative workloads.}
		\label{fig:pvec_efficiency}
	\end{figure}
	
	\subsubsection{Proof footprint and session accumulation}
	\label{sec:comm_cost}
	
	BioZKFHE's verification footprint is driven by opened local proofs. With chunk size $d_{\mathrm{ZK}}$, one block has $n_{\mathrm{chunk}}=\lceil d/d_{\mathrm{ZK}}\rceil$ chunks per RNS limb and $K_m=n_q n_{\mathrm{chunk}}(2L-1)$ counted arithmetic/output instances: $n_qn_{\mathrm{chunk}}L$ multiplication, $n_qn_{\mathrm{chunk}}(L-2)$ accumulation, and $n_qn_{\mathrm{chunk}}$ output instances. The $L-2$ term reflects that $\bm{a}_1=\bm{p}_1$ is fused into multiplication and $\bm{a}_L=\bm{c}^{\mathrm{sim}}_m$ is certified by the output layer (Supplementary Material, Sec.~V). The full opened batch adds one block-binding instance, so it has $K_m+1$ verified instances, while a session over $N$ templates has $M=\lceil N/(nd)\rceil$ blocks, $MK_m$ counted arithmetic/output instances, and $M$ block-binding instances.
	
	\begin{table}[!t]
		\centering
		{
		\caption{Representative PVSC proof footprint at $d_{\mathrm{ZK}}=512$ under the prototype $n_q=3$ proof-channel setting. Session sizes use $N=100$k.}
		\label{tab:pvsc_footprint_summary}
		\scriptsize
		\setlength{\tabcolsep}{2.0pt}
		\renewcommand{\arraystretch}{1.06}
		\begin{tabular}{lrrrrr}
			\toprule
			\textbf{Profile} & \textbf{$M$} & \textbf{$K_m$} & \textbf{Enc./blk} & \textbf{Open/blk} & \textbf{Open/session} \\
			\midrule
			FaceNet $n=1$ & 13 & 12{,}240 & 668.07~MiB & 5.04~MiB & 65.56~MiB \\
			FaceNet $n=3$ & 5 & 12{,}240 & 896.67~MiB & 8.12~MiB & 40.62~MiB \\
			MobileFaceNet $n=1$ & 13 & 49{,}104 & 3.22~GiB & 19.48~MiB & 253.25~MiB \\
			MobileFaceNet $n=2$ & 7 & 49{,}104 & 3.51~GiB & 32.59~MiB & 228.15~MiB \\
			\bottomrule
		\end{tabular}
		}
		\vspace{-4pt}
	\end{table}

\begin{figure}[!t]
	\centering
	\includegraphics[width=0.9\linewidth]{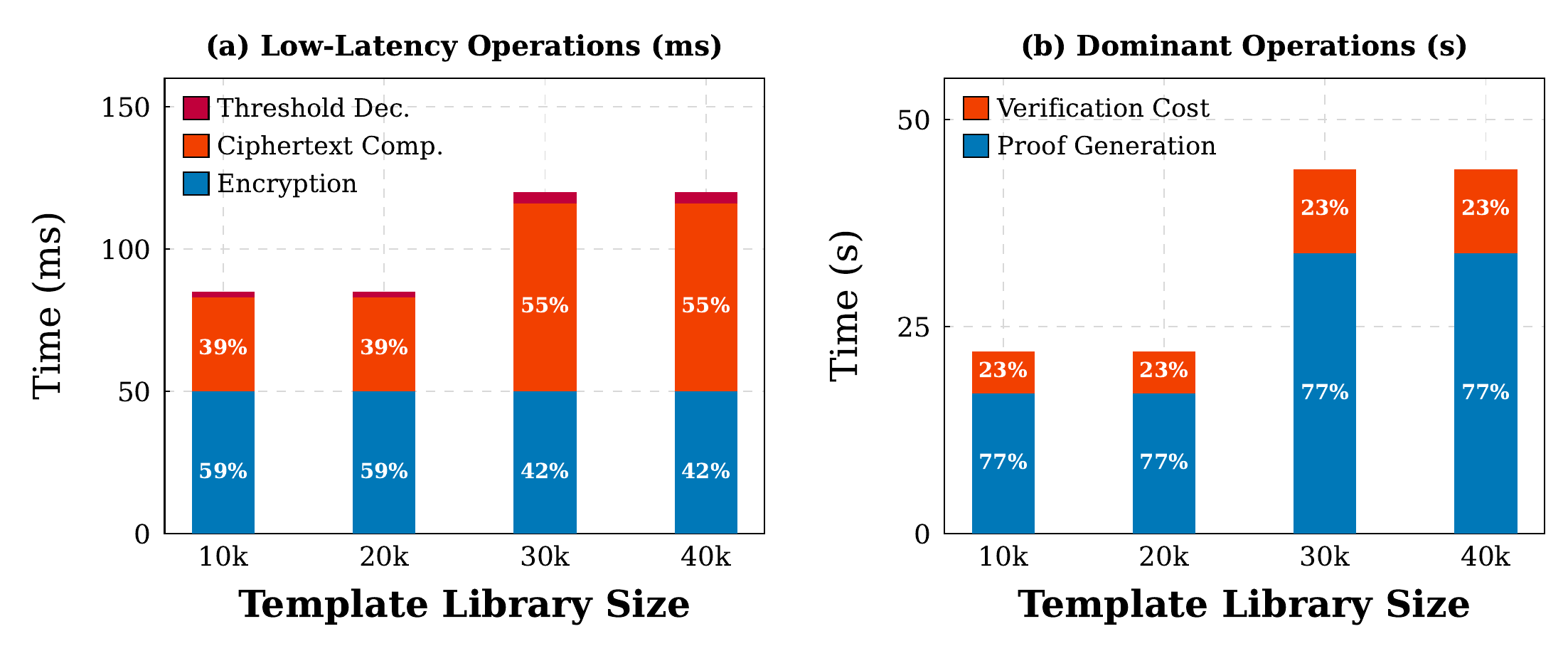}
	\caption{Integrated overhead breakdown versus template-library size.}
	\label{fig:end_to_end_test}
\end{figure}
	
	Table~\ref{tab:pvsc_footprint_summary} reports verified proof-object struct sizes for the implemented $n_q=3$ prototype. The soundness accounting in Supplementary Material, Sec.~V uses the conservative four-limb count $n_q=4$; because proof count and footprint are linear in $n_q$, a four-limb deployment scales by about $4/3$ from the measured rows. Since $K_m$ depends only on $(n_q,n_{\mathrm{chunk}},L)$, profiles differing only in the SCMV binding factor $n$ share $K_m$ but differ in $M$ and ciphertext parameters. Detailed rows for $d_{\mathrm{ZK}}\in\{128,256,512\}$, public-input size, and committee-opening communication appear in Supplementary Material, Sec.~V.
	
	\subsection{Integrated Pipeline Overhead Breakdown}
	
	We next evaluate the end-to-end proof-verified identification pipeline, combining encrypted matching, opened-proof verification, and threshold recovery of the similarity result.
	
	Fig.~\ref{fig:end_to_end_test} reports the integrated pipeline under FaceNet with $n=3$, $d=8192$, $d_{\mathrm{ZK}}=512$, and the same threshold setting as Table~\ref{tab:pvsc_footprint_summary}. Fig.~\ref{fig:end_to_end_test}(a) shows that the low-latency stages remain small across all tested settings: encryption stays around 50~ms, ciphertext computation rises from about 33~ms to 66~ms as the template library grows from 10k to 40k, and committee-side recovery remains minor. In contrast, Fig.~\ref{fig:end_to_end_test}(b) shows that proof-related costs dominate: the total dominant proof-related runtime is about 22~s for 10k--20k templates and about 44~s for 30k--40k templates, with proof generation contributing roughly 77\% and verification-side costs about 23\%. These results support feasibility at 10k--40k scale while identifying proof generation as the main remaining bottleneck.

	\section{Conclusion}
	\label{sec:conclusion}
	
	BioZKFHE combines SCMV for efficient encrypted $1{:}N$ matching with PVSC for proof-verified outsourced similarity evaluation, thereby supporting scalable committee-attested auditability for biometric identification. Experiments on FaceNet and MobileFaceNet show near-lossless utility, up to 67\% reduction in encrypted storage, and feasible end-to-end execution at 10k--40k scale with about 22--44~s runtime.
	
	In BioZKFHE, auditability is intentionally separated from threshold-governed release: PVSC certifies, under the committee threshold assumption, the encrypted similarity outputs accepted after threshold opening, whereas final decryption and decision release remain committee-governed. Within this model, BioZKFHE protects templates, queries, and intermediate homomorphic states against the Compute Node, public observers, and any coalition of fewer than $t$ committee members, while the main remaining opportunities lie in publicly verifiable proving that removes the committee from the verification path, lower-footprint proving, smaller proof/public-input footprints, and lighter threshold-release mechanisms.

	\section*{Acknowledgment}
	The authors would like to thank the editor and the anonymous reviewers for their valuable comments and suggestions, which helped improve the quality of this paper. The work of T.~Wang, X.~Wu, and S.~Zhang is supported by Solo Research Limited.
	\vspace{-8pt}

	\def\BIBdecl{\setlength{\itemsep}{-1.5pt}\setlength{\parsep}{0pt}}
\input{main.bbl}\input{supplementary.tex}

\end{document}

%% file: supplementary.tex
\clearpage
\begingroup
\makeatletter
\renewcommand{\@biblabel}[1]{[s#1]}
\makeatother
\renewcommand{\citeform}[1]{s#1}
\def\BIBdecl{\relax}
\setcounter{section}{0}
\setcounter{subsection}{0}
\setcounter{subsubsection}{0}
\setcounter{equation}{0}
\setcounter{figure}{0}
\setcounter{table}{0}
\renewcommand{\theHsection}{supp.\arabic{section}}
\renewcommand{\theHsubsection}{supp.\arabic{section}.\arabic{subsection}}
\renewcommand{\theHequation}{supp.\arabic{equation}}
\renewcommand{\theHfigure}{supp.\arabic{figure}}
\renewcommand{\theHtable}{supp.\arabic{table}}
\markboth{Supplementary Material for BioZKFHE}{}
\twocolumn[
\begin{center}
{\LARGE Supplementary Material for\\[2pt]
``BioZKFHE: Scalable Encrypted Biometric Identification via Verifiable Homomorphic Similarity Evaluation''\par}
\vspace{1.2em}
{\large Rundong~Xin, Taotao~Wang, \IEEEmembership{Member,~IEEE}, Xiaoxiao~Wu, \IEEEmembership{Senior~Member,~IEEE},\\
Weizhi~Meng, \IEEEmembership{Senior~Member,~IEEE}, Shengli~Zhang, \IEEEmembership{Senior~Member,~IEEE}, and Shui~Yu, \IEEEmembership{Fellow,~IEEE}\par}
\end{center}
\vspace{1.5em}
]
\thispagestyle{plain}

This supplementary material provides self-contained cryptographic background for the main paper. Its purpose is to help readers from biometric systems, applied cryptography, and secure computation follow the technical design of BioZKFHE without assuming prior familiarity with fully homomorphic encryption, lattice-based zero-knowledge proofs, or threshold cryptography. The presentation is aligned with the notation and interfaces used in the main text, and focuses on the specific cryptographic abstractions needed for SCMV and PVSC rather than on the most general formulations of the underlying primitives. Besides the cryptographic background in Secs.~\ref{supp:app:bgv}--\ref{supp:app:td}, this supplement also includes a structural comparison with conventional packing layouts in Sec.~\ref{supp:sec:SCMV_comparison}, the concrete descriptor-level PVSC formulation in Sec.~\ref{supp:supp:pvsc_formulation}, and an indicative cross-paper comparison in Sec.~\ref{supp:supp:comparison}.

\section{BGV Fully Homomorphic Encryption}
\label{supp:app:bgv}

Fully Homomorphic Encryption (FHE) is a foundational technology in post-quantum cryptography because it simultaneously offers security based on hard lattice problems and supports direct computation on ciphertexts without decryption. By enabling addition and multiplication on encrypted data, FHE resolves the traditional tension between data utility and privacy protection. In BioZKFHE, we adopt the Brakerski--Gentry--Vaikuntanathan (BGV) scheme~\cite{supp:bgv2014} as the underlying FHE layer. BGV is built from the Ring-Learning With Errors (RLWE) assumption and, in practical implementations, relies on batching, relinearization, and the Double-CRT representation to make homomorphic computation efficient~\cite{supp:mathCanToolkit,supp:halevi2019improved,supp:homencstandard2024,supp:sealcrypto}. This appendix provides the BGV background needed by the main paper. In particular, we introduce the core interfaces of the adopted BGV model, namely encoding/decoding, key generation, encryption/decryption, and homomorphic addition/multiplication, and then explain how the Double-CRT representation supports efficient evaluation.

\textbf{Mathematical Notations and Structures}: All operations are performed over the cyclotomic ring $\mathcal{R} := \mathbb{Z}[x]/(x^d+1)$, where $d$ is a power of two. The plaintext and ciphertext rings are $\mathcal{R}_p := \mathcal{R}/p\mathcal{R}$ and $\mathcal{R}_q := \mathcal{R}/q\mathcal{R}$, respectively, where $p$ is the plaintext modulus and $q$ is the ciphertext modulus, with $p \ll q$. Equivalently, one may write $\mathcal{R}_p \cong \mathbb{Z}_p[x]/\langle x^d+1\rangle$ and $\mathcal{R}_q \cong \mathbb{Z}_q[x]/\langle x^d+1\rangle$. In the main paper, the ciphertext modulus is further decomposed as $q=\prod_{j=1}^{n_q} q^{(j)}$, where the moduli $\{q^{(j)}\}_{j=1}^{n_q}$ are pairwise co-prime. The discrete Gaussian distribution over $\mathbb{Z}$ with standard deviation $\sigma$ is denoted by $D_{\mathbb{Z},\sigma}$. A BGV ciphertext is written as $\bm{c}=(c^{(0)},c^{(1)})\in\mathcal{R}_q^2$, the public key as $\mathsf{pk}_{\mathrm{HE}}$, the secret key as $\mathsf{sk}_{\mathrm{HE}}$, and the public evaluation-key set as $\mathsf{evk}_{\mathrm{HE}}$. Throughout the supplementary material and the main paper, the symbols $\oplus$ and $\otimes$ denote BGV homomorphic addition and multiplication, respectively.

\subsection{Encoding and Decoding}

This part introduces the encoding and decoding interfaces of the BGV scheme, which convert between raw message vectors and plaintext polynomials.

\textbf{Encoding}: In the batched form of BGV, a raw message is represented by a slot vector $\bm{m}=(m_1,\dots,m_d)\in\mathbb{Z}_p^d$. When $p \equiv 1 \pmod{2d}$, the polynomial $x^d+1$ splits completely over $\mathbb{Z}_p$, so there exist pairwise distinct roots $\xi_1,\dots,\xi_d\in\mathbb{Z}_p$ such that $x^d+1=\prod_{j=1}^{d}(x-\xi_j)$ over $\mathbb{Z}_p$. This induces a CRT isomorphism $\mathcal{R}_p \cong \prod_{j=1}^{d}\mathbb{Z}_p$.

The encoding operation maps the slot vector $\bm{m}$ to the unique plaintext polynomial $v(x)\in\mathcal{R}_p$ of degree smaller than $d$ satisfying $v(\xi_j)=m_j$ for all $j\in\{1,\dots,d\}$. Equivalently, one may write 
\begin{equation}
v(x)=\mathsf{BGV.Encode}(\bm{m})=\sum_{j=1}^{d} m_j\,\ell_j(x),
\end{equation}
where $\ell_j(x)\in\mathcal{R}_p$ is the $j$-th Lagrange basis polynomial associated with the roots $\{\xi_1,\dots,\xi_d\}$. Thus, encoding packs $d$ plaintext slots into one plaintext polynomial.

\textbf{Decoding}: The decoding operation is the inverse of encoding. Given a plaintext polynomial $v(x)\in\mathcal{R}_p$, the decoder recovers the slot vector by evaluating $v(x)$ at the CRT points:
\begin{equation}
\bm{m}=\mathsf{BGV.Decode}(v)=\bigl(v(\xi_1),v(\xi_2),\dots,v(\xi_d)\bigr).
\end{equation}
Hence encoding and decoding are inverse CRT maps between the coefficient representation of a plaintext polynomial and its slot representation.

\textbf{Batch Processing}: The above encoding/decoding mechanism is usually referred to as batching. It enables a single ciphertext to carry $d$ plaintext slots and lets homomorphic operations act on all slots in parallel. This SIMD capability is the basis for efficient encrypted matching in BioZKFHE. The SCMV design in the main paper starts from this batching interface, but further exploits the representational range of each plaintext entry by binding multiple values into one entry through base-$T$ expansion.

\subsection{Key Generation, Encryption and Decryption}

This part presents the key-generation, encryption, and decryption interfaces. Before introducing them, we define polynomial addition, multiplication, and scalar multiplication in the ciphertext ring. Given $a(x)=\sum_{i=0}^{d-1} a_i x^i \in \mathcal{R}_q$, $b(x)=\sum_{i=0}^{d-1} b_i x^i \in \mathcal{R}_q$, and a scalar $k\in\mathbb{Z}_q$, these operations are
\begin{equation}
\label{supp:eq:supp-ring-ops}
\begin{aligned}
a(x)+b(x) &= \sum_{i=0}^{d-1}(a_i+b_i)x^i,\\
a(x)\cdot b(x) &= \left(\sum_{i=0}^{d-1}\sum_{j=0}^{d-1} a_i b_j x^{i+j}\right)\bmod (x^d+1,\;q),\\
k\cdot a(x) &= \sum_{i=0}^{d-1}(k a_i)x^i,
\end{aligned}
\end{equation}
where polynomial multiplication is followed by reduction modulo $x^d+1$, and all coefficients are reduced modulo $q$.

\textbf{Key Generation}: The BGV scheme first generates a public/secret key pair together with public evaluation keys. We write the key-generation interface as
\begin{equation}
\label{supp:eq:supp-keygen}
(\mathsf{pk}_{\mathrm{HE}},\mathsf{sk}_{\mathrm{HE}},\mathsf{evk}_{\mathrm{HE}})
\leftarrow
\mathsf{BGV.KeyGen}(1^\lambda,d,p,q).
\end{equation}
The secret key is a small polynomial $s(x)\in\mathcal{R}_q$, usually sampled coefficientwise from a small set such as $\{-1,0,1\}$. The public key is a pair $\mathsf{pk}_{\mathrm{HE}}=(a(x),b(x))$, where $a(x)\leftarrow\mathcal{R}_q$ is sampled uniformly and
\begin{equation}
\label{supp:eq:supp-pk}
b(x)=-a(x)\cdot s(x)+p\,e_{\mathrm{pk}}(x)\pmod q,
\end{equation}
with $e_{\mathrm{pk}}(x)\in\mathcal{R}_q$ sampled from a small-error distribution such as $D_{\mathbb{Z},\sigma}$. In the current BioZKFHE evaluation path, the only evaluation key actually used is the public relinearization key set, written in the main paper as $\mathsf{evk}_{\mathrm{HE}}=\{\mathsf{rk}_{\mathrm{HE}}\}$.

\textbf{Encryption}: The encryption function takes a plaintext polynomial $v(x)\in\mathcal{R}_p$ and the public key $\mathsf{pk}_{\mathrm{HE}}=(a(x),b(x))$ as input, and outputs a ciphertext $\bm{c}=(c^{(0)}(x),c^{(1)}(x))\in\mathcal{R}_q^2$. We write this as
\begin{equation}
\label{supp:eq:supp-enc-interface}
\bm{c}=\mathsf{BGV.Encrypt}(v,\mathsf{pk}_{\mathrm{HE}}).
\end{equation}
To encrypt, the scheme samples a short random polynomial $u(x)$ and small error polynomials $e_0(x),e_1(x)\in\mathcal{R}_q$, and computes
\begin{equation}
\label{supp:eq:supp-enc}
\begin{aligned}
c^{(0)}(x) &= v(x)+b(x)\cdot u(x)+p\,e_0(x)\pmod q,\\
c^{(1)}(x) &= a(x)\cdot u(x)+p\,e_1(x)\pmod q.
\end{aligned}
\end{equation}
The random masking term $u(x)$ together with the error terms ensures semantic security under the RLWE assumption.

\textbf{Decryption}: The decryption function takes a ciphertext $\bm{c}=(c^{(0)}(x),c^{(1)}(x))\in\mathcal{R}_q^2$ and the secret key $\mathsf{sk}_{\mathrm{HE}}=s(x)$, and outputs the plaintext polynomial $v(x)\in\mathcal{R}_p$:
\begin{equation}
\label{supp:eq:supp-dec-interface}
v=\mathsf{BGV.Decrypt}(\bm{c},\mathsf{sk}_{\mathrm{HE}}).
\end{equation}
The key intermediate quantity is
\begin{equation}
\label{supp:eq:supp-dec-pre}
v^{*}(x)=c^{(0)}(x)+c^{(1)}(x)\cdot s(x)\pmod q.
\end{equation}
Substituting the encryption equations gives
\begin{equation}
\label{supp:eq:supp-dec-expand}
v^{*}(x)=v(x)+p\bigl(e_{\mathrm{pk}}(x)\,u(x)+e_0(x)+e_1(x)\,s(x)\bigr)\pmod q.
\end{equation}
Thus, when the effective error term remains sufficiently small relative to $q$, reducing $v^{*}(x)$ modulo $p$ recovers the plaintext polynomial:
\begin{equation}
\label{supp:eq:supp-dec-final}
v(x)=v^{*}(x)\bmod p.
\end{equation}
This abstract decryption form is the one used by the main paper's correctness and noise analysis.

\subsection{Homomorphic Addition and Multiplication}
\label{supp:appendix:computation}

This part first introduces the Double-CRT representation used in practical BGV computation, and then explains how homomorphic addition and homomorphic multiplication are carried out under that representation.

\textbf{Double-CRT Representation}: The BGV scheme employs the Double-CRT representation as a key optimization technique to improve computational efficiency and reduce memory usage during homomorphic evaluation~\cite{supp:halevi2019improved,supp:homencstandard2024,supp:sealcrypto}. It combines two decomposition layers.
\begin{enumerate}
    \item \textit{Modulus decomposition}: In the first CRT layer, the large ciphertext modulus is decomposed as $q=\prod_{j=1}^{n_q} q^{(j)}$, where the factors $\{q^{(j)}\}$ are pairwise co-prime. This induces the ring decomposition
\begin{equation}
\label{supp:eq:supp-rns}
\mathcal{R}_q \cong \prod_{j=1}^{n_q} \mathcal{R}_{q^{(j)}},
\qquad
\mathcal{R}_{q^{(j)}} := \mathcal{R}/q^{(j)}\mathcal{R}.
\end{equation}
Therefore, each ciphertext component over $\mathcal{R}_q$ is represented by its residues over the smaller rings $\{\mathcal{R}_{q^{(j)}}\}$, enabling arithmetic to be performed in a Residue Number System (RNS).

\item \textit{Polynomial decomposition via NTT}: In the second CRT layer, each small ring $\mathcal{R}_{q^{(j)}}$ is represented in the Number Theoretic Transform (NTT) domain, provided that $q^{(j)} \equiv 1 \pmod{2d}$. If $a(x)\in\mathcal{R}_{q^{(j)}}$ is written in coefficient form, the transform $\mathsf{NTT}_{q^{(j)}}(a)$ maps it to its point-value representation at suitable roots of unity. Under this representation, polynomial multiplication becomes elementwise multiplication:
\begin{equation}
\label{supp:eq:supp-ntt}
\mathsf{NTT}_{q^{(j)}}(a\cdot b)
=
\mathsf{NTT}_{q^{(j)}}(a)\odot \mathsf{NTT}_{q^{(j)}}(b),
\end{equation}
where $\odot$ denotes the Hadamard product. This is exactly the implementation view exploited by PVSC in the main paper.
\end{enumerate}

\textbf{Homomorphic Addition}: Consider two ciphertexts $\bm{c}_1=(c_1^{(0)},c_1^{(1)})\in\mathcal{R}_q^2$ and $\bm{c}_2=(c_2^{(0)},c_2^{(1)})\in\mathcal{R}_q^2$. Under the first CRT layer, each ciphertext component is decomposed into residues $c_{r,j}^{(t)}\in\mathcal{R}_{q^{(j)}}$, where $r\in\{1,2\}$, $t\in\{0,1\}$, and $j\in\{1,\dots,n_q\}$. Under the second CRT layer, each residue polynomial is represented by its NTT vector $\widehat{\bm{c}}_{r,j}^{(t)}\in \mathbb{Z}_{q^{(j)}}^{d}$. The homomorphic addition operation $\bm{c}_{\mathrm{add}}=\bm{c}_1\oplus\bm{c}_2$ is then performed componentwise:
\begin{equation}
\label{supp:eq:supp-hadd}
\widehat{\bm{c}}_{\mathrm{add},j}^{(t)}
=
\widehat{\bm{c}}_{1,j}^{(t)}+\widehat{\bm{c}}_{2,j}^{(t)},
\qquad
\forall\, t\in\{0,1\},\ j\in\{1,\dots,n_q\}.
\end{equation}
Applying the inverse NTT and CRT recombination reconstructs the resulting ciphertext components over $\mathcal{R}_q$.

\textbf{Homomorphic Multiplication}: The multiplication of two ciphertexts $\bm{c}_1=(c_1^{(0)},c_1^{(1)})\in\mathcal{R}_q^2$ and $\bm{c}_2=(c_2^{(0)},c_2^{(1)})\in\mathcal{R}_q^2$ first produces a three-component ciphertext $\widetilde{\bm{c}}_{\mathrm{prod}}=(\tilde c_{\mathrm{prod}}^{(0)},\tilde c_{\mathrm{prod}}^{(1)},\tilde c_{\mathrm{prod}}^{(2)})\in\mathcal{R}_q^3$, where
\begin{equation}
\label{supp:eq:supp-hmul-components}
\begin{aligned}
\tilde c_{\mathrm{prod}}^{(0)}(x) &= c_1^{(0)}(x)\cdot c_2^{(0)}(x),\\
\tilde c_{\mathrm{prod}}^{(1)}(x) &= c_1^{(0)}(x)\cdot c_2^{(1)}(x) + c_1^{(1)}(x)\cdot c_2^{(0)}(x),\\
\tilde c_{\mathrm{prod}}^{(2)}(x) &= c_1^{(1)}(x)\cdot c_2^{(1)}(x).
\end{aligned}
\end{equation}
Under the Double-CRT representation, each polynomial is represented by NTT vectors. And thus, according to Eq.~\eqref{supp:eq:supp-ntt}, the multiplications in Eq.~\eqref{supp:eq:supp-hmul-components} are performed component-wise as
\begin{equation}
\label{supp:eq:supp-hmul}
\begin{aligned}
\widehat{\bm{c}}_{\mathrm{prod},j}^{(0)} &= \widehat{\bm{c}}_{1,j}^{(0)} \odot \widehat{\bm{c}}_{2,j}^{(0)},\\
\widehat{\bm{c}}_{\mathrm{prod},j}^{(1)} &= \widehat{\bm{c}}_{1,j}^{(0)} \odot \widehat{\bm{c}}_{2,j}^{(1)} + \widehat{\bm{c}}_{1,j}^{(1)} \odot \widehat{\bm{c}}_{2,j}^{(0)},\\
\widehat{\bm{c}}_{\mathrm{prod},j}^{(2)} &= \widehat{\bm{c}}_{1,j}^{(1)} \odot \widehat{\bm{c}}_{2,j}^{(1)},
\end{aligned}
\qquad
\forall\, j\in\{1,\dots,n_q\}.
\end{equation}
After inverse NTT and CRT recombination, one obtains the three-component ciphertext $\widetilde{\bm{c}}_{\mathrm{prod}}$ over $\mathcal{R}_q$. In practical BGV implementations, this product ciphertext is then relinearized with the public relinearization key $\mathsf{rk}_{\mathrm{HE}}$ to obtain a standard two-component ciphertext denoted by ${\bm{c}}_{\mathrm{prod}}$. This relinearized ciphertext ${\bm{c}}_{\mathrm{prod}}$ is the multiplication output used in the main paper and is expressed as ${\bm{c}}_{\mathrm{prod}}=\bm{c}_1\otimes \bm{c}_2$. Thus, the computation path exercised in BioZKFHE consists of homomorphic addition $\oplus$, homomorphic multiplication $\otimes$, and relinearization, whereas modulus switching and bootstrapping are not invoked in the current similarity-evaluation path.

\section{Generic Background on Designated-Verifier Proofs}
\label{supp:supp:dvzk}

This section collects the generic designated-verifier proof background used by BioZKFHE beyond the BGV material of Section~\ref{supp:app:bgv}. We first recall the standard algebraic path from Rank-1 Constraint Systems to Quadratic Arithmetic Programs, and then sketch a generic lattice-based designated-verifier proof paradigm. Threshold decryption is summarized separately in Section~\ref{supp:app:td}, while the concrete SCMV packing comparison and the concrete PVSC formulation are deferred to later sections.

\textbf{Mathematical Notations and Structures}: The algebraic background in this section is presented over a finite field $\mathbb{F}$. An R1CS instance with $N_g$ constraints and $N_w$ variables is specified by matrices $\mathbf{A},\mathbf{B},\mathbf{C}\in\mathbb{F}^{N_g\times(N_w+1)}$ and an assignment vector $\bm{z}=[1,\bm{x},\bm{w}]^{\top}$, where $\bm{x}$ is the public input and $\bm{w}$ is the private witness. In the QAP view, $\{\omega_i\}_{i=0}^{N_g-1}$ denotes the interpolation points, $Z(X)$ is the vanishing polynomial over these points, and $A(X),B(X),C(X),H(X)$ denote the aggregate and quotient polynomials induced by $\bm{z}$. For the designated-verifier lattice proving layer, $\mathcal{C}$ denotes the target algebraic relation, $\mathrm{query\_num}$ the number of sampled query points, $\mathrm{query\_size}:=4$ the number of polynomial evaluations per point corresponding to $A,B,C,H$, $\mathrm{pt\_dim}:=\mathrm{query\_num}\cdot \mathrm{query\_size}$ the plaintext dimension of one proof vector, and $\tau$ an auxiliary dimension used for integrity checking of the prover's homomorphic linear aggregation. In the generic proof interface, the public proving material is denoted by $\mathsf{pp}_{\mathrm{ZK}}$, the opening secret by $\mathsf{sk}_{\mathrm{open}}$, the verifier-side public material by $\mathsf{vk}_{\mathrm{ZK}}$, and the encrypted proving queries by $\{\bm{C}_{Q}^{i}\}_{i=1}^{n_{\mathrm{row}}}$.

\subsection{Rank-1 Constraint System}

A \textit{Rank-1 Constraint System (R1CS)}~\cite{supp:liang2025sok} is a standard algebraic framework for encoding computation. Formally, an R1CS instance is defined by three matrices $\mathbf{A},\mathbf{B},\mathbf{C}\in\mathbb{F}^{N_g\times(N_w+1)}$, where $\mathbb{F}$ is a finite field, $N_g$ is the number of constraints, and $N_w$ is the number of witness and public-input variables excluding the constant $1$. Let $\bm{z}=[1,\bm{x},\bm{w}]^{\top}\in\mathbb{F}^{N_w+1}$, where $\bm{x}$ is the public input and $\bm{w}$ is the private witness. The R1CS relation is
\begin{equation}
\label{supp:eq:supp-r1cs}
(\mathbf{A}\bm{z})\odot(\mathbf{B}\bm{z})=\mathbf{C}\bm{z},
\end{equation}
where $\odot$ denotes the Hadamard product. Intuitively, each row of $(\mathbf{A},\mathbf{B},\mathbf{C})$ encodes one multiplicative constraint, and the assignment $(\bm{x},\bm{w})$ is valid if and only if all such constraints are satisfied simultaneously.

\subsection{Quadratic Arithmetic Program and R1CS-to-QAP Conversion}
\label{supp:supp:R1CStoQAP}

A \textit{Quadratic Arithmetic Program (QAP)}~\cite{supp:nainwal2025comparative} converts the constraint-satisfaction problem of R1CS into a polynomial identity problem. This conversion is widely used in succinct proof systems because it compresses many local constraints into one global algebraic relation.

\textbf{QAP Overview}: Let $\{\omega_i\}_{i=0}^{N_g-1}$ be interpolation points associated with the $N_g$ constraints. From the R1CS matrices, one constructs three families of polynomials $\{A_i(X)\}_{i=0}^{N_w}$, $\{B_i(X)\}_{i=0}^{N_w}$, and $\{C_i(X)\}_{i=0}^{N_w}$ by interpolation over these points, and defines the vanishing polynomial
\begin{equation}
\label{supp:eq:supp-zx}
Z(X)=\prod_{i=0}^{N_g-1}(X-\omega_i).
\end{equation}
Given an assignment $\bm{z}$, the aggregate polynomials are
\begin{equation}
\begin{aligned}
\label{supp:eq:supp-qap-abc}
A(X)&=\sum_{i=0}^{N_w} z_i A_i(X),\\
B(X)&=\sum_{i=0}^{N_w} z_i B_i(X),\\
C(X)&=\sum_{i=0}^{N_w} z_i C_i(X).
\end{aligned}
\end{equation}
Then the assignment $\bm{z}$ satisfies the original R1CS if and only if there exists a quotient polynomial $H(X)$ such that
\begin{equation}
\label{supp:eq:supp-qap}
A(X)\cdot B(X)=H(X)\cdot Z(X)+C(X).
\end{equation}
Thus, the entire constraint system is reduced to a polynomial divisibility condition.

\textbf{Purpose of the R1CS-to-QAP Conversion}: The main purpose of this conversion is to support succinct verification. A verifier could check all $N_g$ R1CS constraints directly, but this would require work linear in $N_g$. In contrast, the QAP formulation reduces verification to checking whether Eq.~\eqref{supp:eq:supp-qap} holds at random points. By the Schwartz--Zippel lemma, a false identity is unlikely to hold at a randomly chosen challenge point, so the verifier can test correctness with much less work than re-evaluating all constraints. This algebraic compression is the core reason why SNARK-style proof systems can achieve succinct verification.

\subsection{Lattice-Based Designated-Verifier ZKSNARK Construction}

We now sketch a generic designated-verifier lattice proving paradigm. At a high level, the setup algorithm derives public proving material and a secret opening key. The prover then combines encrypted query material with witness-dependent values to produce an encrypted proof object. Finally, the designated verifier opens this object and checks the corresponding polynomial identities.

\subsubsection{Setup phase}

The core function of the setup algorithm takes the target algebraic relation and system parameters as input and generates the designated-verifier proving material:
\begin{equation}
\label{supp:eq:supp-zk-setup}
\bigl(\mathsf{pp}_{\mathrm{ZK}},\mathsf{sk}_{\mathrm{open}}\bigr)
\leftarrow
\mathsf{ZK.Setup}(\mathcal{C},\mathrm{params}),
\end{equation}
where $\mathsf{pp}_{\mathrm{ZK}}$ contains the relation family and the encrypted proving queries, while $\mathsf{sk}_{\mathrm{open}}$ is the secret opening key used to open encrypted proof objects. The setup also produces a secret verification state $\mathsf{vk}_{\mathrm{ZK}}$, which includes the sampled query points below and is kept by the designated verifier.

The setup procedure may be viewed as having three steps. First, it generates the internal public/secret parameters for the underlying lattice-based linear encryption layer. Concretely, one may view the secret opening key as containing a secret matrix pair $(\mathbf{S},\mathbf{T}_{\mathrm{mat}})$, where $\mathbf{S}$ controls decryption of encrypted proof components and $\mathbf{T}_{\mathrm{mat}}$ supports an integrity check on the opened vector. The corresponding public parameters include a uniformly random matrix $\mathbf{A}$ and a matrix $\mathbf{D}$ derived from $\mathbf{S}$ by an RLWE-style relation over the ring $\mathcal{R}_q$. This is the designated-verifier analogue of key generation for the proof-opening layer.

Second, the setup compiles the target constraint system into its QAP form and samples random query points $\{t_1,\dots,t_{\mathrm{query\_num}}\}$. At these points, it evaluates the relevant QAP polynomials, thereby obtaining the query values associated with the polynomials $A,B,C$, and $H$.

Third, the setup organizes these evaluations into a query matrix $\mathbf{Q}$ whose columns correspond to the evaluation positions of the proof vector. Each row of $\mathbf{Q}$ is then encrypted under the designated-verifier opening key, yielding the encrypted proving queries that appear in the public proving material. Abstractly, these objects are collected in
\begin{equation}
\label{supp:eq:supp-pp-pvsc}
\mathsf{pp}_{\mathrm{ZK}}
=
\bigl(\mathbf{R}, d_{\mathrm{ZK}}, \{\bm{C}_{Q}^{i}\}_{i=1}^{n_{\mathrm{row}}}\bigr).
\end{equation}

\subsubsection{Proof Generation}

The prover executes the proof-generation phase using the public proving material, the public input, and the private witness. At the generic interface level, this is written as
\begin{equation}
\label{supp:eq:supp-zk-prove}
\Pi_{\mathrm{enc}}
\leftarrow
\mathsf{ZK.Prove}\bigl(\bm{x},\bm{w},\mathsf{pp}_{\mathrm{ZK}}\bigr),
\end{equation}
where $\Pi_{\mathrm{enc}}=\{(\delta^{k},\pi_{\mathrm{enc}}^{k})\}_{k=1}^{K}$ denotes the encrypted proof batch. Pedagogically, the proof-generation procedure has two conceptual steps.

\begin{enumerate}
    \item \textit{Witness Preparation}: The prover first computes the relevant witness-dependent algebraic values. In a generic R1CS/QAP picture, this means forming the witness vector $\bm{z}$, deriving the quotient polynomial $H(X)$ from Eq.~\eqref{supp:eq:supp-qap}, and then collecting the required evaluations of $A$, $B$, $C$, and $H$ at the sampled challenge points.

    \item \textit{Homomorphic Proof Aggregation}: The prover next combines the witness-dependent values with the encrypted query material. Abstractly, if the witness-derived proof vector is denoted by $\bm{\pi}$, then the encrypted proof object is obtained as an encrypted linear aggregation of the query ciphertext rows against the entries of $\bm{\pi}$. This is the key designated-verifier feature: the prover never reveals the raw evaluation vector directly, but only its encrypted image under the query system derived in setup.
\end{enumerate}

Because the aggregation is linear and the underlying encryption scheme is additively homomorphic, the resulting proof object is valid with respect to the witness-dependent evaluations while still hiding them from anyone who does not hold the opening key. This is the point at which succinctness and zero knowledge meet: the prover has compressed the witness into a short encrypted algebraic object that can later be opened and checked.

\subsubsection{Proof Verification}

In a designated-verifier proof system, verification naturally has two layers: opening and checking.

At the generic interface level, the designated verifier first opens the encrypted proof batch:
\begin{equation}
\label{supp:eq:supp-zk-open}
\Pi
\leftarrow
\mathsf{ZK.Open}\bigl(\Pi_{\mathrm{enc}},\mathsf{sk}_{\mathrm{open}}\bigr),
\end{equation}
where $\Pi=\{(\delta^{k},\pi^{k})\}_{k=1}^{K}$ is the corresponding plaintext proof batch. It then checks the opened batch by running
\begin{equation}
\label{supp:eq:supp-zk-verify}
b \leftarrow \mathsf{ZK.Verify}\bigl(\bm{x},\Pi,\mathsf{vk}_{\mathrm{ZK}}\bigr),
\end{equation}
where $b\in\{0,1\}$ is the verifier's accept/reject decision. Pedagogically, this verification process consists of four conceptual steps.

\begin{enumerate}
    \item \textit{Proof Opening}: The encrypted proof object is opened using the designated-verifier secret opening material. The result is a plaintext vector containing the algebraic proof values needed for checking.

    \item \textit{Integrity Check}: The verifier checks that the opened vector is structurally consistent with the auxiliary integrity relation induced by the matrix $\mathbf{T}_{\mathrm{mat}}$. This prevents a malicious prover from exploiting the homomorphic aggregation layer to produce malformed openings that do not correspond to a valid encrypted linear combination.

    \item \textit{Evaluation Reconstruction}: Once integrity has been validated, the verifier reconstructs the full polynomial evaluations used in the proof. In a QAP view, this means recovering the evaluations of $A$, $B$, $C$, and $H$ at the sampled challenge points after accounting for the public-input contribution.

    \item \textit{Algebraic Consistency Check}: Finally, the verifier checks the required polynomial identity. In the pedagogical QAP picture, this amounts to testing
\begin{equation}
\label{supp:eq:supp-zk-qap-check}
A(t_i)\cdot B(t_i)\stackrel{?}{=}H(t_i)\cdot Z(t_i)+C(t_i)
\end{equation}
for all sampled query points $t_i$. If these checks pass, the verifier accepts; otherwise, the proof is rejected.
\end{enumerate}

\textit{Secrecy of the verification state.} The check of Eq.~\eqref{supp:eq:supp-zk-qap-check} uses the secret query points $t_i$, and the opened values are evaluations at those points of polynomials that the prover knows. A prover that learned either the points or the opened values could therefore solve for the points and forge later proofs. BioZKFHE accordingly keeps both inside the committee's threshold boundary (main paper, Section~\ref{sec:PVSC_ver}).

This algebraic viewpoint explains the logic behind designated-verifier lattice proofs. Section~\ref{supp:supp:pvsc_formulation} later instantiates this generic setup/prove/open/verify pattern for the blockwise similarity trace used in BioZKFHE.

\section{Lattice-Based Threshold Decryption}
\label{supp:app:td}

Threshold Decryption (TD) is a fundamental primitive in distributed cryptographic systems~\cite{supp:bendlin2010threshold}. In BioZKFHE, it enables the Decryption Committee of $n_{\mathrm{com}}$ members to jointly recover the plaintext of a ciphertext such that no fewer than a threshold $t$ of committee members can perform the recovery, while no single member ever holds the complete secret key. This primitive is useful whenever decryption authority must be distributed across multiple entities rather than concentrated in one trusted server. Since the underlying FHE layer is BGV, the threshold decryption interface should be consistent with the ring structure and decryption form introduced in Section~\ref{supp:app:bgv}. This subsection therefore summarizes a lattice-based threshold decryption view over the same BGV framework, focusing on key generation and secret sharing, distributed partial decryption, and share aggregation for plaintext recovery.

\textbf{Mathematical Notations and Structures}: We use the same ring notation as in Section~\ref{supp:app:bgv}. In particular, $\mathcal{R}:=\mathbb{Z}[x]/(x^d+1)$, the plaintext ring is $\mathcal{R}_p:=\mathcal{R}/p\mathcal{R}$, and the ciphertext ring is $\mathcal{R}_q:=\mathcal{R}/q\mathcal{R}$, where $p \ll q$. The discrete Gaussian distribution over $\mathbb{Z}$ with standard deviation $\sigma$ is denoted by $D_{\mathbb{Z},\sigma}$. A ciphertext is written as $\bm{c}=(c^{(0)},c^{(1)})\in\mathcal{R}_q^2$, the global public key is written as $\mathsf{pk}_{\mathrm{HE}}$, and the Decryption Committee holds threshold shares $\{\mathsf{sk}_{\mathrm{HE}}^{i}\}_{i=1}^{n_{\mathrm{com}}}$ of the FHE secret key. For a subset $S\subseteq\{1,\dots,n_{\mathrm{com}}\}$ with $|S|\ge t$, the corresponding Lagrange coefficients are written as $\{\lambda_i(S)\}_{i\in S}$.

\subsection{Key Generation and Secret Sharing}

This part defines the threshold key-generation procedure, which generates the global public key and distributes secret-key shares among the $n_{\mathrm{com}}$ committee members.

\textbf{Key Generation}: The threshold key-generation interface is written as
\begin{equation}
\label{supp:eq:supp-td-keygen}
\bigl(\mathsf{pk}_{\mathrm{HE}}, \{\mathsf{sk}_{\mathrm{HE}}^{i}\}_{i=1}^{n_{\mathrm{com}}}\bigr)
\leftarrow
\mathsf{TD.KeyGen}(1^\lambda,n_{\mathrm{com}},t,d,p,q).
\end{equation}
The global secret key is first sampled in the same way as in the BGV scheme. Concretely, a small polynomial $s(x)\in\mathcal{R}_q$ is sampled coefficientwise from a small set such as $\{-1,0,1\}$, and a public key $\mathsf{pk}_{\mathrm{HE}}=(a(x),b(x))$ is generated with $a(x)\leftarrow\mathcal{R}_q$ uniform and
\begin{equation}
\label{supp:eq:supp-td-pk}
b(x)=-a(x)\cdot s(x)+p\,e(x)\pmod q,
\end{equation}
where $e(x)\in\mathcal{R}_q$ is a small error polynomial.

The key difference from standard BGV is that the secret key $s(x)$ is not kept by any single entity. Instead, it is split into shares under a $(t,n_{\mathrm{com}})$ threshold mechanism. For exposition, one may describe this through a Shamir-style sharing polynomial
\begin{equation}
\label{supp:eq:supp-td-sharepoly}
F(z)=s(x)+\sum_{j=1}^{t-1}\alpha_j(x)z^j,
\end{equation}
where $\alpha_j(x)\in\mathcal{R}_q$ are sampled uniformly or from an appropriate masking distribution. The secret-key share of committee member $i$ is then
\begin{equation}
\label{supp:eq:supp-td-share}
\mathsf{sk}_{\mathrm{HE}}^{i}=s_i(x):=F(i)\in\mathcal{R}_q,
\qquad i\in\{1,\dots,n_{\mathrm{com}}\}.
\end{equation}
Thus, no single committee member learns the full secret key, while any set of at least $t$ shares can support decryption through interpolation in the sharing space.

\subsection{Distributed Partial Decryption}

This section describes how an individual committee member uses its own secret share to produce a partial decryption contribution.

\textbf{Share Generation}: Given a ciphertext $\bm{c}=(c^{(0)}(x),c^{(1)}(x))\in\mathcal{R}_q^2$ and a secret share $\mathsf{sk}_{\mathrm{HE}}^{i}=s_i(x)$, committee member $i$ computes a partial decryption share
\begin{equation}
\label{supp:eq:supp-td-partdec-intf}
\mu_i(x)\leftarrow \mathsf{TD.PartDec}(\bm{c},\mathsf{sk}_{\mathrm{HE}}^{i}).
\end{equation}
At a high level, this partial share should contribute to the global decryption term $c^{(1)}(x)\cdot s(x)$, but without revealing the full secret key. A natural form is
\begin{equation}
\label{supp:eq:supp-td-partdec}
\mu_i(x)=c^{(1)}(x)\cdot s_i(x)+p\,e_i^{\mathrm{sm}}(x)\pmod q,
\end{equation}
where $e_i^{\mathrm{sm}}(x)\in\mathcal{R}_q$ is an additional masking or smudging polynomial. The factor $p$ is important: it ensures that the extra masking term vanishes after the final reduction modulo $p$, and therefore does not alter the recovered plaintext. In this way, each party releases only a masked partial contribution rather than an exact decryption term.

Depending on the concrete threshold implementation, a party may also attach an auxiliary proof that its partial share was computed consistently from a valid secret-key share. The main paper does not rely on any specific share-verification mechanism, and therefore abstracts this step away.

\subsection{Share Aggregation and Reconstruction}

The final plaintext is recovered by aggregating at least $t$ valid partial decryption shares.

\textbf{Reconstruction}: Let $S\subseteq\{1,\dots,n_{\mathrm{com}}\}$ be any subset with $|S|\ge t$, and let $\{\mu_i(x)\}_{i\in S}$ be the corresponding valid partial shares. The aggregation interface is
\begin{equation}
\label{supp:eq:supp-td-combine-intf}
v(x)\leftarrow \mathsf{TD.Combine}\bigl(\bm{c},\{\mu_i\}_{i\in S}\bigr).
\end{equation}
The combiner first computes the Lagrange coefficients
\begin{equation}
\label{supp:eq:supp-td-lagrange}
\lambda_i(S)=\prod_{\substack{j\in S\\ j\neq i}}\frac{-j}{i-j},
\qquad i\in S,
\end{equation}
over the sharing domain. It then forms the aggregated partial decryption term
\begin{equation}
\label{supp:eq:supp-td-agg}
\mu_{\mathrm{agg}}(x)=\sum_{i\in S}\lambda_i(S)\cdot \mu_i(x)\pmod q.
\end{equation}
By the linearity of the threshold sharing scheme, the secret shares interpolate back to the original secret key, so
\begin{equation}
\label{supp:eq:supp-td-secret-recover}
\sum_{i\in S}\lambda_i(S)\cdot s_i(x)=s(x).
\end{equation}
Substituting Eq.~\eqref{supp:eq:supp-td-partdec} into Eq.~\eqref{supp:eq:supp-td-agg} yields
\begin{equation}
\label{supp:eq:supp-td-agg-expand}
\mu_{\mathrm{agg}}(x)=c^{(1)}(x)\cdot s(x)+p\,E_{\mathrm{sm}}(x)\pmod q,
\end{equation}
where $E_{\mathrm{sm}}(x)$ denotes the aggregated masking polynomial.

Finally, the plaintext polynomial is recovered by combining the first ciphertext component with the aggregated decryption term and reducing modulo $p$:
\begin{equation}
\label{supp:eq:supp-td-final}
v(x)=\bigl(c^{(0)}(x)+\mu_{\mathrm{agg}}(x)\bigr)\bmod p.
\end{equation}
Using the decryption form from Section~\ref{supp:app:bgv}, one sees that
\begin{equation}
\label{supp:eq:supp-td-final-expand}
c^{(0)}(x)+\mu_{\mathrm{agg}}(x)=v(x)+p\,E(x)\pmod q,
\end{equation}
for an effective error polynomial $E(x)$ absorbing both the original encryption noise and the aggregated masking noise. Therefore, reduction modulo $p$ recovers the plaintext polynomial exactly, provided that the effective error remains within the correctness regime of the BGV scheme.

This threshold view matches the system abstraction used in the main paper. Each committee member contributes a local decryption share, and an untrusted combiner may aggregate at least $t$ such shares. The resulting plaintext is accepted only through the later committee-side and on-chain consistency checks of BioZKFHE. Hence, the combiner may affect liveness by delaying or misaggregating shares, but cannot by itself cause an incorrect plaintext result to be finalized.

\section{Comparison with Conventional Packing Layouts}
\label{supp:sec:SCMV_comparison}

To make the role of SCMV more explicit, we compare it with two conventional encrypted biometric layouts, namely horizontal packing and vertical packing. The goal of this comparison is not to claim a universal dominance result for every implementation detail, but to isolate the structural difference that matters most in BioZKFHE: how many templates are covered by one processed database unit, and whether query-time similarity evaluation requires slot aggregation or ciphertext rotation.

For clarity, we use the term \emph{processed database unit} to denote the minimum encrypted unit consumed by one invocation of the matching kernel. Under vertical packing and SCMV, this unit is naturally one block of $L$ row ciphertexts. Under horizontal packing, the analogous unit is one encrypted template, since one template is typically packed across slots and processed separately.

The comparison is also meant to clarify the novelty boundary. Base-$T$ digit shifting is a known packing primitive: it may be applied in the coefficient domain by storing digit bundles as polynomial coefficients, or combined with batching by storing one logical value per CRT slot and using SIMD parallelism. SCMV uses the second representation layer differently: each CRT slot stores a base-$T$ bundle of $n$ same-coordinate biometric values, and the slots are arranged row-wise by feature coordinate. This row-wise slot-level binding is what lets one processed block cover $nd$ templates and keeps the online trace aligned with the blockwise PVSC proof batch.

{
The distinction can be written explicitly as
\begin{equation}
\label{supp:eq:supp_packing_contrast}
\begin{aligned}
P_{\mathrm{coef}}(X)&=\sum_{j=1}^{d} c_jX^{j-1},\\
c_j&=\sum_{t=1}^{n}a_{j,t}T^{t-1},\\
P_{\mathrm{slot}}(X)&=\mathsf{Encode}(a_1,\ldots,a_d),\\
P_{\mathrm{SCMV},i}^{m}(X)&=\mathsf{Encode}\big((b_{i,j}^{m})_{j=1}^{d}\big),\\
b_{i,j}^{m}&=\sum_{t=1}^{n}v_{i,r(m,t,j)}T^{t-1},\\
r(m,t,j)&=(m-1)nd+(t-1)d+j.
\end{aligned}
\end{equation}
Here $P_{\mathrm{coef}}$ represents coefficient-domain digit packing, $P_{\mathrm{slot}}$ represents ordinary one-value-per-slot batching, and $P_{\mathrm{SCMV},i}^{m}$ is the row-wise slot-domain binding used by SCMV.
}

\begin{table*}[t]
\centering
\caption{Structural comparison of encrypted biometric packing layouts. Here $N$ is the database size, $d$ is the batching degree, $L$ is the embedding dimension, and $n$ is the SCMV binding factor.}
\label{supp:tab:packing_comparison}
\scriptsize
\setlength{\tabcolsep}{4pt}
\begin{tabular}{p{2.5cm}p{3.5cm}p{2.0cm}p{2.5cm}p{2cm}p{3cm}}
\toprule
\textbf{Layout} & \textbf{Packing principle} & \textbf{Templates covered by one processed unit} & \textbf{Query-time slot aggregation / rotation} & \textbf{Number of processed units for $N$ templates} & \textbf{Incremental update characteristic} \\
\midrule
Horizontal packing &
One template is packed contiguously across slots and compared separately &
About one template per processed unit &
Typically required to reduce slot-wise products to one similarity value &
$\Theta(N)$ &
Simple append, but poor repository-level amortization \\

Row-wise vertical baseline &
The same feature coordinate of many templates is packed into one row-wise SIMD structure &
$d$ templates per block of $L$ row ciphertexts &
Typically required to aggregate slot-level results or align data across the matching trace &
$\Theta(\lceil N/d\rceil)$ &
Amortizes storage better than horizontal packing, but the similarity path remains rotation-sensitive \\

SCMV &
$n$ quantized values are bound into each plaintext entry through base-$T$ expansion and then arranged in the same row-wise block structure &
$nd$ templates per block of $L$ row ciphertexts &
No query-time rotation or slot aggregation in the similarity path &
$\Theta(\lceil N/(nd)\rceil)$ &
Supports incremental enrollment by homomorphic addition into the last partially filled block \\
\bottomrule
\end{tabular}
\end{table*}

Table~\ref{supp:tab:packing_comparison} highlights the main design trade-off. Horizontal packing is conceptually simple, but it provides little amortization across the repository and usually requires an internal reduction step to turn slot-wise products into one similarity score. Row-wise vertical baselines improve ciphertext utilization by letting one processed unit cover $d$ templates, but the matching path is still commonly dominated by slot aggregation and rotation-related overhead. SCMV changes the granularity of packing itself: instead of storing only one quantized value in each plaintext entry, it binds $n$ values into that entry through base-$T$ expansion. As a result, one SCMV block covers $nd$ templates while preserving the same row-wise block structure.

This structural difference directly affects both storage and online matching cost. Since $\mathsf{SCMV.Identify}$ evaluates one block at a time and the number of blocks scales as $\lceil N/(nd)\rceil$, increasing the binding factor $n$ reduces the number of encrypted blocks that must be stored, scanned, and proved. At the same time, SCMV removes query-time rotation from the similarity path, so the online cost is driven mainly by the block count and the per-block multiply-and-accumulate trace rather than by slot-shuffling operations.

Table~\ref{supp:tab:supp_scmv_measured_baselines} adds a matched-codebase baseline comparison at $N=100$k. Horizontal packing is included to show the cost of template-major slot aggregation, vertical packing corresponds to the row-wise $n=1$ baseline, and SCMV increases the binding factor while preserving the same rotation-free online trace. These rows are intended to complement the main paper's Table~I: base-$T$ digit packing is a known representation technique, while BioZKFHE uses it in a row-wise biometric SIMD layout that reduces the block count and keeps each block trace aligned with PVSC verification.

\begin{table*}[!t]
\centering
{
\caption{Measured SCMV baseline comparison at $N=100$k. Latency is the encrypted matching latency in milliseconds.}
\label{supp:tab:supp_scmv_measured_baselines}
\scriptsize
\setlength{\tabcolsep}{4.8pt}
\renewcommand{\arraystretch}{1.05}
\begin{tabular}{llrrrr}
\toprule
\textbf{Model} & \textbf{Method} & \textbf{$n$} & \textbf{Blocks} & \textbf{Online rotations} & \textbf{Latency / storage} \\
\midrule
FaceNet & Horizontal packing & -- & 1{,}563 & 10{,}941 & 75{,}183.3~ms / 306.06~MiB \\
FaceNet & Row-wise vertical baseline & 1 & 13 & 0 & 1{,}427.0~ms / 326.13~MiB \\
FaceNet & SCMV & 2 & 7 & 0 & 720.9~ms / 221.52~MiB \\
FaceNet & SCMV & 3 & 5 & 0 & 543.4~ms / 188.51~MiB \\
MobileFaceNet & Horizontal packing & -- & 6{,}250 & 56{,}250 & 381{,}467.2~ms / 1.27~GiB \\
MobileFaceNet & Row-wise vertical baseline & 1 & 13 & 0 & 6{,}311.7~ms / 1.36~GiB \\
MobileFaceNet & SCMV & 2 & 7 & 0 & 3{,}412.9~ms / 984.30~MiB \\
\bottomrule
\end{tabular}
}
\end{table*}

Table~\ref{supp:tab:supp_e3_security_tradeoff} reports the per-$n$ security/performance trade-off induced by increasing the SCMV binding factor. The security estimates were obtained with SageMath 10.9.rc0 and \texttt{lattice-estimator}, using rough Core-SVP/GSA estimates under an RLWE-to-LWE proxy with secret distribution $\mathrm{CenteredBinomial}(2)$ and error standard deviation $\sigma=3.2$. Core-SVP is a lower-bound cost model and the flat RLWE-to-LWE proxy ignores ring structure, so it is conservative in the sense of under-estimating attack cost; the two-decimal values are listed only for reproducibility and should be read as approximate security margins rather than exact bit-security levels. The estimates characterize the BGV/RLWE confidentiality layer; the module-lattice parameters of the designated-verifier proof backend that carry the local soundness $\epsilon_{\mathrm{loc}}$ are governed by~\cite{supp:ishai2021dvzksnark} and are not re-estimated here.

\begin{table*}[!t]
\centering
{
\caption{Per-$n$ parameter security and matching trade-off at $N=10^6$. The latency column reports encrypted SCMV matching latency.}
\label{supp:tab:supp_e3_security_tradeoff}
\scriptsize
\setlength{\tabcolsep}{4.2pt}
\renewcommand{\arraystretch}{1.05}
\begin{tabular}{lrrrrlrrr}
\toprule
\textbf{Model} & \textbf{$n$} & \textbf{$p$ bits} & \textbf{$q$ bits} & \textbf{Security bits} & \textbf{RNS limbs} & \textbf{$M$} & \textbf{Storage} & \textbf{Latency} \\
\midrule
FaceNet & 1 & 20 & 110 & 243.66 & 30+25+25+30 & 123 & 5.77~GiB & 6275.585~ms \\
FaceNet & 2 & 31 & 150 & 161.18 & 45+30+30+45 & 62 & 2.91~GiB & 3149.330~ms \\
FaceNet & 3 & 46 & 170 & 135.49 & 45+40+40+45 & 41 & 1.92~GiB & 2150.655~ms \\
MobileFaceNet & 1 & 21 & 120 & 217.54 & 30+30+30+30 & 123 & 23.07~GiB & 25701.169~ms \\
MobileFaceNet & 2 & 40 & 160 & 147.46 & 45+35+35+45 & 62 & 11.63~GiB & 12804.881~ms \\
\bottomrule
\end{tabular}
}
\end{table*}

The table makes the cost of increasing $n$ explicit rather than treating SCMV packing as free. For FaceNet, moving from $n=1$ to $n=3$ increases the ciphertext modulus from 110 to 170 bits and reduces the estimated security margin from 243.66 to 135.49 bits, but the number of processed blocks drops from 123 to 41 and the measured/extrapolated matching latency drops by 65.7\%. For MobileFaceNet, moving from $n=1$ to $n=2$ keeps a 147.46-bit estimate while halving the block count from 123 to 62 and reducing the matching latency by 50.2\%.

\section{Concrete PVSC Formulation for BioZKFHE}
\label{supp:supp:pvsc_formulation}

We now make explicit the descriptor-level relation compiled by PVSC for one SCMV database block. For each block $\bm{C}_m\in\mathcal{D}_{\mathrm{SCMV}}$, PVSC certifies a deterministic BGV trace for the fixed input/output tuple $\bigl(\bm{C}_{\mathrm{identify}},\bm{C}_m,\bm{c}^{\mathrm{sim}}_m\bigr)$. Letting $\bm{p}^{m}_{i}$ denote the relinearized product ciphertext obtained from the $i$-th query row and the $i$-th database row, and $\bm{a}^{m}_{i}$ the running accumulator, the admitted trace is
\begin{equation}
\label{supp:eq:supp-pvsc-trace}
\bm{p}^{m}_{i}:=\mathsf{Relin}\!\left(\bm{c}_{i}^{\mathrm{id}}\otimes \bm{c}_{i}^{m}\right),\quad
\bm{a}^{m}_{0}:=\bm{0},\quad
\bm{a}^{m}_{i}:=\bm{a}^{m}_{i-1}\oplus \bm{p}^{m}_{i},
\end{equation}
for $i\in\{1,\dots,L\}$, with terminal condition $\bm{a}^{m}_{L}=\bm{c}^{\mathrm{sim}}_{m}$. Thus PVSC certifies one fixed multiply--relinearize--accumulate schedule rather than an arbitrary claim of output consistency.

\paragraph{Public statement and witness.}
For block index $m$, the public statement contains the published output ciphertext $\bm{c}^{\mathrm{sim}}_m$, the block index $m$, the session context $\mathsf{ctx}$, the query-binding commitment $h_{\mathrm{id}}:=\mathcal{H}(\mathsf{ctx}\Vert \bm{C}_{\mathrm{identify}})$, the session-bound block commitment $h_m:=\mathcal{H}(\mathsf{ctx}\Vert m\Vert u_m)$ with $u_m:=\mathcal{H}(\bm{C}_m)$, a batch-root digest $\rho_m$ of the canonical proof-instance descriptor sequence, and the public system parameters $\mathsf{pp}_{\mathrm{sys}}$. We write
\begin{equation}
\label{supp:eq:supp-pvsc-public}
\bm{x}^{(m)} :=
\big(
\bm{c}^{\mathrm{sim}}_m,\,
m,\,
\mathsf{ctx},\,
h_{\mathrm{id}},\,
h_m,\,
\rho_m,\,
\mathsf{pp}_{\mathrm{sys}}
\big).
\end{equation}
The private witness $\bm{w}^{(m)}$ contains the ciphertext inputs $\bm{C}_{\mathrm{identify}}$ and $\bm{C}_m$, the product ciphertexts $\{\bm{p}^{m}_{i}\}_{i=1}^{L}$, the accumulator chain $\{\bm{a}^{m}_{i}\}_{i=0}^{L}$, the Double-CRT decomposition witnesses used by the local proof templates, and the preimages of the descriptor digests needed to show that all local instances refer to one consistent hidden trace.

\paragraph{Relation decomposition.}
Let $\mathcal{J}:=\{1,\dots,n_q\}$ and $\mathcal{X}:=\{1,\dots,n_{\mathrm{chunk}}\}$. For each $(j,\chi)\in\mathcal{J}\times\mathcal{X}$, let $d_{\chi}\le d_{\mathrm{ZK}}$ be the size of the corresponding NTT chunk, and let $\mathbb{F}_{\nu_j}$ be the proof field used to compile modulus-$q^{(j)}$ arithmetic into R1CS/QAP form, with $\nu_j$ chosen large enough to represent the local modular equations without ambiguity. PVSC proves the NP relation
\begin{equation}
\label{supp:eq:supp-pvsc-np}
\exists\,\bm{w}^{(m)} \ \text{s.t.}\ \mathsf{Rel}_{\mathrm{PVSC}}\big(\bm{x}^{(m)},\bm{w}^{(m)}\big)=1,
\end{equation}
where the defining conjunction is
\begin{equation}
	\label{supp:eq:supp-pvsc-rel}
	\begin{aligned}
		&\bigwedge_{j\in\mathcal{J}}\bigwedge_{\chi\in\mathcal{X}}\bigwedge_{i=1}^{L} {} 
		\qquad \mathcal{R}^{\mathrm{mul}}_{j}\big(
		\delta_{\mathrm{mul}}^{(m,j,\chi,i)},
		\omega_{\mathrm{mul}}^{(m,j,\chi,i)}
		\big)=1 \\
		&\land \bigwedge_{j\in\mathcal{J}}\bigwedge_{\chi\in\mathcal{X}}\bigwedge_{i=2}^{L-1} {}
		\qquad \mathcal{R}^{\mathrm{acc}}_{j}\big(
		\delta_{\mathrm{acc}}^{(m,j,\chi,i)},
		\omega_{\mathrm{acc}}^{(m,j,\chi,i)}
		\big)=1 \\
		&\land \bigwedge_{j\in\mathcal{J}}\bigwedge_{\chi\in\mathcal{X}} {} 
		\qquad \mathcal{R}^{\mathrm{out}}_{j}\big(
		\delta_{\mathrm{out}}^{(m,j,\chi)},
		\omega_{\mathrm{out}}^{(m,j,\chi)}
		\big)=1 \\
		&\land \mathcal{R}^{\mathrm{bind}}\big(
		\delta_{\mathrm{bind}}^{(m)},
		\omega_{\mathrm{bind}}^{(m)}
		\big)=1.
	\end{aligned}
\end{equation}
The conjunction in Eq.~\eqref{supp:eq:supp-pvsc-rel} enforces four requirements simultaneously: local multiplication correctness, local accumulation correctness, the terminal accumulation step together with output consistency with the public ciphertext $\bm{c}^{\mathrm{sim}}_m$, and one block-level binding check tying the entire batch to $(\mathsf{ctx}, h_{\mathrm{id}}, h_m, \rho_m)$.

\paragraph{Counted instances per block.} The per-block instance count $K_m=n_qn_{\mathrm{chunk}}(2L-1)$ follows directly from Eq.~\eqref{supp:eq:supp-pvsc-rel}. For each of the $n_qn_{\mathrm{chunk}}$ modulus/chunk pairs $(j,\chi)$, the multiplication relation $\mathcal{R}^{\mathrm{mul}}_j$ contributes $L$ instances (one relinearized product $\bm{p}_i$ per embedding dimension $i\in\{1,\dots,L\}$). The accumulation relation $\mathcal{R}^{\mathrm{acc}}_j$ contributes $L-2$ instances: the chain $\bm{a}_0=\bm{0},\,\bm{a}_i=\bm{a}_{i-1}\oplus\bm{p}_i$ has its first step $\bm{a}_1=\bm{p}_1$ fused into the multiplication output (so it is not separately proved) and its terminal step $\bm{a}_L=\bm{c}^{\mathrm{sim}}_m$ certified by the output relation, leaving the intermediate transitions $i=2,\dots,L-1$. The output relation $\mathcal{R}^{\mathrm{out}}_j$ contributes one instance per pair, certifying the terminal accumulation step and the equality of $\bm{a}_L$ with the published ciphertext. Summing over the $n_qn_{\mathrm{chunk}}$ pairs gives $n_qn_{\mathrm{chunk}}\bigl(L+(L-2)+1\bigr)=n_qn_{\mathrm{chunk}}(2L-1)$, where the fused initial step is realized by setting the first accumulator digest equal to the first product digest (formalized below). Thus $K_m$ counts the arithmetic and output instances; the single per-block binding relation $\mathcal{R}^{\mathrm{bind}}$ adds one further instance, so the full opened batch comprises $K_m+1$ instances, with the additive $+1$ an $O(1)$ term affecting neither the linear footprint scaling nor, by more than one bit, the logarithmic union-bound slack of the session soundness accounting below.

\paragraph{Descriptor families and canonical batch layout.}
PVSC realizes Eq.~\eqref{supp:eq:supp-pvsc-rel} through the fixed family
\begin{equation}
\label{supp:eq:supp-pvsc-family}
\begin{aligned}
\mathbf{R}_{\mathrm{PVSC}}
:=\;&
\{\mathcal{R}^{\mathrm{mul}}_{j},\mathcal{R}^{\mathrm{acc}}_{j},\mathcal{R}^{\mathrm{out}}_{j}\}_{j=1}^{n_q}\cup\{\mathcal{R}^{\mathrm{bind}}\}.
\end{aligned}
\end{equation}
Each instantiated proof instance carries a public descriptor recording its lexicographic index tuple and the digests of the boundary values it is meant to connect. Concretely,
\begin{itemize}
    \item $\delta_{\mathrm{mul}}^{(m,j,\chi,i)}=\bigl(m,j,\chi,i,\eta_{\mathrm{id}}^{(m,j,\chi,i)},\eta_{\mathrm{db}}^{(m,j,\chi,i)},\eta_{\mathrm{prod}}^{(m,j,\chi,i)}\bigr)$ binds the relevant input chunks to the claimed relinearized product chunk.
    \item $\delta_{\mathrm{acc}}^{(m,j,\chi,i)}=\bigl(m,j,\chi,i,\eta_{\mathrm{acc}}^{(m,j,\chi,i-1)},\eta_{\mathrm{prod}}^{(m,j,\chi,i)},\eta_{\mathrm{acc}}^{(m,j,\chi,i)}\bigr)$ links the previous accumulator chunk, the product chunk, and the next accumulator chunk.
    \item {The output descriptor
    \[
    \begin{aligned}
    \delta_{\mathrm{out}}^{(m,j,\chi)}
    =\bigl(&m,j,\chi,\eta_{\mathrm{acc}}^{(m,j,\chi,L-1)},
    \eta_{\mathrm{prod}}^{(m,j,\chi,L)},\\
    &\eta_{\mathrm{acc}}^{(m,j,\chi,L)},
    \eta_{\mathrm{pub}}^{(m,j,\chi)}\bigr)
    \end{aligned}
    \]
    certifies the terminal accumulation step $\bm{a}_L=\bm{a}_{L-1}\oplus\bm{p}_L$ by binding the previous accumulator chunk and the final product chunk to the final accumulator chunk, and links that final accumulator chunk to the corresponding public-output chunk extracted from $\bm{c}^{\mathrm{sim}}_m$.}
    \item $\delta_{\mathrm{bind}}^{(m)}=\bigl(m,\mathsf{ctx},h_{\mathrm{id}},h_m,\rho_m,\eta_{\mathrm{pub}}^{(m)}\bigr)$ binds the descriptor-sequence digest, session information, and public output ciphertext digest $\eta_{\mathrm{pub}}^{(m)}:=\mathcal{H}(\bm{c}^{\mathrm{sim}}_m)$.
\end{itemize}
To realize the fused initial step, the parser sets $\eta_{\mathrm{acc}}^{(m,j,\chi,1)}:=\eta_{\mathrm{prod}}^{(m,j,\chi,1)}$, identifying the first accumulator chunk $\bm{a}_1$ with the first product chunk $\bm{p}_1$; consequently $\delta_{\mathrm{acc}}$ is instantiated only for $i=2,\dots,L-1$, while the terminal step $i=L$ is folded into $\delta_{\mathrm{out}}$.

Accordingly, the opened proof batch for block $m$ is parsed as
\begin{equation}
\label{supp:eq:supp-pvsc-batch}
\Pi^{(m)}
=
\begin{aligned}[t]
&\bigl\langle\bigl(\delta_{\mathrm{mul}}^{(m,j,\chi,i)},\pi_{\mathrm{mul}}^{(m,j,\chi,i)}\bigr)\bigr\rangle_{j,\chi,i}^{\mathrm{lex}} \\
&\mathbin{\|}\;\bigl\langle\bigl(\delta_{\mathrm{acc}}^{(m,j,\chi,i)},\pi_{\mathrm{acc}}^{(m,j,\chi,i)}\bigr)\bigr\rangle_{j,\chi,i}^{\mathrm{lex}} \\
&\mathbin{\|}\;\bigl\langle\bigl(\delta_{\mathrm{out}}^{(m,j,\chi)},\pi_{\mathrm{out}}^{(m,j,\chi)}\bigr)\bigr\rangle_{j,\chi}^{\mathrm{lex}} \\
&\mathbin{\|}\;\bigl\langle\bigl(\delta_{\mathrm{bind}}^{(m)},\pi_{\mathrm{bind}}^{(m)}\bigr)\bigr\rangle,
\end{aligned}
\end{equation}
where $\mathbin{\|}$ denotes ordered concatenation, and each indexed angle-bracket term denotes the corresponding sequence listed in lexicographic order over its index tuple. Thus the full batch is ordered first by family type and then by the relevant indices. Each modulus/chunk pair contributes $L$ first-layer instances, $L-2$ second-layer instances, and one output-binding instance. Hence
\begin{equation}
\label{supp:eq:supp-pvsc-km}
K_m = n_q\,n_{\mathrm{chunk}}\,(2L-1).
\end{equation}
Here $K_m$ counts only the arithmetic and output instances, i.e., the $\mathcal{R}^{\mathrm{mul}}$, $\mathcal{R}^{\mathrm{acc}}$, and $\mathcal{R}^{\mathrm{out}}$ families. The full opened batch $\Pi^{(m)}$ of Eq.~\eqref{supp:eq:supp-pvsc-batch} additionally contains the single block-binding instance $(\delta_{\mathrm{bind}}^{(m)},\pi_{\mathrm{bind}}^{(m)})$, so it comprises $K_m+1$ proof instances in total. The additive $+1$ is an $O(1)$ term that affects neither the linear footprint scaling nor, by more than one bit, the logarithmic union-bound slack of the session soundness accounting.
The sequence digest is computed as
\begin{equation}
\label{supp:eq:supp-pvsc-rho}
\rho_m:=\mathcal{H}\bigl(\delta^{(m,1)}\Vert\cdots\Vert\delta^{(m,K_m)}\bigr),
\end{equation}
where $\delta^{(m,1)},\dots,\delta^{(m,K_m)}$ are the $K_m$ counted arithmetic and output descriptors (the $\mathcal{R}^{\mathrm{mul}}$, $\mathcal{R}^{\mathrm{acc}}$, and $\mathcal{R}^{\mathrm{out}}$ families) listed in the canonical order above, \emph{excluding} the binding descriptor $\delta_{\mathrm{bind}}^{(m)}$. The binding instance then commits to this digest, since $\rho_m$ is an input to $\delta_{\mathrm{bind}}^{(m)}$; hence the definition of $\rho_m$ is not circular.

\begin{table*}[!t]
	\centering
	{
		\caption{Leakage profile of threshold-opened PVSC proof batches.}
		\label{supp:tab:supp_opened_proof_leakage}
		\scriptsize
		\setlength{\tabcolsep}{4.2pt}
		\renewcommand{\arraystretch}{1.08}
		\begin{tabular}{p{3.0cm}p{5.6cm}p{7.0cm}}
			\toprule
			\textbf{Component} & \textbf{Public after opening} & \textbf{Hidden under witness zero knowledge} \\
			\midrule
			Session and statement binding &
			Block index $m$, session context $\mathsf{ctx}$, commitments $h_{\mathrm{id}}$, $h_m$, descriptor digest $\rho_m$, public output-ciphertext digest &
			Biometric plaintexts, decryption shares, and secret keys; the plaintext similarity vector is released only through the threshold-governed application flow \\
			Descriptor sequence &
			Family type, lexicographic order, local instance count $K_m$, RNS limb identifiers, NTT chunk identifiers, and expected boundary-digest wiring &
			Private trace values that realize those descriptors, including hidden ciphertext-state decompositions and witness assignments \\
			Double-CRT witness material &
			Parameter-derived shape information such as $n_q$, $n_{\mathrm{chunk}}$, and chunk indices &
			RNS/NTT coefficient values, product ciphertext components, accumulator states, and intermediate homomorphic states \\
			Opened local proofs &
			None: opened values stay secret-shared inside the committee's threshold verification; only accept bits and attestations are released &
			Witness-dependent values beyond what is implied by the public statement, descriptors, and the validity of the proved relation, and the verification state those values could reveal \\
			\bottomrule
		\end{tabular}
	}
\end{table*}

\paragraph{Proof generation and deterministic parsing.}
At proof-generation time, the prover first materializes the trace of Eq.~\eqref{supp:eq:supp-pvsc-trace}, then decomposes the relevant ciphertext states into Double-CRT witnesses, deterministically forms the descriptor sequence, computes $\rho_m$ by Eq.~\eqref{supp:eq:supp-pvsc-rho}, and finally produces the encrypted local proofs $\{\pi_{\mathrm{enc}}^{(m,k)}\}_{k=1}^{K_m+1}$ against the corresponding descriptor/witness pairs (the $K_m$ counted arithmetic and output proofs together with the single block-binding proof).

After threshold opening, the smart contract resolves $u_m$ from the session-bound gallery snapshot, recomputes the expected block commitment $h_m^{\star}:=\mathcal{H}(\mathsf{ctx}\Vert m\Vert u_m)$, and rejects repeated or mismatched indices. It then runs a deterministic parser over the submitted descriptor sequence of $\Pi^{(m)}$ and accepts the structure only if all of the following hold:
\begin{enumerate}
    \item the descriptor multiset has exactly the required cardinalities for $\mathcal{R}^{\mathrm{mul}}$, $\mathcal{R}^{\mathrm{acc}}$, $\mathcal{R}^{\mathrm{out}}$, and $\mathcal{R}^{\mathrm{bind}}$;
    \item the descriptors appear in the canonical lexicographic order over family type and index tuple;
    \item every accumulation descriptor consumes the output digest of the preceding accumulator state and the product digest of the matching multiplication instance;
    \item every output descriptor checks the terminal accumulation step for the same modulus/chunk pair, consuming the previous accumulator digest $\eta_{\mathrm{acc}}^{(\cdot,L-1)}$ and the final product digest $\eta_{\mathrm{prod}}^{(\cdot,L)}$ and binding them to the final accumulator digest $\eta_{\mathrm{acc}}^{(\cdot,L)}$ and the corresponding public-output chunk;
    \item the recomputed sequence digest equals the submitted $\rho_m$.
\end{enumerate}
The committee runs $\mathsf{PVSC.Verify}\bigl(\bm{x}^{(m)},\Pi^{(m)}\bigr)$ on all opened local proofs inside its threshold boundary, and the contract accepts the batch only if these structural checks pass and at least $t$ consistent attestations accept it. This parser is what prevents omission, duplication, reordering, and proof-splicing attacks at the opened-batch level.

\paragraph{Leakage profile after threshold opening.}
Threshold opening converts the encrypted designated-verifier proof object into a secret-shared opened batch that is consumed by the committee's threshold verification; neither the opened values nor the private PVSC witness become public. The zero-knowledge guarantee is interpreted relative to the public statement and the descriptor sequence that the verifier must see in order to parse the batch. Table~\ref{supp:tab:supp_opened_proof_leakage} summarizes this boundary. In particular, the Double-CRT decomposition affects the shape and count of local proof instances, but the RNS/NTT witness coefficients, product ciphertext components, accumulator states, and biometric plaintext values remain witness-hidden under the adopted proof backend.

\paragraph{Session-level soundness accounting.}
Let $\epsilon_{\mathrm{loc}}(\lambda_{\mathrm{ZK}})$ denote the maximum soundness error of one opened local proof instance in the designated-verifier lattice proof backend, instantiated from linear-only vector encryption over rank-2 module lattices under the Module-LWE/Module-SIS-type assumptions stated in~\cite{supp:ishai2021dvzksnark}, and let $\epsilon_{\mathcal{H}}(\lambda_{\mathcal{H}})$ denote the collision advantage against the hash function used for commitments and descriptor roots. For a query session with $M=\lceil N/(nd)\rceil$ gallery blocks and $K_m+1$ verified local instances per block (with $K_m=n_qn_{\mathrm{chunk}}(2L-1)$ counted arithmetic/output instances and one block-binding instance), any accepted false session implies either at least one false local instance verifies or a hash binding collision occurs. Therefore
\begin{equation}
\label{supp:eq:supp-session-soundness}
\epsilon_{\mathrm{sess}}
\le
\sum_{m=1}^{M}(K_m+1)\,\epsilon_{\mathrm{loc}}
+
\epsilon_{\mathcal{H}}
\le
M (K_{\max}+1)\epsilon_{\mathrm{loc}}+\epsilon_{\mathcal{H}}.
\end{equation}
The deterministic parser contributes no probabilistic error under correct implementation. Thus a target $\lambda_{\mathrm{sess}}$-bit session soundness level requires the local proof target to include the union-bound loss, e.g., $\lambda_{\mathrm{loc}}\ge \lambda_{\mathrm{sess}}+\lceil\log_2(M(K_{\max}+1))\rceil+s$ for slack $s$. This bound is the session-level statement promoted to Corollary~1 in the main text. Since all blocks of a session share the same instance count here, $K_{\max}=K_m$. Table~\ref{supp:tab:supp_soundness_accounting} tabulates the dominant counted term $MK_m$; adding the $M$ block-binding instances does not change the rounded slack $\lceil\log_2(\cdot)\rceil$ in any reported row, so the tabulated local-soundness targets are unchanged. To keep the accounting conservative, the table uses the full four-limb count $n_q=4$, which upper-bounds the instance total of the implemented $n_q=3$ prototype used for the footprint tables; over-counting instances can only loosen the union bound, so the resulting local-soundness target is an upper bound for the prototype as well. The four-limb parameter sets discussed in the main text give the following representative accounting.

\begin{table*}[!t]
\centering
{
\caption{Representative session-level soundness accounting for four active RNS proof channels ($n_q=4$), $d=8192$, and slack $s=16$ bits.}
\label{supp:tab:supp_soundness_accounting}
\scriptsize
\setlength{\tabcolsep}{4.0pt}
\renewcommand{\arraystretch}{1.04}
\begin{tabular}{lrrrrrrr}
\toprule
\textbf{Profile} & \textbf{$N$} & \textbf{$n$} & \textbf{$d_{\mathrm{ZK}}$} & \textbf{$M$} & \textbf{$K_m$} & \textbf{Counted $MK_m$} & \textbf{$128+\lceil\log_2(M(K_m+1))\rceil+s$} \\
\midrule
FaceNet $n=1$ & 100{,}000 & 1 & 512 & 13 & 16{,}320 & 212{,}160 & 162 \\
FaceNet $n=3$ & 100{,}000 & 3 & 512 & 5 & 16{,}320 & 81{,}600 & 161 \\
MobileFaceNet $n=1$ & 100{,}000 & 1 & 512 & 13 & 65{,}472 & 851{,}136 & 164 \\
MobileFaceNet $n=2$ & 100{,}000 & 2 & 512 & 7 & 65{,}472 & 458{,}304 & 163 \\
FaceNet $n=3$ & 1{,}000{,}000 & 3 & 128 & 41 & 65{,}280 & 2{,}676{,}480 & 166 \\
MobileFaceNet $n=2$ & 1{,}000{,}000 & 2 & 128 & 62 & 261{,}888 & 16{,}237{,}056 & 168 \\
MobileFaceNet $n=1$ & 1{,}000{,}000 & 1 & 128 & 123 & 261{,}888 & 32{,}212{,}224 & 169 \\
\bottomrule
\end{tabular}
}
\end{table*}

\paragraph{Extension path for multi-layer traces.}
\label{supp:supp:pvsc_multilayer}
The concrete descriptor language above is intentionally specialized to the depth-1 SCMV similarity trace. A multi-layer BGV computation would require descriptors of the form $\delta=(\mathsf{op},\ell,\iota,\eta_{\mathrm{in}},\eta_{\mathrm{out}},\mathsf{lev},\mathsf{ek})$, where $\mathsf{op}$ identifies the operation family, $\ell$ is the circuit layer, $\iota$ is the local RNS/NTT/chunk index, $\eta_{\mathrm{in}}$ and $\eta_{\mathrm{out}}$ bind the consumed and produced ciphertext states, $\mathsf{lev}$ records the ciphertext level and modulus-chain position, and $\mathsf{ek}$ identifies the evaluation key when a key-dependent operation is used. Additions and multiplications can reuse the current local templates with layer/state tags. Rotations would add automorphism and key-switch consistency relations; modulus switching would add exact level-transition, scaling, and rounding relations; and bootstrapping would either need to be expanded into its own verified subtrace or represented by a subtrace digest whose internal proof is checked separately. The deterministic parser would have to verify topological order, single-use or allowed reuse of intermediate states, level compatibility, and complete output coverage before invoking the local verifiers.

\paragraph{Proof Footprint and Contract-Side Enforcement Details.}
\label{supp:supp:footprint_contract}

This subsection provides the detailed implementation and footprint rows used in the experimental discussion. For one block, PVSC instantiates $K_m=n_qn_{\mathrm{chunk}}(2L-1)$ counted arithmetic/output proof instances plus one block-binding instance, where $n_{\mathrm{chunk}}=\lceil d/d_{\mathrm{ZK}}\rceil$. Table~\ref{supp:tab:supp_pvsc_block_footprint} reports the per-block proof footprint of the implemented prototype, while Table~\ref{supp:tab:supp_pvsc_session_footprint} aggregates the corresponding opened-proof and communication footprints at the session level. Both tables use the implemented prototype with $n_q=3$ active proof channels; this differs from the conservative $n_q=4$ count used for the session soundness accounting above, which upper-bounds the instance total. Because the counted proof-instance term and the per-block proof-object footprint are linear in $n_q$, a full four-limb ($n_q=4$) deployment footprint scales approximately by the factor $4/3$ from the measured three-channel rows; we report the measured three-channel prototype to avoid presenting extrapolated bytes as measured. The proof-footprint rows are struct-size footprints from generated proof objects whose proof generation and verification passed; session rows multiply the counted arithmetic/output term and the per-block byte quantities by $M=\lceil N/(nd)\rceil$. Because the current backend does not expose a canonical binary serialization API for all proof objects, the byte counts are computed from the generated proof object's field and ring dimensions rather than from a deployment-specific wire format, and should be read as struct-size estimates rather than on-wire measurements.

\begin{table*}[!t]
\centering
{
\caption{Per-block PVSC proof footprint, measured on the implemented three-channel prototype ($n_q=3$); these are measured prototype struct sizes, not four-limb ($n_q=4$) measured bytes. All rows use $d=8192$.}
\label{supp:tab:supp_pvsc_block_footprint}
\scriptsize
\setlength{\tabcolsep}{4.0pt}
\renewcommand{\arraystretch}{1.03}
\begin{tabular}{llrrrrrr}
\toprule
\textbf{Profile} & \textbf{$L$} & \textbf{$d_{\mathrm{ZK}}$} & \textbf{$n_{\mathrm{chunk}}$} & \textbf{$K_m$} & \textbf{Enc./block} & \textbf{Opened/block} & \textbf{Public input/block} \\
\midrule
FaceNet $n=1$ & 128 & 128 & 64 & 48{,}960 & 2.61~GiB & 20.17~MiB & 167.34~MiB \\
FaceNet $n=1$ & 128 & 256 & 32 & 24{,}480 & 1.30~GiB & 10.09~MiB & 167.34~MiB \\
FaceNet $n=1$ & 128 & 512 & 16 & 12{,}240 & 668.07~MiB & 5.04~MiB & 167.34~MiB \\
FaceNet $n=2$ & 128 & 128 & 64 & 48{,}960 & 3.50~GiB & 32.50~MiB & 251.02~MiB \\
FaceNet $n=2$ & 128 & 256 & 32 & 24{,}480 & 1.75~GiB & 16.25~MiB & 251.02~MiB \\
FaceNet $n=2$ & 128 & 512 & 16 & 12{,}240 & 896.67~MiB & 8.12~MiB & 251.02~MiB \\
FaceNet $n=3$ & 128 & 128 & 64 & 48{,}960 & 3.50~GiB & 32.50~MiB & 251.02~MiB \\
FaceNet $n=3$ & 128 & 256 & 32 & 24{,}480 & 1.75~GiB & 16.25~MiB & 251.02~MiB \\
FaceNet $n=3$ & 128 & 512 & 16 & 12{,}240 & 896.67~MiB & 8.12~MiB & 251.02~MiB \\
MobileFaceNet $n=1$ & 512 & 128 & 64 & 196{,}416 & 12.88~GiB & 77.92~MiB & 671.34~MiB \\
MobileFaceNet $n=1$ & 512 & 256 & 32 & 98{,}208 & 6.44~GiB & 38.96~MiB & 671.34~MiB \\
MobileFaceNet $n=1$ & 512 & 512 & 16 & 49{,}104 & 3.22~GiB & 19.48~MiB & 671.34~MiB \\
MobileFaceNet $n=2$ & 512 & 128 & 64 & 196{,}416 & 14.05~GiB & 130.37~MiB & 1007.02~MiB \\
MobileFaceNet $n=2$ & 512 & 256 & 32 & 98{,}208 & 7.03~GiB & 65.19~MiB & 1007.02~MiB \\
MobileFaceNet $n=2$ & 512 & 512 & 16 & 49{,}104 & 3.51~GiB & 32.59~MiB & 1007.02~MiB \\
\bottomrule
\end{tabular}
}
\end{table*}

\begin{table*}[!t]
\centering
{
\caption{Session-level PVSC footprint at $d_{\mathrm{ZK}}=512$ and threshold $t=3$, on the implemented three-channel prototype ($n_q=3$); measured prototype struct sizes, not four-limb measured bytes.}
\label{supp:tab:supp_pvsc_session_footprint}
\scriptsize
\setlength{\tabcolsep}{5.0pt}
\renewcommand{\arraystretch}{1.04}
\begin{tabular}{lrrrrrr}
\toprule
\textbf{Profile} & \textbf{$N$} & \textbf{$n$} & \textbf{$M$} & \textbf{Counted $MK_m$} & \textbf{Opened/session} & \textbf{Opening comm.} \\
\midrule
FaceNet $n=1$ & 100{,}000 & 1 & 13 & 159{,}120 & 65.56~MiB & 18.21~MiB \\
FaceNet $n=1$ & 1{,}000{,}000 & 1 & 123 & 1{,}505{,}520 & 620.26~MiB & 172.29~MiB \\
FaceNet $n=2$ & 100{,}000 & 2 & 7 & 85{,}680 & 56.87~MiB & 9.81~MiB \\
FaceNet $n=2$ & 1{,}000{,}000 & 2 & 62 & 758{,}880 & 503.71~MiB & 86.85~MiB \\
FaceNet $n=3$ & 100{,}000 & 3 & 5 & 61{,}200 & 40.62~MiB & 7.00~MiB \\
FaceNet $n=3$ & 1{,}000{,}000 & 3 & 41 & 501{,}840 & 333.10~MiB & 57.43~MiB \\
MobileFaceNet $n=1$ & 100{,}000 & 1 & 13 & 638{,}352 & 253.25~MiB & 73.05~MiB \\
MobileFaceNet $n=1$ & 1{,}000{,}000 & 1 & 123 & 6{,}039{,}792 & 2.34~GiB & 691.20~MiB \\
MobileFaceNet $n=2$ & 100{,}000 & 2 & 7 & 343{,}728 & 228.15~MiB & 39.34~MiB \\
MobileFaceNet $n=2$ & 1{,}000{,}000 & 2 & 62 & 3{,}044{,}448 & 1.97~GiB & 348.41~MiB \\
\bottomrule
\end{tabular}
}
\end{table*}

We separately measure the smart-contract surface that enforces public session state. Table~\ref{supp:tab:supp_contract_cost} reports the resulting gas consumption, calldata size, storage writes, and execution latency for the measured contract functions. The prototype uses Solidity~0.8.24, Hardhat~3.7.0, optimizer 200 runs, and the Hardhat EDR simulated L1 backend. The contract checks session binding, query/gallery commitments, opened-proof headers, duplicate block indices, coverage, threshold confirmations, finalization, and the opened-QAP verifier kernels used by $\mathsf{PVSC.Verify}$. These measurements cover contract-side enforcement and verifier-kernel execution for opened local proof instances. They do not treat all uncompressed PVSC proof payloads as a gas-optimized full-batch L1 submission; instead, full-batch size and scaling are accounted for in the proof-footprint and session-scaling results. Under the corrected verification placement of the main paper (Section~\ref{sec:PVSC_ver}), $\mathsf{PVSC.Verify}$ runs inside the committee, so the \texttt{verifyQAP} and \texttt{submitOpenedProofWithQAP} rows are reported for reference only; committee attestations have the same form as the measured committee confirmations. At the session level, total enforcement gas scales mainly with $M$: for example, FaceNet with $n=3$ costs 799{,}768 gas at $N=100$k and 4{,}185{,}244 gas at $N=1$M, while MobileFaceNet with $n=2$ costs 987{,}850 gas at $N=100$k and 6{,}160{,}105 gas at $N=1$M.

\begin{table*}[!t]
\centering
{
\caption{Contract-side enforcement cost measured from transaction receipts.}
\label{supp:tab:supp_contract_cost}
\scriptsize
\setlength{\tabcolsep}{6.0pt}
\renewcommand{\arraystretch}{1.05}
\begin{tabular}{lrrrr}
\toprule
\textbf{Function} & \textbf{Gas} & \textbf{Calldata bytes} & \textbf{Storage writes} & \textbf{Latency (ms)} \\
\midrule
\texttt{createSession} & 114{,}799 & 164 & 6 & 1.039 \\
\texttt{submitOpenedProofHeader} & 54{,}037 & 164 & 2 & 0.852 \\
\texttt{submitCommitteeConfirmation:first} & 74{,}225 & 68 & 2 & 1.387 \\
\texttt{submitCommitteeConfirmation:additional} & 54{,}445 & 68 & 2 & 0.679 \\
\texttt{finalizeSession} & 31{,}649 & 68 & 1 & 0.450 \\
\texttt{verifyQAP:q14} & 50{,}980 & 2{,}596 & 0 & 0.696 \\
\texttt{verifyQAP:q15} & 52{,}820 & 2{,}756 & 0 & 0.662 \\
\texttt{verifyQAP:q21} & 63{,}830 & 3{,}716 & 0 & 0.838 \\
\texttt{submitOpenedProofWithQAP:q21} & 94{,}041 & 3{,}908 & 2 & 1.631 \\
\bottomrule
\end{tabular}
}
\end{table*}

\section{Indicative Cross-Paper Comparison of the Encrypted-Matching Layer}
\label{supp:supp:comparison}

Table~\ref{supp:tab:comparison} presents an indicative cross-paper comparison of the encrypted-matching layer against representative privacy-preserving biometric identification schemes.

\begin{table*}[!t]
	\centering
	\caption{Indicative cross-paper comparison of the encrypted-matching layer with representative privacy-preserving biometric identification schemes. Reported values are taken from the corresponding papers and are not strictly apples-to-apples because datasets, hardware platforms, and cryptographic settings differ.}
	\label{supp:tab:comparison}
	\begin{tabular}{@{}l cc cc c c@{}}
		\toprule
		\multirow{2}{*}{\textbf{Scheme}} & \multicolumn{2}{c}{\textbf{Search Latency (s)}} & \multicolumn{2}{c}{\textbf{Enrollment Latency (s)}} & \multirow{2}{*}{\textbf{Storage}} & \multirow{2}{*}{\textbf{Primitive}} \\
		\cmidrule(lr){2-3} \cmidrule(lr){4-5}
		& Single-core & Multi-core & Single-core & Multi-core & & \\ \midrule
		Bauspiess \textit{et al.} & 1000.0 & 130.0 & 1001.0 & 130.1 & 10.0 GB & BGV \\
		CryptoMask       & 25.0   & 3.3   & 27.0   & 3.4   & 5.1 GB  & BGV \\
		MFBR-ID          & 5.2    & 0.6   & 13132.0& 1644.0& 132.1 GB& BFV \\
		IDFace           & \textbf{2.9}    & \textbf{0.3}   & 880.0  & 109.0 & 4.125 GB& CKKS \\
		\textbf{SCMV}         & 8.2    & 1.3   & \textbf{9.2}    & \textbf{1.4}   & \textbf{1.7 GB}  & BGV \\ \bottomrule
	\end{tabular}
\end{table*}

The reported values are not strictly apples-to-apples because the compared systems use different datasets, hardware platforms, cryptographic parameterizations, and implementation optimizations. For this reason, the main paper does not treat cross-paper latency numbers as a primary quantitative claim. Representative points of reference include BGV-based packed-identification systems such as Bauspie{\ss} \textit{et al.} and CryptoMask~\cite{supp:packing2,supp:CryptoMask}, as well as more recent HE-based identification systems such as MFBR-ID and IDFace~\cite{supp:bassit2025practical,supp:ICCV}.

Within this limitation, SCMV provides strong storage efficiency and a balanced trade-off among search cost, enrollment cost, and encrypted database footprint. In particular, compared with prior BGV-based systems such as Bauspie{\ss} \textit{et al.} and CryptoMask, SCMV reduces both latency and storage overhead; compared with large-footprint BFV-based baselines, it avoids a very large encrypted database footprint. We therefore view SCMV as especially attractive when both encrypted storage efficiency and dynamic database maintenance matter.

\input{supplementary.bbl}

\endgroup